\documentclass[11pt]{article}

\usepackage[T1]{fontenc}
\usepackage[a4paper,margin=1.05in]{geometry}
\usepackage{lmodern}
\usepackage{microtype}
\usepackage{amsmath,amsthm,amssymb,mathtools,bbm}
\usepackage{dsfont}
\usepackage{hyperref}
\usepackage{subfig}
\usepackage{microtype}
\usepackage{booktabs} 
\usepackage[backend=biber,style=nature, backref=true]{biblatex}
\hypersetup{colorlinks=true,linkcolor=blue,citecolor=blue,urlcolor=blue}
\allowdisplaybreaks
\usepackage[utf8]{inputenc}
\usepackage[english]{babel}
\usepackage{xcolor}
\usepackage{parskip}
\usepackage{appendix}
\usepackage{tikz-cd}
\usetikzlibrary{arrows.meta,decorations.pathreplacing,backgrounds,calc,positioning}
\usepackage{braket}
\usepackage{mathrsfs}
\usepackage{graphicx}
\usepackage{listings}
\usepackage{multirow}
\usepackage{mhchem}
\usepackage[export]{adjustbox}
\usepackage{spalign}
\usepackage{cancel}
\usepackage{verbatim}
\usepackage{abstract}
\usepackage{authblk}
\usepackage{bm}

\numberwithin{equation}{section}
\numberwithin{figure}{section}
\numberwithin{table}{section}

\theoremstyle{definition}
\newtheorem{mydef}{Definition}[section]

\theoremstyle{plain}
\newtheorem{thm}{Theorem}[section]
\newtheorem{lem}[thm]{Lemma}
\newtheorem{prop}[thm]{Proposition}

\newtheorem{cor}[thm]{Corollary}

\theoremstyle{remark}
\newtheorem{obs}[thm]{Remark}

\newcommand{\w}{\omega}

\newcommand{\Z}{\mathbb{Z}}
\newcommand{\p}{\rho}

\newcommand{\CC}{\mathbb{C}}
\newcommand{\Hil}{\mathcal{H}}
\newcommand{\Id}{\mathbb{I}}

\newcommand{\RR}{\mathcal{R}}
\newcommand{\Li}{\mathcal{L}}
\newcommand{\LL}{\widetilde{\Li}}
\newcommand{\A}{\mathcal{A}}
\newcommand{\vA}{\bm{A}}

\newcommand{\recovery}{\mathcal{R}_{B\to AB}^t}

\newcommand{\La}{\Lambda}
\newcommand{\G}{\Gamma}
\newcommand{\T}{T^\Lambda}

\newcommand{\TT}{\widetilde{T}^\Lambda}
\newcommand{\g}{g_{\Lambda}}
\newcommand{\norm}[1]{\left\lVert#1\right\rVert}
\newcommand{\normm}[2]{\left\lVert#1\right\rVert_{#2}}

\newcommand{\summ}[1]{\sum\limits_{\substack{ #1 }}}

\newcommand{\tr}{\operatorname{tr}}

\newcommand{\supp}{\operatorname{supp}}
\newcommand{\poly}{\mathrm{poly}}

\newcommand{\abs}[1]{\left \vert#1\right\vert}

\definecolor{ms}{rgb}{0,.4,1}

\begin{document}

\title{Static features from mixing in short- and long-range Lindbladians: Markov property and correlations}



\author[1]{\href{https://orcid.org/0009-0000-9555-5604}{Paul Rosa-Ruiz}\thanks{paul.rosa@ift.csic.es}}
\author[1]{\href{https://orcid.org/0000-0002-6395-3971}{Matteo Scandi}\thanks{matteo.scandi@csic.es}}
\author[2,3]{\href{https://orcid.org/0000-0001-6713-6760}{\'{A}ngela Capel}\thanks{ac2722@cam.ac.uk}}
\author[1]{\href{https://orcid.org/0000-0002-5889-4022}{\'{A}lvaro M. Alhambra}\thanks{alvaro.alhambra@csic.es}}

\affil[1]{Instituto de Física Téorica UAM/CSIC, C. Nicolás Cabrera 13-15, Cantoblanco, 28049 Madrid, Spain}
\affil[2]{Department of Applied Mathematics and Theoretical Physics, University of Cambridge,
Wilberforce Road, Cambridge CB3 0WA, United Kingdom}
\affil[3]{Fachbereich Mathematik, Universität Tübingen, 72076 Tübingen, Germany}
\date{\today}
\maketitle

\begin{abstract}
The classification of mixed-state phases requires criteria beyond two-point correlation functions, such as the decay of the mutual information (MI) and the conditional mutual information (CMI), with the latter encapsulated in the notion of \textit{Markov length}. Here we show how such static properties of the fixed point of a Lindbladian follow from natural dynamical features of its generators: rapid mixing and frustration-freeness. We focus on systems with long-range interactions, and prove \emph{i)} that local Lindbladians satisfying (global) rapid mixing and frustration-freeness have fixed-points whose CMI decays with the shielding distance, and \emph{ii)} that (local) rapid mixing together with primitivity and regularity implies global decay of MI. For long-range interactions decaying with a power law with rate $\alpha$, both quantities decay polynomially rather than exponentially, in contrast to the finite- and short-range regimes where exponential decay (a finite Markov length) is expected within a phase. We further show that Gibbs states of long-range, non-commuting Hamiltonians satisfy a local Markov property at any temperatures, extending the recent results (Chen--Rouzé, 2025) for short-range systems to the long-range regime relevant to a variety of experimental platforms. As a numerical example, we study the long-range Ising model both with and without a transverse field. We find regimes in which the polynomial decay of the CMI holds, in accordance with the bounds proven. 
\end{abstract}

\tableofcontents

\section{Introduction}

Quantum many-body systems display a wealth of different phases, often going beyond the well-established paradigms of phase characterization such as the classic one based on symmetry breaking. Thus, the classification of quantum phases is still an ongoing effort, and constitutes one of the central areas of study at the intersection of condensed matter, quantum information and mathematical physics. For pure ground states of closed systems, a well-established definition of what it means to be within a phase is provided by adiabatic continuity of two points in the phase through Hamiltonians with a finite positive spectral gap \cite{hastings_spectral_2006, NacSim06}. 
As such, the gap is crucial in phase characterization, and its presence has direct consequences for the ground state: a spectral gap implies exponential decay of two-point correlations \cite{hastings_spectral_2006, NacSim06}, and also an area law of entanglement in many important scenarios \cite{Cramer_2006,Has07,arad2013arealawsubexponentialalgorithm,Anshu_2022}. 

Many interesting quantum systems, both for physics and in the context of quantum computation, are coupled to their environment, and as such divert from any theoretical framework tailored to pure states.
Extending the classification of phases to those dissipative settings and mixed states has motivated many recent efforts, spanning the characterization of mixed-state topological order and decoherence-induced transitions~\cite{Bao2023,Lee_2023,Fan+24,Chen_2024,sang2026mixedstatequantumphasesrenormalization,wang2025fractionalquantumhallstates}, the study of paradigmatic phase transitions such as strong-to-weak spontaneous symmetry breaking \cite{Bu_a_2012,Lieu_2020,de_Groot_2022,Sala_2024,Ma_2025,Lessa_2025,wang2026strongtoweakspontaneoussymmetrybreaking}, and their relation to the construction of topological quantum memories \cite{Dennis2002,Sang_Renormalization_2024}. It also includes recent experimental efforts \cite{wang2026observationstrongtoweakspontaneoussymmetry,zhang2025probingmixedstatephasesquantum}. 
A characterization of what it means to be ``within a phase'' has also been proposed in terms of fast dissipative evolutions~\cite{coser_classification_2019,rakovszky_defining_2024}, and stability under local perturbations~\cite{rakovszky_defining_2024, sang_mixed-state_2025, li_unified_2026}.  

It is by now well established that the behaviour of correlation functions is insufficient to provide a characterization of mixed-state phases, since it is unable to detect, as a prime example, the aforementioned strong-to-weak spontaneous symmetry breaking \cite{Lee_2023,Sala_2024,Ma_SPT_2025,Lessa_2025,shu2026universaldynamicalscalingstrongtoweak}. Thus, one needs to consider further criteria, and one of the most prominent are information-theoretic quantities such as the mutual information (MI) $I(A:C)_\rho$ and the conditional mutual information (CMI) $I(A:C|B)_\rho$, which are the main focus of the present work\footnote{Other objects not studied here are the ``fidelity correlators'' \cite{Lessa_2025,Weinstein_2025,Liu2025,Zhang_2026}, which can explicitly detect the strong-to-weak spontaneous symmetry breaking.}. 

The MI of two disjointly supported sets $A$ and $C$ on a fixed-point $\rho$ is defined as
\begin{equation}
    I(A:C)_\rho = S(\rho_A)+S(\rho_C)-S(\rho_{AC}),
\end{equation}
where $S(\sigma):=-\tr(\sigma \log(\sigma))$ is the von Neumann entropy and $\rho_A:=\tr_{A^C}(\rho)$ is the reduced state on the system $A$. It captures the amount of correlations between two separate subsystems \cite{Groisman_2005}, and has a wealth of interesting features in many-body mixed states \cite{Melko_2010,Scalet_2021,Bluhm_2022,yi2025universaldecayconditionalmutual}, including an area law \cite{wolf_area_2008}. 
In data hiding, there are states with arbitrarily small operator correlation and large MI \cite{Hastings2007b, Hayden2004}.

On the other hand, given a tripartition of the total space $A\sqcup B\sqcup C=\La$, the quantum CMI between $A$ and $C$ conditioned by $B$ on the state $\rho$ is defined as
\begin{equation}
    I(A:C|B)_\rho:=S(\rho_{AB})+S(\rho_{BC})-S(\rho_B) - S(\rho_{ABC}).
\end{equation}

The CMI measures the conditional independence of a state between two systems conditioned to a third one, and states with small CMI can be locally reconstructed from their marginals; in other words, the CMI is related to how well the global state can be locally reconstructed from its marginals \cite{Sutter2018,FawRen15}, and its fast decay with distance has been identified as an essential criterion for phase stability making it a crucial quantity in mixed-state phase characterization~\cite{Sang_2025, MKS25, yi2025universaldecayconditionalmutual, chen_quantum_2025, li_unified_2026, chen2026localreversibilitydivergentmarkov}.
The main reason for this is that a decaying CMI allows the recovery of the state after local noise is applied to $A$: if information in $A$ is lost, one can faithfully reconstruct the global state by applying a local recovery map on $AB$.

 When studying the behaviour of information-theoretic quantities such as MI or CMI in the many-body setting, two aspects matter for the classification of mixed-state phases \cite{yi2025universaldecayconditionalmutual,ma2025circuitbasedcharacterizationfinitetemperaturequantum}: (i) how it decays with the shielding distance $d(A,C)$, and (ii) how it scales with the sizes $|A|$ and $|C|$. 
Concerning the dependence on $\abs{A}$ and $\abs{C}$, one can distinguish three levels of CMI decay or \emph{Markov property}:
\begin{itemize}
    \item \textit{Pairwise Markov property}: $I(A:C|B)_\rho \lesssim \exp(\abs{AC})f(d(A,C))$,
    \item \textit{Local Markov property}: $I(A:C|B)_\rho \lesssim p(\abs{A})\exp(\abs{C})f(d(A,C))$,
    \item \textit{Global Markov property}: $I(A:C|B)_\rho \lesssim p(\abs{A},\abs{C})f(d(A,C))$,
\end{itemize}
where $f\colon\mathbb{N}\to\mathbb{R}_{\geq 0}$ is a decaying function and $p$ is a polynomial function. Depending on the Markov level, one can consider a tri-partition with bounded or unbounded $A$ and $C$. 

This distinction has recently been explored in the context of Gibbs states, where a pairwise Markov property has been shown for rather general conditions in  \cite{Kuw24,chen_quantum_2025,bergamaschi2025structuraltheoryquantummetastability}, while at restricted settings such as high temperatures \cite{kato_clustering_2025,bakshi_dobrushin_2025}, 1D \cite{Kato2019,Kuw24}, under correlation decay \cite{chen2026notestrongquantummarkov} and under the existence of an exponentially-decaying effective Hamiltonian \cite{scalet2025classicalestimationfreeenergy} the stronger global property holds. A similar distinction between pairwise, local and global can also be established for the MI, in terms of a decaying function $f$ and the prefactors in terms of $\vert A\vert, \vert C \vert$. For Gibbs states of finite-range Hamiltonians, the decay of the MI is shown to be global in 1D \cite{Bluhm_2022} and pairwise in any D, under the existence of an exponentially-decaying effective Hamiltonian \cite{BCP2024}. Although the literature in decay of MI is scarce, it is expected that in many situations, the MI of Gibbs states will decay in the same way as the correlation functions. The latter is known to be global in 1D at any temperature for finite-range Hamiltonians \cite{Ara69}, above a certain temperature for short-range interactions \cite{Alhambra_2026}, and in any D at high-enough temperature for finite-range interactions \cite{Kliesch2014} and pairwise for short-range interactions \cite{Ueltschi2004}. In all these results, the decay is exponential. The 1D case is extended to any temperature for long-range Hamiltonians with subexponential decay \cite{kimura_clustering_2025}. Other systems that present pairwise correlation decay are for instance weakly interacting quantum systems \cite{adhikari2026uniformintemperaturelocalityestimatesweakly}.

To understand why the dependence with respect to $\abs{A}$ and $\abs{C}$ is important note that, on one hand, an exponential growth with $\abs{A}$ and $\abs{C}$ implies that $A$ and $C$ must be bounded (of order $\mathcal{O}(1)$) and not grow with the system size $\abs{\La}$ in order to have a non-trivial bound. On the other hand, a polynomial dependence allows $A$ and $C$ to grow with the system size $\abs{\La}$. One should say specifically how the distance between $A$ and $C$ scales then, and how the decay with the shielding distance typically is, so that the bound is overall decaying.

Since these are static properties of individual states, it is important to understand under what dynamical setting we expect them to arise. 
Mixed states are generated through a dissipative process with an external, possibly uncontrolled, environment, so their open system dynamics is described by \emph{Lindbladians} generators of completely positive trace-preserving semigroups~\cite{breuer_theory_2002, manzano_short_2020}, with specific mixed states as fixed points. Understanding the interplay between dynamical properties of the Lindbladian and the static properties of their fixed points is thus essential for the characterization of mixed-state phases,
in a similar way as the gap (a property of the Hamiltonian) implies non-trivial features of its (static) ground states.

In this work we show how static properties of the fixed point (decay of CMI and MI) follow from natural dynamical features of their generating Lindbladians. Specifically, we consider the consequences of \textit{rapid mixing} (Definitions~\ref{def:globalRM}, \ref{def:localRM}) and \textit{frustration-freeness} (Definition~\ref{def:frustrationFreeness}), with a special focus in the setting of long-range interactions. We show that the MI decays in a global sense under the assumptions of locality and rapid mixing, and that the CMI also decays similarly under the additional assumption of frustration-freeness. While we show that our results apply to fixed-points of local Lindbladians, they do not directly apply to Gibbs states of non-commuting Hamiltonians since they are fixed points of a \emph{quasi-local} (KMS) Lindbladian. However, we are able to prove a local Markov property on these Gibbs states in the long range regime extending the proof of \cite{chen_quantum_2025}. Such long-range interactions are directly realised in current (noisy) experimental platforms, including Rydberg atom arrays \cite{browaeys_many-body_2020}, trapped-ion quantum simulators \cite{foss-feig_progress_2025}, and magnetic dipolar quantum gases \cite{chomaz_dipolar_2023}.

For short-range interacting systems, it is expected that both the MI and the CMI within a phase decay exponentially with the shielding distance. For the CMI, this is formalised through the notion of \textit{Markov length} $\xi$ ~\cite{Sang_2025, MKS25}, which measures the decay rate of the exponential function $f(d) \propto e^{- {d}/{\xi}}$ and plays the role of a correlation length for multipartite correlations; a finite Markov length has been established as a crucial stability condition for mixed-state phases~\cite{Sang_2025,yi2025universaldecayconditionalmutual}. Our results illustrate how the picture changes when interactions decay algebraically in long-range systems. For these systems, rapid thermalization implies that the fixed point exhibits a \emph{polynomial} decay of CMI and MI, illustrating how a classification of mixed-state phases must accommodate this qualitatively different regime. Put together, we explore the conditions under which dissipative systems do not generate any long-range order in the sense that is relevant for the classification of mixed-state phases \cite{Sang_2025,yi2025universaldecayconditionalmutual,li_unified_2026}, that is, with a fast decay of both the MI and the CMI.

These results provide static properties of fixed points of Lindbladians that are implied by dynamical properties of the Lindbladians. The converse direction has also been studied in the literature \cite{Bardet_2023,Bardet_2024,kochanowski_rapid_2025,Capel2024Gibbs,kastoryano_gibbs_2016,art:2localPaper,BardetCapelLuciaPerezGarciaRouze-HeatBath1DMLSI-2019,bergamaschi2026fastmixingquantumspin}, especially for Gibbs states, by using the existence of a positive modified logarithmic Sobolev inequality/spectral gap in the Lindbladian as an intermediate step to prove rapid/fast mixing. In particular, \cite{Capel2024Gibbs} introduced a condition termed \textit{matrix valued quantum conditional mutual information (MCMI)} whose exponential decay implies, jointly with the assumption of a positive local gap in the generator, rapid mixing of the Lindbladian. The MCMI is an operator version of the CMI, which in particular implies decay of CMI and MI for specific subsystems. This suggests that a possible converse for the main results of the current manuscript might be feasible. 

Finally, we simulate numerically the CMI of the Gibbs state of the classical long-range Ising model in a spin chain. We find it has a polynomial decay of the CMI, providing a concrete example of the results from Section \ref{sec:thermal}. We also see that adding a transverse field and increasing the decay rate of the long-range interactions help the CMI to decay exponentially, since it makes the system behave akin to a short-range system (for which an exponential decay of the CMI is expected \cite{chen_quantum_2025}). The details of the numerical simulations are in Section \ref{sec:numerics}.

\subsection{Summary of Main results}\label{sec:mainResults}

First of all, consider a general dissipative process generated by a local Lindbladian \[\Li=\summ{Z\subset \La}L_Z,\quad \normm{L_Z}{cb}\lesssim F(\text{diam}(Z)),\]for an $F$-function $F$ (see Definition \ref{def:Ffunction}), and such that $\Li$ satisfies global rapid mixing (see Definition \ref{def:globalRM}). The $F$-function $F$ is a monotonically decreasing function used to control the decay of the local terms composing the Lindbladian. Consider a tripartition $A\sqcup B \sqcup C=\La $ with $R:=d(A,C) >0$ as in Figure \ref{fig:geometryCMI}. Our first main result states the following: 

\begin{figure}
    \centering
    \includegraphics[width=0.5\linewidth]{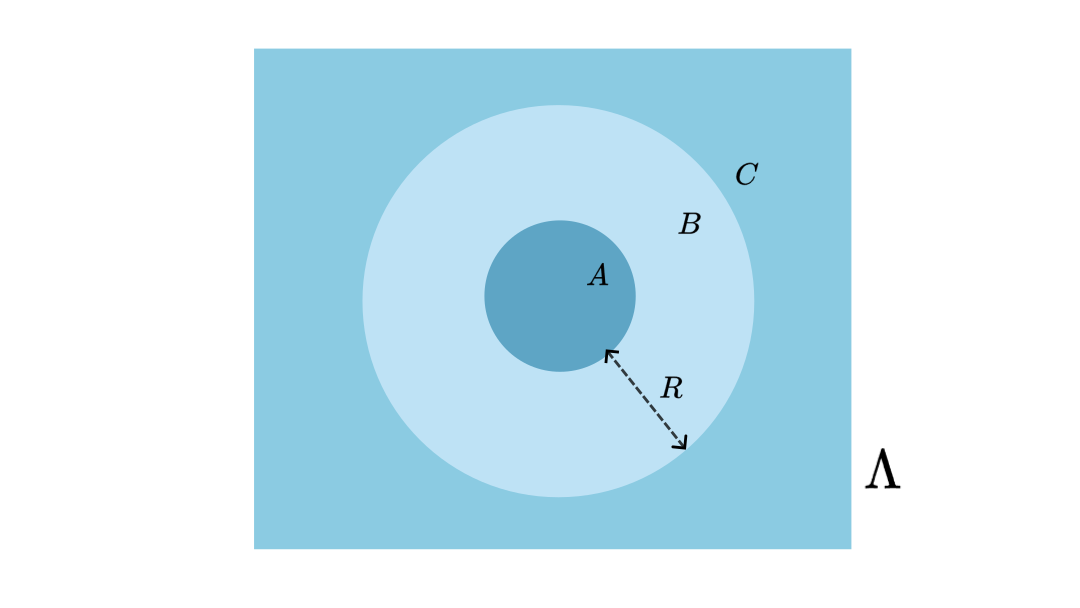}
    \caption{Geometry of the tripartition $A\sqcup B\sqcup C=\La$.}
    \label{fig:geometryCMI}
\end{figure}

\medskip
\begin{thm}[Informal Theorem \ref{thm:decayCMIgeneralLindbladian}]\label{thm:informaldecayCMIgeneralLindbladian}
    The fixed-point of a rapid mixing and frustration-free Lindbladian has a global Markov property
    \begin{equation}
            I(A : C \,|\, B)_\rho \le \, p(\abs{A},\abs{C})\, R^{\alpha_1} (F(R))^{\alpha_2} ,
    \end{equation}
    where $A\sqcup B\sqcup C=\La $ are as in Figure \ref{fig:geometryCMI}, $p$ is a polynomial,  and $\alpha_1,\alpha_2$ are positive constants.
\end{thm}

In particular, for long-range Lindbladians (i.e. with $F_\alpha(x)=1/(1+x)^\alpha$ for $\alpha >D$), we obtain the following consequence.
\medskip
\begin{cor}[Informal Corollary \ref{cor:decayCMILongRange}]\label{thm:informaldecayCMIgeneralLindbladianLongRange}
    Fixed-points of long-range Lindbladians also satisfy the global Markov property for a sufficiently high interaction decay rate $\alpha >\alpha_0$ 
    \begin{equation}
        I(A:C|B)_\rho\leq  \, p(\abs{A},\abs{C})\, F_{\eta}(R) \, ,
    \end{equation}
    with a decay rate $\eta $ and a polynomial $p$.
\end{cor}

Secondly, if a Lindbladian satisfies local rapid mixing and primitivity (See Definitions  \ref{def:localRM} and \ref{def:primitivity}), then its fixed point shows a decay of MI. 
\medskip
\begin{thm}[Informal Theorem \ref{thm:MutualInfoDecay}]\label{thm:informalMutualInfoDecay}
    The fixed-point of a local, primitive and rapid mixing Lindbladian\footnote{No frustration-freeness needed.} satisfies a global decay of mutual information
        \begin{equation}
        I(A:C)_\rho\leq k(\abs{A},\abs{C}) \cdot F_{\alpha '}(d(A,C)),
    \end{equation}
     where $k$ is a polynomial and $\alpha'$ is the decay rate. 
\end{thm}

Finally, for the Gibbs state $\rho_\beta$ of a (non-commuting) long-range Hamiltonian $H$ we have the following\footnote{The local Markov property requires a polynomial dependence on $\abs{C}$ and an exponential dependence on $\abs{A}$; see \cite{chen_quantum_2025} for the short-range analogue.}.
\medskip
\begin{thm}[Informal Theorem \ref{thm:cmiCKGLindbladian}]\label{thm:informalCMICKGLindbladian}
    The Gibbs state of a non-commuting long-range Hamiltonian satisfies the local Markov property for every inverse temperature $\beta $
\begin{equation}
    I(A : C \,|\, B)_\rho \le \log(\dim(C)) \,e^{\mu_1 \abs{A}} \dfrac{(\log (R))^{\mu_2}}{R^{\mu_3}} =\tilde{\mathcal{O}}(\abs{C}e^{\mu_1 \abs{A}}R^{-\mu_3}),
\end{equation}
for some positive constants $\mu_1,\mu_2,\mu_3$.
\end{thm}

All the main results are summarized in Table \ref{tab:hypothesesMainResults} with all the explicit constants, including the hypotheses for each theorem, the conditions on the decay rates and the bounds obtained. Note that in all these results the decay for the local terms of the Lindbladian and the decays obtained for the CMI or MI are in general different.

\begin{table}[t]
\centering
\setlength{\tabcolsep}{5pt}
\renewcommand{\arraystretch}{1.5}
\footnotesize
\begin{tabular}{|p{1.9cm}|p{4.2cm}|p{3.6cm}|p{4.6cm}|}
\hline
\textbf{Result} & \textbf{Hypotheses} & \textbf{Condition on $\alpha$ $\Rightarrow$ decay rate} & \textbf{$I(A:C|B)_\rho$} \\
\hline
\hline
Theorem~\ref{thm:informaldecayCMIgeneralLindbladian} \newline
\emph{(CMI, general Lindbladian)}
&
$k$-local, primitive, frustration-free Lindbladian, with decay by $F$-function $F$; global rapid mixing $\mathrm{RM}(n,\lambda)$.
&
General $F$-function $F$;\newline Decay set by the exponent $R^{nD/2}F(R)^{\lambda/2(\lambda+v)}$.
&
$p(\abs{A},\abs{C})\, R^{nD/2} \big(F(R)\big)^{\lambda/[2(\lambda+v)]}$
\newline
\emph{Global} Markov property
\\
\hline
Corollary~\ref{thm:informaldecayCMIgeneralLindbladianLongRange} \newline
\emph{(CMI, long-range Lindbladian)}
&
Same as above, specialised to $F_\alpha(x) = (1+x)^{-\alpha}$ (long-range Lindbladian); global rapid mixing $\mathrm{RM}(n,\lambda)$.
&
$\alpha > nD(\lambda+v)/\lambda$ \  $\Rightarrow$ \newline $ \eta=\dfrac{\lambda \alpha}{2(\lambda+v)} - \dfrac{nD}{2} > 0$
&
$ p(\abs{A},\abs{C})\, F_\eta(R)$
\newline 
\emph{Global} Markov property
\\
\hline
Theorem~\ref{thm:informalMutualInfoDecay} \newline
\emph{(decay of MI)}
&
Long-range, local, primitive, regular Lindbladian (\emph{frustration-freeness not needed}); local rapid mixing $\mathrm{LRM}(\lambda)$.
&
$\alpha > 3D - 1$ \ $\Rightarrow$ \newline $\alpha' = \dfrac{\lambda(\alpha - 2D)}{\lambda + v} > 0$
&
$ k(\abs{A},\abs{C})\, F_{\alpha'}(R-1)$

\emph{Global} decay of the MI
\\
\hline
Theorem \ref{thm:informalCMICKGLindbladian} \newline
\emph{(local Markov property)}
&
Gibbs state of a $k$-local long-range \emph{non-commuting} Hamiltonian.
&
$\alpha > D$ on $H$\newline $\Rightarrow$ \ $\alpha' = \dfrac{\lambda'(\alpha - D)}{2} > 0$
&
$I(A:C|B)_\rho = \tilde{\mathcal{O}}(\abs{C}e^{\mu \abs{A}}R^{-\alpha'})$

\emph{Local} Markov property
\\
\hline
\end{tabular}
\caption{Hypotheses, conditions on $\alpha$ and decay rates of the four main
results (Section~\ref{sec:mainResults}). Here $R = d(A,C)$,
$D$ is the lattice dimension, $v = \normm{\Li}{F}\, C_F$ the Lieb-Robinson
velocity, and $p,k$ polynomial functions. The conditions on $\alpha$ in
rows~2 and~3 are the ones actually required by the proofs.}
\label{tab:hypothesesMainResults}
\end{table}

Numerical results in Section \ref{sec:numerics} show that the Gibbs state of the classical long-range Ising model exhibits a polynomial decay of the CMI. However, for high enough decay rates $\alpha$, the system increasingly resembles a short-range system, and the CMI of the Gibbs state is better approximated by an exponential decay than by a polynomial one. Additionaly, adding a transverse field also seems to contribute to the shift from a polynomial to an exponential decay of the CMI. Figure \ref{fig:cmiDecayCurves} shows the CMI decay for different values of $\alpha$ and $h$. We leave for future work the study of this qualitative change of the CMI.

\subsection{Comparison with prior work}\label{sec:comparisonPriorWork}

\textbf{Correlations in long-range open systems.} Roon--Sims \cite{roon_quasi-locality_2024} proved  decay of correlations for long-range dissipative processes in a very similar setting, bounding $|\pi(O_A O_C)-\pi(O_A)\pi(O_C)|$ for the state $\pi$ and for disjointly supported observables $O_A,O_C$. Our results extend theirs in two ways: we consider the decay of the CMI under additional assumptions, and also our proof technique for the MI allows us to obtain a \emph{global} decay and not a pairwise decay, as defined in the introduction (see Section~\ref{sec:MI} for a more detailed discussion). 
In \cite{kastoryano_rapid_2013}, a similar result is proven but limited to finite-range and detailed-balanced Lindbladians. Another closely related result is the proof of the area law of the fixed points frustration-free Lindbladians in \cite{brandao_area_2015}, under the same assumptions as our CMI decay (in fact, our proof relies on a direct adaptation of their main technical lemma).

\textbf{CMI decay for short-range systems.} Significant advances have recently been made on the decay of CMI for quantum Gibbs states. In particular, \cite{chen_quantum_2025} proves the pairwise Markov property for Gibbs states of finite-range Hamiltonians at all temperatures. Our results extend this to the long-range setting. Also for finite-range Hamiltonians, \cite{kato_clustering_2025} establishes exponential decay of CMI for high temperatures when the subsystems $A$ and $C$ are at most constant, via quantum belief-propagation channels. This decay was also established in \cite{scalet2025classicalestimationfreeenergy} using the decay of the so-called \emph{effective} or \emph{mean-force} Hamiltonian.  Finally, \cite{bakshi_dobrushin_2025} improved the restrictions on the size of $A$ and $C$, proving the global Markov property in the high temperature regime. More recently, \cite{chen2026notestrongquantummarkov} showed the equivalence between the global and the pairwise Markov property under assumptions on clustering of correlations. All of these results are specific to Gibbs states of finite-range systems. Here, we extend that line of work in two complementary ways: by targetting long-range systems, and by proving a CMI decay for the fixed points of general local Lindbladians. While we do not prove it here, we expect the global Markov property akin to \cite{bakshi_dobrushin_2025} to hold at high-enough temperatures in long-range systems. See Table \ref{tab:priorWorkCMI} for a summary of the previous results on the decay of CMI for different regimes.

\textbf{Characterization of mixed phases.} Recent efforts have considered the role of both the CMI and the MI in the characterization of phases in the context of mixed states. In particular, \cite{Sang_2025} showed that a finite Markov length guarantees the existence of finite-depth channel circuits mapping between the states within a phase  (a condition termed local reversibility \cite{sang_mixed-state_2025}), then \cite{yi2025universaldecayconditionalmutual} proved that the decay of MI and CMI are robust properties of a phase, by showing their stability under shallow channel circuits. This equivalence was also recently refined by the notion of local stability (see also \cite{cubitt_stability_2015}), and its relation to the MI and CMI was also recently established in \cite{li_unified_2026}, where they consider the same notions of correlation decay that we study. These definitions are based on channel circuits mapping across phases, with the connection to thermal phase diagrams established in \cite{ma2025circuitbasedcharacterizationfinitetemperaturequantum}. These works do not directly target the properties of systems described by Lindbladians, as we do here. With our results, we narrow down the conditions under which the fixed points of Lindbladians belong to the mixed-state analogue of ``gapped phases''.

\begin{table}[t!]
\centering
\setlength{\tabcolsep}{4pt}
\renewcommand{\arraystretch}{1.25}
\footnotesize
\begin{tabular}{|p{2.2cm}|p{3.2cm}|p{1.2cm}|p{4.5cm}|p{3cm}|}
\hline
\textbf{Reference} & \textbf{State and Regime} & \textbf{Markov property} &\textbf{Decay rate of $I(A:C|B)$}& \textbf{Comments} \\
\hline
\hline
Kuwahara~\cite{Kuw24}
&
Gibbs state of short-range Hamiltonian,\newline arbitrary temperatures;\newline D>1;\newline
&
Pairwise
&
$\mathcal{O}(e^{\abs{AC}}$\newline $e^{-c_1 R/(\beta^{D+1}\log(R))+c_2 \log(\beta \abs{AC})})$
&
$\abs{A},\abs{C}=\mathcal{O}(1)$
\\
&
1D,
&
Global
&
$\mathcal{O}(e^{-c_3 R/\beta+c_4 \beta \log(\beta R)})$
&
Dimensional restriction (1D)
\\
\hline
Kato--Kuwahara ~\cite{kato_clustering_2025}
&
Gibbs state of finite-range Hamiltonian, \newline high temperature
&
Global
&
$\mathcal{O}((\abs{A}+\abs{C})$\newline$d(A,C)^{2D+1}e^{-c \,d(A,C)^{1/(D+3)}})$
&
via Belief Propagation channels
\\
\hline
Scalet et al.~\cite{scalet2025classicalestimationfreeenergy}
&
Gibbs state of short range Hamiltonian,\newline under existence of an effective Hamiltonian
&
Global
&
$\mathcal{O}((\abs{A}+\abs{C})e^{-\Omega (d(A,C))})$
&
At the moment, hypothesis proven for high temperatures for arbitrary dimensions
\\

\hline
Kato--Brandão~\cite{Kato2019}; 
&
1D, Gibbs state of finite-range Hamiltonian
&
Global
&
$\mathcal{O}(e^{-\Theta (\sqrt{d(A,C)})})$
&
Dimensional restriction (1D)
\\
\hline
Bakshi et al. ~\cite{bakshi_dobrushin_2025}
&
Gibbs state of finite range Hamiltonian, \newline high temperature
&
Global
&
$\mathcal{O}(\abs{A},\abs{C})e^{-d(A,C)/\xi}$
&
Improves the restrictions on $\abs{A}$ and $\abs{C}$. Dobroushin condition.
\\
\hline
Chen--Rouzé \cite{chen_quantum_2025}
&
Gibbs state of short-range Hamiltonians,\newline \textit{all} temperatures
&
Local
&
$\mathcal{O}(\abs{A}\abs{C} e^{\abs{A}-d(A,c)/\xi})$
&
First proof for all temperatures
\\
\hline
Theorem \ref{thm:informaldecayCMIgeneralLindbladian}
&
Fixed-point of local, rapid mixing, frustration-free Lindbladians decaying with $F$.
&
Global
&
$\mathcal{O}((\abs{A},\abs{C})\newline d(A,C)^{nD/2} F(d(A,C))^{\lambda/2(\lambda+v)})$
&
Frustration-freeness $+$ global rapid mixing
\\
\hline
Corollary \ref{thm:informaldecayCMIgeneralLindbladianLongRange}
&
Fixed point of local, rapid mixing, frustration-free, long-range Lindbladians
&
Global
&
$\mathcal{O}((\abs{A},\abs{C}) \cdot F_{\alpha '}(d(A,C)-1))$\newline
     where $\alpha'=\dfrac{\lambda(\alpha-2D)}{\lambda +v}$,
&
Frustration-freeness $+$ global rapid mixing
\\
\hline
Theorem \ref{thm:cmiCKGLindbladian}
&
Gibbs state of long-range non-commuting Hamiltonian at \textit{all} temperatures
&
Local
&
$\tilde{\mathcal{O}}(R^{-\alpha'})$\newline
    for $\alpha':=\lambda ' (\alpha-D)/2$
&
Extends \cite{chen_quantum_2025} to the long-range regime
\\
\hline
\end{tabular}
\caption{Comparison with prior results on the Markov property of the CMI for
fixed points / Gibbs states (Section \ref{sec:comparisonPriorWork}). $D$ is the dimension and $c,c_1,c_2,c_3,c_4,\xi$ are constants, $\beta$ is the inverse temperature, $\lambda$ depends on the long-range decay rate $\alpha$ and $v$ is the Lieb-Robinson velocity.}
\label{tab:priorWorkCMI}
\end{table}


\section{Preliminaries}
Throughout this paper, $(\G,d)$ denotes a countable metric whose elements (called sites) represent the spatial positions of the individual particles. Each point $x\in\G$ carries a finite-dimensional local Hilbert space $\Hil_x$ of (local) dimension $D_x$.

Even though our results hold for any countable metric space with any local dimensions, the primary example which is considered in Section \ref{sec:thermal} is the hypercubic lattice $\G=\mathbb{Z}^D$ of dimension $D\in \mathbb{N}$ with the same two-dimensional ($D_x=2$, $\forall x\in\mathbb{Z}^D$) Hilbert space in every point $\Hil_x=\Hil\cong \mathbb{C}^2$ representing the spatial position of the individual particles (qubits\footnote{There is no restriction with considering higher local dimensions such as qutrits ($D_x=3$) and so on.}). The distance used is the $\ell^1$-metric $d(x,y)=\sum_{j=1}^D\abs{x_j-y_j}$. However, the precise structure of $\mathbb{Z}^D$ is not necessary for most of the results discussed here.

If $\La\subset\subset \G$ then $\La$ is a finite subset of $\G$. Furthermore, we define $\mathcal{P}_0(\G)$ as the collection of all finite subsets of $\G$.

\noindent\textbf{Geometric constants.} We say $(\Gamma, d)$ is 
$(\kappa, \nu)$\textbf{-regular} if
\begin{equation}
    \sup_{x \in \Gamma} \left|\{y \in \Gamma : d(x,y) \leq r\}\right| 
    \leq \kappa\, r^\nu 
    \qquad \text{for all } r \geq 1,
    \label{eq:regularity}
\end{equation}
for constants $\kappa < \infty$ and $\nu > 0$. 

In particular, for $\Gamma = \mathbb{Z}^D$ 
with the $\ell^1$-metric, this holds with $\nu=D-1$ and \begin{equation}\kappa = \kappa_D = \frac{D\pi^{D/2}}{\Gamma(D/2+1)}.\label{eq:geometricK}\end{equation}
\medskip

\noindent\textbf{Inflation and distance.} For $X \subset \Gamma$ and 
$r > 0$, the $r$\textbf{-inflation} of $X$ is
\begin{equation}
    X(r) := \{y \in \Gamma : d(X, y) \leq r\}, 
    \qquad d(X,y) := \inf_{x \in X} d(x,y).
\end{equation}
We write $d_y(X) := \sup_{x \in X} d(y,x)$ for the farthest distance 
from $y$ to $X$, and $\mathrm{diam}(X) := \sup_{x,y \in X} d(x,y)$ 
for the diameter of $X$, i.e. the farthest distance between two sites in $X$. 

\subsection{Hilbert Space Structure}\label{sec:hilbertSpaceStructure}

 For every point $x\in\Gamma$ we choose a finite-dimensional Hilbert space $\Hil_x\cong \CC^{n_x}$ of dimension $n_x>0$; from now on, we assume that every point has a Hilbert space of the same local dimension $n$, i.e. $\Hil_x\cong \CC^n,\quad \forall x\in\Gamma$. 
 
 The corresponding algebra of observables of each point $x$ is the set of bounded linear operators $\A_x =\mathcal{B}(\Hil_x)$ from $\Hil_x$ to itself. Then, for every subset $\Lambda\subset \Gamma$, the Hilbert space of the states on $\Lambda$ is defined by taking the tensor product\[\Hil_\La=\bigotimes_{x\in\La}\Hil_x,\]with the corresponding algebra of observables\[\A_\La=\bigotimes_{x\in\La}\A_x\cong \mathcal{B}(\Hil_\La).\]

For every subset $\Lambda'\subset\La\subset\G$, we get a partial order by identifying $\A_{\La'}$ with $\A_{\La'}\otimes\Id_{\La \setminus\La'}$ and write $\A_{\La'}\subset \A_{\La}$. The algebra of all strictly local observables $\A_\G^{\text{loc}}$ is then defined as the inductive limit
\begin{equation}
\A_\G^{\text{loc}}=\bigcup_{\La\in\mathcal{P}_0(\G)}\A_\La.
\end{equation}
The $C^*$-algebra of quasi-local observables $\A_\G$ is taken to be the norm completion of $\A_\G^{\text{loc}}$.

\medskip

\noindent\textbf{States.} A state on $\mathcal{A}_\Lambda$ is a positive 
linear functional $\omega: \mathcal{A}_\Lambda \to \mathbb{C}$ with 
$\omega(\Id) = 1$ uniquely identified with a density matrix $\rho \in \mathcal{A}_\Lambda$, $\rho \geq 0$, via $\omega(O) = \mathrm{Tr}(\rho O)$ (see Gleason's Theorem \cite{gleason});
by a standard abuse of notation we use the term \emph{state} for both the functional $\omega$ and its associated density matrix $\rho$. We write $\mathcal{S}(\mathcal{A}_\Lambda):=\{\w\in \A_\La^*:\w \geq 0, \w(\Id)=1\}$ for the set of states on $\mathcal{A}_\Lambda$. Note that if we are working in a finite-dimensional algebra $\A_\La$ the predual $\A_*=\A^*$ is the same as the dual space, then, we will denote a state as an element of the dual $\rho\in\A_\La^*$, however, in case of considering infinite dimensional algebras, the state actually is part of the predual space.

\medskip

\noindent\textbf{Partial traces.} For $X \subset \Lambda$, the reduced state on $X$ is $\rho_X := \mathrm{Tr}_{\Lambda \setminus X}(\rho)$. 
We write
\begin{equation}
    \rho_{-X} := \frac{\Id_X}{n^{|X|}} \otimes \rho_{\Lambda \setminus X}
    \label{eq:partialTrace}
\end{equation}
for the state obtained by replacing the $X$-marginal of $\rho$ with the maximally mixed state on $\mathcal{H}_X$.


\subsection{Interactions and $F$-Functions}

The (local) interactions in open quantum systems will consist of a family $\Li=\{L_Z:Z\subset \subset \G\}$ of bounded linear operators $L_Z:\A_Z\to\A_Z$ in the so-called Lindblad form, which characterises generators of norm-continuous, one parameter semigroup of unital completely positive maps ~\cite{breuer_theory_2002, manzano_short_2020}. 
Given a finite subset $\La\subset\subset \G$, we always use $^*$ to denote the adjoint of a linear transformation on $\A_\La$ with respect to the Hilbert-Schmidt inner product $\langle \cdot,\cdot \rangle_{HS}$. Whenever we treat with finite dimensional algebra $\A_\La$.

It is well-known \cite{Gorini:1975nb,lindblad_generators_1976} that for each finite $Z\subset \La$, there exists a self-adjoint $H_Z=H_Z^*\in \A_Z$,  a number $l_Z\in\mathbb{N}$, and jump operators $K_1,\dots,K_{l_Z}\in\A_Z$ such that, for all $O\in\A_Z$ and $\rho\in\mathcal{S}(\A_\La)$,
\begin{equation}
    L_Z(O)=i[H_Z,O]+\summ{j=1}^{l_Z}K_j^*OK_j-\frac 12 \{K_j^*K_j,O\}
    \end{equation}
\begin{equation}
    L_Z^*(\rho)=-i[H_Z,\rho]+\summ{j=1}^{l_Z}K_j\rho K_j^*-\frac 12 \{K_j^* K_j,\rho\}.
\end{equation}

Given a finite volume $\La\subset\subset\G$, we call a \textit{local interaction} the operator
\begin{equation}
    \Li_\La=\summ{Z\subset \La}L_Z.
    \label{eq:localInteraction}
\end{equation}
As an abuse of notation, we sometimes write $\Li_\La=\{L_Z\}_{Z\subset \La}$ to denote the local interaction $\Li_\La=\summ{Z\subset \La}L_Z$, and we call $\Li_\La$ a dissipative process, as this can be thought as the generator of the time evolution of an open quantum system. We may remove the subscript of $\Li_\La$ and simply write $\Li$ if the set $\La$ can be understood from the context. 

The local operator $\Li_\La$, due to its Lindblad form, generates a norm-continous, one-parameter semigroup of operators denoted by
\begin{equation}
    \T_t (O_A)=e^{t\Li_\La}(O_A),\quad O_A\in\A_\La,
\end{equation}
which is well-defined since $\Li_\La$ is bounded. We will simply denote by $T_t:=\T_t$ $\forall t$, if $\La$ is clear from context.

Complete positivity is a general condition that some maps between $C^*$-algebras satisfy. For any finite dimensional Hilbert space $\Hil$, any linear map $\Phi :\mathcal{B}(\Hil)\to\mathcal{B}(\Hil)$ induces a linear map on $\mathcal{B}(\Hil)\otimes \mathbb{M}_n$, where $\mathbb{M}_n$ denotes the set of $n\times n$ matrices, by acting block-wise
\begin{equation}
    \Phi\otimes \Id_{n\times n}([[X_{i,j}]]_{i,j=1}^n):= [[\Phi(X_{i,j})]]_{i,j=1}^n.
\end{equation}
We say that the linear map $\Phi$ is completely positive if for every $n\in\mathbb{N}$, the matrix $\Phi\otimes\Id_{n\times n}(X^*X)$ is positive semidefinite for all $X\in\mathcal{B}(\Hil)\otimes \mathbb{M}_n$. The completely positive maps are also completely bounded, i.e.
\begin{equation}
    \normm{\Phi}{cb}:=\sup_{n\in\mathbb{N}}\norm{\Phi\otimes \Id_{n\times n}}<\infty.
\end{equation}
We define a class of completely bounded maps as the following:
\medskip
\begin{mydef}
    Given a finite volume $Y$, the set $\mathcal{CB}_0(Y)$ is defined as
    \begin{equation}
        \mathcal{CB}_0(Y)=\{\Phi:\A_Y\to\A_Y:\normm{\Phi}{cb}<\infty,\,\Phi(\Id_Y)=0 \}.
    \end{equation}
\end{mydef}
An example of an element in $\mathcal{CB}_0(Y)$ is the commutator $\Phi(\cdot)=[\cdot,B]$ for $B\in \A_Y$.

The decay of interactions with the distance is encoded in an $F$-function.
\medskip
\begin{mydef}[$F$-function]\label{def:Ffunction}
    A non-increasing function $F:[0,\infty)\to (0,\infty )$ is said to be a $F$-function on $(\Gamma,d)$ if:
\begin{enumerate}
    \item $F$ is \textit{uniformly summable}, in the sense that
    \begin{equation}
        \lvert\lvert\lvert F\rvert \rvert \rvert:=\sup_{x\in\Gamma}\sum_{y\in \Gamma}F(d(x,y))<\infty,
        \label{eq:normF}
    \end{equation}
    and
    \item $F$ satisfies a \textit{convolution condition}:
    \begin{equation}
        C_F:=\sup_{x,y\in \G}\sum_{z\in\G}\dfrac{F(d(x,z))F(d(z,y))}{F(d(x,y))}<\infty
        \label{eq:convolutionCF}
    \end{equation}
\end{enumerate}
\end{mydef}
Note that an immediate consequence of \eqref{eq:convolutionCF} is that for any pair $x,y\in\G$, we have the bound
\begin{equation}
   \sum_{z\in\G} F(d(x,z))F(d(z,y))\leq C_FF(d(x,y))
\end{equation}

The two canonical families of $F$-functions used throughout 
this paper are:
\begin{itemize}
    \item \textbf{Short-range:} $F_{a,b}(r) = e^{-ar^b}$ for $a > 0$, 
    $b \in (0,1]$.
    \item \textbf{Long-range:} $F_\alpha(r) = (1+r)^{-\alpha}$ for 
    $\alpha > D$, on a $(\kappa,D)$-regular space.
\end{itemize}
For $F_\alpha$ on $\mathbb{Z}^D$, one has\[\|F_\alpha\| \leq \frac{2\kappa_D}{\alpha - D},\] where we use that the metric space $(\Z^D,d)$ is $(\kappa_D,D-1)$-regular according to Equation \eqref{eq:regularity}:
\begin{equation}\label{eq:SR}
\left|\{x:\ d(x,S)\in(r-1,r]\}\right|\ \le\ \kappa_D \,|\partial S|\,(1+r)^{D-1}\qquad (\text{finite }S\subset\Lambda,~r\in\mathbb N),
\end{equation}
 for $\kappa_D$ in Equation \eqref{eq:geometricK} and 
$C_{F_\alpha} \leq c(\alpha, D) < \infty$ for any $\alpha > D$.

\medskip

\begin{mydef}
    Let $(\Gamma,d)$ be a countable metric space equipped with an F-function $F$. Let $\Li=\{L_Z\}_{Z\in \mathcal{P}_0(\Gamma)}$ be a local interaction. We say that $\Li$ is a \emph{doubly anchored} interaction governed by $F$ if
    \begin{equation}\label{eq:defFNormLindbladian}
        \normm{\Li}{F}:=\sup_{x,y\in \La}\sum\limits_{\substack{Z\subset \La \\x,y\in Z}}\dfrac{\normm{L_Z}{cb}}{F(d(x,y))}<\infty
    \end{equation}
\end{mydef}

Except otherwise stated, we will assume that the dissipative system is doubly-anchored to a function $F$, unless it is not clear by the context. Note that a natural consequence of \eqref{eq:defFNormLindbladian} is that
    \begin{equation}\label{eq:FNormLindbladian}
    \sum\limits_{\substack{Z\subset \La \\x,y\in Z}}\normm{L_Z}{cb}\leq \normm{\Li}{F} F(d(x,y)) 
\end{equation}
for every $x,y \in \Lambda$.

In a similar way, a closed quantum system with a local Hamiltonian $H=\sum_{Z\subset\Lambda} h_Z$, has long-range interactions if
\begin{equation}\label{eq:defFNormHamiltonian}
    \normm{H}{F}:= \sup_{x,y\in \La}\sum\limits_{\substack{Z\subset \La \\x,y\in Z}}\dfrac{\normm{h_Z}{\infty}}{F_\alpha(d(x,y))}<\infty
\end{equation}
for the $F$-function\[F_\alpha(r)=\dfrac{1}{(1+r)^\alpha};\quad r\geq 0,\]
for $\alpha>D$. In other words, there exists a $J>0$ such that
\begin{equation}\label{eq:powerlaw}
\norm{h_Z}\ \le\ \frac{J}{(1+\mathrm{diam}(Z))^{\alpha}} 
\end{equation}

for every $Z \subset \Lambda$. Unless otherwise stated, all the Hamiltonians we consider will be $k$-local, i.e., for a $k\in\mathbb{N}$,
\begin{equation}\label{eq:Hklocal}
H=\sum_{Z\subset\Lambda,\ |Z|\le k} h_Z.
\end{equation}

\subsection{Fixed Points and Mixing}

Following the notation in \cite{roon_quasi-locality_2024}, we now define fixed points and mixing conditions. Let $(\Gamma,d)$ be a countable metric space and consider a finite subset $\La\subset\subset \G$.

\medskip
\begin{mydef}\label{def:dynamicalFixedPoint}
    Let $\Li=\{L_Z\}_{Z\in \mathcal{P}_0(\Gamma)}$ be a dissipative interaction, $\pi\in\mathcal{S}(\mathcal{A}_\La^*)$ a state on the local algebra $\mathcal{A}_\Lambda$ and denote by $\T_t$ the time evolution generated by $\Li$. We say that $\pi$ is a \emph{dynamical fixed point} of $\Li$ if
    \begin{equation}
        \pi \circ \T_t=\pi\quad \forall t\geq 0.
    \end{equation}
\end{mydef}

\begin{obs}
    Note that, from the above definition, $\w \circ \T_\infty$ is a dynamical fixed point for every $\omega\in\mathcal{S}(\mathcal{A}_\La^*)$, i.e. \[\w \circ \T_\infty\circ \T_t=\w \circ \T_\infty \, ,\]where $\T_\infty:=\lim_{t\to\infty}\T_t$.
\end{obs}

This does not specify how fast an arbitrary state converges towards the fixed point. In order to get some information about the convergence speed, we define the following. 
\medskip
\begin{mydef}
    We say that $\pi\in\mathcal{S}(\mathcal{A}_\La^*)$ is a \emph{local dynamical fixed point governed by $\g$ with respect to $A\subset\La$} if it is a dynamical fixed point and there exists a non-increasing function $\g:[0,\infty)\to [0,2]$ for which one has 
    \begin{equation}
        \abs{(\pi-\w)\circ \T_t(O_A)}\leq\norm{O_A}\cdot p(\abs{A})\cdot  \g(t)
        \label{eq:localDynamicalFP}
    \end{equation}
    for all states $\w\in\mathcal{S}(\mathcal{A}_\La^*)$, $p$ a polynomial function,  $O_A\in\mathcal{A}_A
    $, $t\geq0$, and, from now on, we denote by $\norm{\cdot}:=\normm{\cdot}{\infty}$.
\end{mydef}

However, in some cases, instead of treating with positive functionals $\w\in\mathcal{S}(\A_\La)$, we use the density matrix $\rho\in\A_\La$ corresponding to $\w$ by the correspondence (see Gleason's Theorem \cite{gleason})\[\w(O)=\tr(\rho O),\quad \forall O\in\A_\La.\]See Section \ref{sec:hilbertSpaceStructure} for more details.

When considering density matrices, the common definition of fixed-point is that it is contained in the kernel of the Lindbladian. 
\medskip
\begin{mydef}[Fixed-point]\label{def:fixed-point}
    A state $\rho\in\A_\La$ is a fixed-point of a Lindbladian $\Li$ if
    \begin{equation}
        \Li^*(\rho)=0.
    \end{equation}
\end{mydef}

Both Definitions \ref{def:dynamicalFixedPoint} and \ref{def:fixed-point} are equivalent:
\medskip
\begin{lem}
    Let $\w\in\mathcal{S}(\A_\La)$ be a state and let $\rho\in\A_\La$ be the density matrix corresponding to $\w$. Then, $\w$ is a dynamical fixed-point if and only if $\rho$ is a fixed-point.
\end{lem}
\begin{proof}
    $\Rightarrow$ : Suppose $\w$ is a dynamical fixed-point (i.e. $\w\,\circ \, \T_t=\w$). Then, for any $O\in\A_\La$,
    \begin{equation}
        \tr((\T_t)^*(\rho)O)=\tr( \rho \T_t(O))=\w(\T_t(O))=\w(O)=\tr(\rho O).
    \end{equation}
    By the non-degeneracy of the Hilbert-Schmidt inner product, $(\T_t)^*(\rho)=\rho$ for all $t$. Differentiating at $t=0$ yields
    \begin{equation}
        \Li^*(\rho)=0.
    \end{equation}
    Hence $\rho$ is a fixed-point of the Lindbladian.

    $\Leftarrow$ : Suppose $\rho$ is a fixed-point of $\Li$ ($\Li^*(\rho)=0$). Then, for any $O\in\A_\La$,
    \begin{equation}
        (\w\,\circ\,\T_t)(O)=\w(\T_t(O))=\tr(\rho \T_t(O))=\tr((\T_t)^*(\rho)O)=\tr(\rho O)=\w(O).
    \end{equation}
    Then, $\w\,\circ\,\T_t=\w$.
\end{proof}
The Lemma above shows that we can indistinctively  talk about fixed-points or dynamical fixed-points, only depending on whether our state is a functional or a density-matrix.

Even though in most of the cases we study local dynamical and static properties, there is a stronger notion of global speed of convergence called (global) rapid mixing. I is introduced and used more often in the literature (see \cite{brandao_area_2015,cubitt_stability_2015}) and defined in terms of the contraction of the quantum Markov semigroup.
\medskip
\begin{mydef}[Contraction]\label{def:contraction}
    Let $\Li=\{L_Z\}_{Z\in \mathcal{P}_0(\Gamma)}$ be a dissipative interaction. Define the \textit{contraction of $\T_t$} as
    \begin{equation}
        \eta(\T_t):=\sup_{O\in \A_\La,\norm{O}=1} \norm{(\T_t-\T_\infty)(O)}.
    \end{equation}
    Take $A\subset \Lambda$ and define the \textit{contraction of $\T_t$ relative to region $A$} as
    \begin{equation}
        \eta^A(\T_t):=\sup_{\w \in \A_\La^*}\sup_{O_A\in \A_A,\norm{O_A}=1} \abs{\w \circ (\T_t-\T_\infty)(O_A)}
    \end{equation}
\end{mydef}
\begin{mydef}[Global rapid mixing]\label{def:globalRM}
    We say that $\Li$ satisfies  \textit{global rapid mixing} if 
    \begin{equation}
        \eta(\T_t)\leq p(\abs{\Lambda})e^{-\lambda t},
    \end{equation}
    where $\lambda >0$ and $p$ a $n$-degree polynomial. We denote it as $\text{RM}(n,\lambda)$, or by $\text{RM}(\lambda)$ if the degree of the polynomial is not relevant as long as it is finite ($n<\infty$).
\end{mydef}

This definition considers any operator defined in $\La$ and as such it provides a trivial bound in the thermodynamical limit ($\abs{\Lambda}\to\infty$). 
However, a weaker condition than global rapid mixing can be required on $\Li$ if we only consider operators defined in a region $X\subset \La$.
\medskip
\begin{mydef}[Local rapid mixing]\label{def:localRM}
    We say that $\Li$ satisfies \textit{local rapid mixing} if, for each $X\subset \Lambda$,
    \begin{equation}
        \eta^X(\T_t)\leq p(\abs{X}) e^{-\lambda t}
    \end{equation}
    where $\lambda >0$ and $p$ a $n$-degree polynomial. We denote it as $\text{LRM}(n,\lambda)$, or by $\text{LRM}(\lambda)$ if the degree of the polynomial is not relevant as long as it is finite ($n<\infty$).
    \medskip
\end{mydef}
\medskip
\begin{obs}
    Note that for every two sets $A\subset B$, both subsets of $\Lambda$, then $\eta^A\leq \eta ^B$. In particular, $\eta ^A\leq \eta$ for all $A\subset \Lambda$.
\end{obs}
\medskip

\begin{obs}
    Note that if $\Li=\{L_Z\}_{Z\in \mathcal{P}_0(\Gamma)}$ is primitive (see Definition \ref{def:primitivity}) and satisfies local rapid mixing, then, its fixed point is a local dynamical fixed point governed by $e^{-\lambda t}$; this is the way we treat fixed points for locally rapidly mixing dissipative processes. 
\end{obs}

Local rapid mixing can be thought as the following: consider a state in a finite region $A\subset \Lambda$, then this state will converge into its local Gibbs distribution in a time that grows logarithmically with the size of the region $A$.

Table \ref{tab:rapid_mixing_literature} is a summary of what it is known so far on rapid mixing, including the assumptions and the main results of each of the references.

\begin{table}[h!]
\centering
\small
\renewcommand{\arraystretch}{1.25}
\setlength{\tabcolsep}{4pt}
\begin{tabular}{@{}|p{2.8cm}|p{3.8cm}|p{4.0cm}|p{2.2cm}|@{}}
\toprule
\textbf{Reference} & \textbf{Assumption} & \textbf{Main result} & \textbf{Consequence} \\
\midrule

Kastoryano \& Temme (2013) \cite{kastoryano_quantum_2013}
  & General primitive Lindbladian
  & MLSI constant $\alpha>0$
    $\Rightarrow$ mixing time $O(\log n)$
  & MLSI as route to RM \\
\hline

Kastoryano \& Eisert (2013) \cite{kastoryano_rapid_2013}
  & Local, primitive Lindbladian
  & Gap $\Rightarrow$ exp.\ decay of correlations;
    MLSI $\Rightarrow$ area law for MI (log corr.)
  & Structural consequences of RM \\
\hline

Lucia et al.\ (2015) \cite{lucia_rapid_2015}
  & Local, primitive, unique fixed point
  & RM $\Rightarrow$ stability of local observables
    under local perturbations of $\mathcal{L}$
  & Robustness of RM systems \\
\hline

Brand\~ao et al.\ (2015) \cite{brandao_area_2015}
  & Local; fixed point pure \emph{or} FF
  & RM $\Rightarrow$ area law for MI
    (log.\ correction)
  & Area law for mixed-state fixed points \\
\hline

Kastoryano \& Brand\~ao (2016) \cite{kastoryano_gibbs_2016}
  & Davies generator, comm.\ $H$;
    \textbf{1D}: all $T$;
    \textbf{any D}: high $T$
  & Gap $\Leftrightarrow$ (strong) clustering
    $\Rightarrow$ $O( n)$ mixing (fast-mixing)
  & Gap--clustering equivalence;
   Polinomial preparation of\ Gibbs state of comm. $H$. \\

\hline

 Bardet et al. (2023) \cite{Bardet_2023}, Bardet et al. (2024) \cite{Bardet_2024}
& Finite-range, translation-invariant commuting Hamiltonians in $1$D, at any $T$
&
Rapid mixing of the Davis generator with the Gibbs state as its fixed-point
&
Efficient preparation of Gibbs state of commuting Hamiltonian in $1$D
\\
\hline

Kochanowski et al.\ (2025) \cite{kochanowski_rapid_2025}
  & Davies generator, $2$-local comm.\ $H$;
    \textbf{1D}: all $T$;
    \textbf{any D}: high $T$
  & SLI
    $\Rightarrow$ MLSI $\Rightarrow$ RM;
  & Unification of correlation measures
    under RM \\

\hline

Rouzé, Fran\c{c}a \& Alhambra (2024) \cite{rouze_optimal_2024}
  & CKG Lindbladian;
    non-commuting $H$ (short- and long-range):
    high $T$
  & CKG Lindbladian $O(\log n)$ mixing
  & First rigorous RM
    for Gibbs state of a non-comm.\ $H$ \\

\hline

Bakshi et al.\ (2025) \cite{bakshi_dobrushin_2025}
  & Local Gibbs sampler; high $T$
  & Quantum Dobrushin condition
    $\Rightarrow$ RM $+$ exp.\ CMI decay
  & CMI decay from RM. \\




\bottomrule
\end{tabular}
\caption{Known results on rapid mixing (RM) for quantum Lindbladians.
  \emph{High~T}: sufficiently high temperature;
  \emph{All~T}: any finite $T>0$;
  comm./non-comm.: (non-)commuting $H$;
  FF: frustration-free;
  RM: mixing time $O(\log n)$;
  MLSI: modified log-Sobolev inequality;
  SLI: strong local indistinguishability (max-rel-entropy clustering);
  CKG: Chen--Kastoryano--Gilyén Gibbs sampler \cite{chen_efficient_2023};
  CMI/MI: conditional mutual information / mutual information.}
\label{tab:rapid_mixing_literature}
\end{table}


\subsection{Frustration freeness and primitivity}

Frustration freeness is a property of some dissipative processes analogous to the concept in ground states of closed systems. Some fixed points of local evolutions may not be locally steady; in other words, each local dissipative term may have a nonvanishing contribution to the fixed point, and it is the sum of all the local contributions which leaves the fixed point invariant. Frustration freeness excludes such cases.
\medskip

\begin{mydef}\label{def:frustrationFreeness}
    Let $\Li=\{L_Z\}_{Z\subset \La}$ be a local dissipative interaction. We say $\Li$ is frustration free if, for every fixed point $\sigma\in \mathcal{S}(\A_\La)$ of $\Li$ and every $Z\subset\La$, then
    \begin{equation}
        L_Z^*(\sigma)=0.
    \end{equation}
\end{mydef}

As pointed in \cite{brandao_area_2015}, while not in general, many of the most interesting cases of dissipative processes are frustration free:
\begin{itemize}
    \item many of the most common dissipative state engineering procedures \cite{Diehl_2008,verstraete_quantum_2009,Kashyap_2025,baruah2026dissipativepreparationinjectivetensor},
    \item the parent Lindbladians of the RG fixed points of the set of matrix-product density operators \cite{Liu_2026}
    \item detailed balanced quantum Markov processes, in particular, the Glauber dynamics \cite{kastoryano_rapid_2013, kochanowski_rapid_2025} and the KMS detailed balance Lindbladians such as \cite{chen_efficient_2023,Ding_2025,Scandi_2026}, which are Gibbs samplers for commuting and non-commuting Hamiltonians respectively (the latter Lindbladians are quasi-local, not strictly local),
    \item typical noise channels such as depolarizing noise.
\end{itemize}

Another important property of dissipative processes is \emph{primitivity} (see \cite{Sanz_2010}).
\medskip

\begin{mydef}\label{def:primitivity}
A Lindbladian $\Li$ is \textit{primitive} if it has a unique full-rank fixed-point.
\end{mydef}

When it comes to perturbing the Lindbladian, it is important to assure that after removing terms of a local Lindbladian, it will remain being primitive and have a unique fixed-point. That property is called regularity.
\medskip
\begin{mydef}[Regular]\label{def:regular}
    A local primitive Lindbladian $\Li=\{L_Z\}_{Z\subset \La}$ is \textit{regular} if every Lindbladian\[\hat\Li \subset \Li,\]is also primitive. By $\hat\Li\subset \Li$ we mean that $\hat\Li$ is a sum of terms of $\Li=\{L_Z\}_{Z\subset \La}$. 
\end{mydef}
Although this assumption might seem
odd, it turns out to be very natural when discussing scaling of a dissipative system with the
system size. In particular, primitive translationally invariant systems satisfy this property.

\subsection{Information-Theoretic Quantities}
Given two states $\rho,\sigma$ from the Hilbert space $\Hil$, the following quantities are defined
\begin{align}
        S(\rho) &:= -\mathrm{Tr}(\rho\log\rho) 
    && \text{(von Neumann entropy)}, \\ \nonumber
    S(\rho\|\sigma) &:= \mathrm{Tr}(\rho\log\rho) - \mathrm{Tr}(\rho\log\sigma) 
    && \text{(quantum relative entropy)}, \\ \nonumber 
    \text{Fid}(\rho,\sigma) &:= \|\sqrt{\rho}\sqrt{\sigma}\|_1 
    && \text{(quantum fidelity)},
\end{align}
always computed with the natural logarithm. We assume the convention that $-s \log(0)=\infty$ for all $s>0$, so\[S(\rho \|\sigma)=\infty,\text{ if }\supp(\rho)\cap \ker(\sigma)\neq \{0\}.\]

All of these quantities are 
non-negative. The trace distance $\|\rho-\sigma\|_1 := \mathrm{Tr}|\rho-\sigma|$ 
and fidelity are related by the Fuchs--van de Graaf inequalities:
\begin{equation}
    1 - \text{Fid}(\rho,\sigma) \leq \tfrac{1}{2}\|\rho-\sigma\|_1 
    \leq \sqrt{1-\text{Fid}(\rho,\sigma)^2},
\end{equation}
and the quantum Pinsker inequality gives 
$\|\rho-\sigma\|_1 \leq \sqrt{2\,S(\rho\|\sigma)}$.

For any tripartition\footnote{See Figure \ref{fig:geometryCMI}.} $A\sqcup B \sqcup C=\La$, if there exists a quantum channel, which we call \textit{recovery map} $\mathcal{R}_{B\to AB}$, acting on $AB$ that can recover the reduced state $\rho_{BC}$ of a Gibbs state $\rho$, then, the Gibbs state satisfies a Markov property. This relation is key for the proofs in Sections \ref{sec:CMI} and \ref{sec:thermal}. Originally stated in \cite[Equation 10]{FawRen15} and a refined version obtained in \cite[Equation 35]{Bluhm_2023}, the existence of a recovery map implies a bound in the CMI:

\medskip 
\begin{prop}[Fawzi--Renner bound {\cite[Equation 35]{Bluhm_2023}}]
\label{prop:FR}
For any tripartite state $\rho_{ABC}$ and a quantum channel 
$\mathcal{R}_{B\to AB}$ acting on $AB$:
\begin{equation}
    I(A:C|B)_\rho \leq (\sqrt{2}\min\{\abs{A},\abs{C}\}+1) \cdot\normm{\mathcal{R}_{B\to AB} (\rho_{-A})-\rho_{ABC}}{1}^{1/2}
    \label{eq:FR}
\end{equation}
\end{prop}
\medskip
Finally, the following proposition is a reverse Pinsker inequality.

\medskip
\begin{prop}[Hiai--Ohya--Tsukada bound {\cite{HOT81}}]
\label{prop:HOT}
For states $\rho, \sigma$ on $\mathcal{H}$ with $\sigma > 0$:
\begin{equation}
    S(\rho\|\sigma) \leq \log\|\sigma^{-1}\|_\infty \cdot \|\rho - \sigma\|_1.
    \label{eq:HOT}
\end{equation}
\end{prop}

\subsection{Lieb-Robinson Bounds}

Lieb--Robinson bounds (LR bounds in short), introduced in the seminal work~\cite{LieRob72}, establish that quantum information propagates at a finite effective speed in lattice systems with local interactions: for observables $O_A\in \mathcal{A}_A$ and $O_C \in \mathcal{A}_C$ at distance $R = d(A,C)$, the commutator $\|[\T_t(O_A), O_C]\|$ is exponentially (polynomially for long-range systems) suppressed for times $|t| \ll R/v$, where the Lieb--Robinson velocity $v$ is set by the interaction strength.   
This quasi-local structure is a key property to prove rigorously phenomena in quantum many-body theory; see~\cite{ChenLucasYin23} for a comprehensive review. 
When interactions decay as a power law $\|h_Z\| \lesssim \mathrm{diam}(Z)^{-\alpha}$ the first polynomial bounds were given in~\cite{hastings_spectral_2006}, subsequently improved in \cite{MatKomNak17} and \cite{else_improved_2020}, while the work \cite{kuwahara_strictly_2020} established strictly linear light cones for $\alpha > 2D+1$, and for $\alpha > 2D$ the best current bound is \cite{Tran_2021}. 

For open dissipative systems, Lieb--Robinson bounds controlling information propagation under Lindbladian dynamics were established by \cite{Poulin_2010} and \cite{nachtergaele_lieb-robinson_2011} and further developed in the quasi-local setting by ~\cite{barthel_quasilocality_2012}; the extension to dissipative systems with long-range interactions is due to ~\cite{Sweke2019}.

In this paper we use two specific versions: the open-system bound of Theorem~\ref{thm:LR_open} (after~\cite{nachtergaele_lieb-robinson_2011,hastings_spectral_2006}) for the CMI decay proof in Sec. \ref{sec:CMI}, and the long-range Hamiltonian bound of Theorem~\ref{thm:LR_hamiltonian} (after~\cite{kuwahara_strictly_2020}) for the thermal state results in Sec. \ref{sec:thermal}. These are stated as follows:
\medskip

\begin{thm}[LR bound for open systems 
{\cite[Theorem~2]{nachtergaele_lieb-robinson_2011}}]
\label{thm:LR_open}
Let $\mathcal{L} = \{\mathcal{L}_Z\}$ be a doubly-anchored dissipative 
interaction on $(\Gamma,d)$ with F-function $F$ and $\|\mathcal{L}\|_F < \infty$. 
Set $v := \|\mathcal{L}\|_F \cdot C_F$. For disjoint finite sets 
$A,C \subset \Gamma$, any $O_A \in \mathcal{A}_A$, and any 
$K \in \mathcal{CB}_0(C)$ :
\begin{equation}
    \|K(T_t^\Lambda(O_A))\| \leq \frac{\|K\|_{cb}\|O_A\|}{C_F}
    \left(e^{v|t|}-1\right)
    \sum_{x\in A}\sum_{y\in C}F(d(x,y)).
    \label{eq:LR_open}
\end{equation}
\end{thm}

\medskip

\begin{thm}[LR bound for  long-range dissipative processes \cite{hastings_spectral_2006}]
\label{th:HKLR}
For $\alpha>D$, there exist $c,v,\kappa$ such that for local observables $O_A,O_C$ supported on $A,C\subset\Lambda$ respectively, at distance $R=d(A,C)>0$,
\begin{equation}\label{eq:Hastings}
\norm{[\T_t(O_A),O_C]}\ \le\ c\norm{O_A}\norm{O_C}\abs{A}\abs{C}\dfrac{e^{v \abs{t}}-1}{R^\alpha}.
\end{equation}
where the positive constants, $c$ and $v$, depend only on the interaction of the Hamiltonian and the metric of the lattice.
\end{thm}

\medskip

\begin{thm}[LR bound for long-range Hamiltonians 
{\cite[Theorem~1.2]{kuwahara_strictly_2020}}]
\label{thm:LR_hamiltonian}
Let $H = \sum_{Z \subset \Lambda} h_Z$ be a long-range Hamiltonian on 
a $D$-dimensional lattice with $\|h_Z\| \leq J|Z|^{-\alpha}$ and 
$\alpha > 2D+1$. 
For disjoint regions $A,C\subset \La$ at distance $R=d(A,C)>0$ and observables 
$O_A \in \mathcal{A}_A$, $O_C \in \mathcal{A}_C$:

\begin{equation}\label{eq:LRKuwahara}
\|[\T_t(O_A),O_C]\|
\le
c\, \|O_A\|\,\|O_C\|\,
\frac{|t|^{D+1}}{(|R - \tilde{v}|t||)^{\alpha - D}},
\end{equation}
for a constant $c$.
\end{thm}

\section{From dynamical to information theoretic and static properties}\label{sec:long-range}

We now show how dynamical properties of the dissipative systems as defined in the previous section imply static properties on their fixed points. Our results apply to fixed-points of both short and long-range dissipative processes.


\subsection{Decay of Conditional Mutual Information} \label{sec:CMI}

Let $A\sqcup B\sqcup C\subset\Lambda$ be a disjoint tripartition of $\Lambda$, with $B$ shielding\footnote{This means that every path from $A$ to $C$ in the lattice must pass through $B$.} $A$ from $C$ as in Figure \ref{fig:geometryCMI} with $R:=d(A,C)>0$. In this section we will see how the CMI $I(A:C|B)$ decays with $R$.

In order to prove the decay of the CMI we use of the time evolution of a long-range Lindbladian as the recovery map, however, the time evolution of a general Lindbladian is non-local, so we need to localize it in the region $AB$ with an error that grows exponentially with time but decays with the distance $R$. On the other hand, due to rapid mixing, the Lindbladian evolution will converge into the fixed-point (the desired state) exponentially fast in time but will grow with $R$. Then, as we have seen, there are two terms that have to be compensated by a correct choice of $t(R)$.

\subsubsection{Locality of the time-evolution}

Let us begin proving a Lieb-Robinson-type bound inspired by the Lemma $13$ from \cite{brandao_area_2015}. Let \[\sigma_{-A}(t):=(T^\Lambda_t)^* (\sigma_{-A})=(T^\Lambda_t)^*\left (\dfrac{\Id_A}{d_A}\otimes \tr_A(\sigma)\right ) ,\]be the time evolution of the reduced state $\sigma_{-A}$.
\medskip
 \begin{lem}\label{lem:LiebRobinsonTypeBoundGeneral}
       Let $\Li=\{L_Z\}_{Z\subset \mathcal{P}(\G)}$ be a $k$-local Lindbladian doubly anchored to a $F$-function $F$, satisfying frustration freeness. Let $A\sqcup B\sqcup C=\La\subset \subset\G$ be a tripartition of $\La$ as in Figure \ref{fig:geometryCMI}, and $R=d(A,C)$. Then, for each region\footnote{Note that, $A\cap X=\emptyset$.} $X\subset BC$ and superoperator $K\in \mathcal{CB}_0(X)$, the following bound holds
        \begin{equation}
            \normm{K(\sigma_{-A}(t))}{1}\leq \dfrac{\abs{A}\abs{X}}{C_F} (e^{\normm{\Li}{F}C_Ft}-1) \, F(d(X,A)) 
        \end{equation}
        where $C_F$ is defined in Equation \eqref{eq:convolutionCF}.
    \end{lem}
  
\begin{proof}[Proof of Lemma \ref{lem:LiebRobinsonTypeBoundGeneral}]
This proof is inspired by that of Lemma $13$ from \cite{brandao_area_2015}. For each $X\subset BC\subset \La $, we denote by $\tilde\Li_X:=\sum_{Z\cap X\neq \emptyset}L_Z$ and $\Li_X=\summ{Z\subset X}L_Z$.

We define the following quantity
\begin{equation}
    C(X,t):=\sup_{L\in \Li_X}\dfrac{\normm{L(\sigma_{-A}(t))}{1}}{\normm{L}{cb}}.
\end{equation}
Frustration freeness implies that, for any $L_Z$ with $Z\subset X$ and $X\cap A=\emptyset$, then $L_Z(\sigma_{-A})=0$, because $L_Z$ commutes with the (normed) partial trace.

What follows replicates the proof technique of standard Lieb-Robinson bounds \cite{LieRob72,hastings2010localityquantumsystems}. For $K\in \Li_X$, differentiating with respect to time $K(\sigma_{-A}(t))$, we can separate in two terms
\begin{equation}
    \dfrac{d}{dt}K(\sigma_{-A}(t))=K\Li^*(\sigma_{-A}(t))=\Li^*_{\La\setminus X}K(\sigma_{-A}(t))+K\tilde \Li_X^* (\sigma_{-A}(t)),
\end{equation}
since $K$ is time-independent. Then, we can implicitly solve the equation above to obtain an expression for $K(\sigma_{-A}(t))$ (see  \cite{brandao_area_2015}), which is the Duhamel formula (see \cite{renardy_introduction_2006}):
\begin{equation}
    K(\sigma_{-A}(t))=e^{t\Li^*_{\La\setminus X}}K(\sigma_{-A}(0))+\int_0^te^{(t-s)\Li^*_{\La\setminus X}}K\tilde\Li^*_X(\sigma_{-A}(s))ds.
\end{equation}
Taking norms,
\begin{equation}
    \normm{K(\sigma_{-A}(t))}{1}\leq \normm{K(\sigma_{-A}(0))}{1}+\normm{K}{cb}\summ{Z\cap X\neq \emptyset}\int_0^t\normm{L_Z}{cb}C(Z,s)ds,
\end{equation}
since $e^{t\Li^*_{\La\setminus X}}$ contracts the one-norm. This directly implies that:
\begin{equation}
    C(X,t)\leq C(X,0)+\summ{Z\cap X\neq\emptyset}\normm{L_Z}{cb}\int_0^tC(Z,s)ds.
    \label{eq:recursionEquation}
\end{equation}
By recursively applying Equation \eqref{eq:recursionEquation}, we obtain that
\begin{equation}
    C(X,t)\leq \summ{m=0}^\infty a_m\dfrac{t^m}{m!},\label{eq:CXTtaylor}
\end{equation}
where
\begin{equation}
    a_m=\summ{Z_1\cap X\neq \emptyset}\dots \summ{Z_{m-1}\cap Z_m\neq \emptyset}\normm{L_{Z_1}}{cb}\dots \normm{L_{Z_m}}{cb}C(Z_m,0).
\end{equation}
In the case $Z\cap A=\emptyset$, due to frustration freeness, $C(Z,0)=0$: this happens because any operator $L_Z$ for $Z\subset X$  acting on $\rho(0)=(T^\Lambda_0)^*(\sigma_{-A})=\sigma_{-A}$ gives zero, since $L_Z^*\rho(0)=L_Z^*\sigma_{-A}=\dfrac{\Id_A}{d_A}\otimes \tr_A(\underbrace{L_Z^*(\sigma)}_{=0})=0$. If $Z\cap A\neq \emptyset$, then,
\begin{equation}
    a_m=\summ{Z_1\cap X\neq \emptyset}\dots \summ{Z_{m-1}\cap Z_m\neq \emptyset\\Z_m\cap A\neq \emptyset }\normm{L_{Z_1}}{cb}\dots \normm{L_{Z_m}}{cb}.
\end{equation}
since $C(Z,0)\leq 1$. The Lindbladian $\mathcal{L}$ is doubly anchored governed by $F$, so the first term $a_1$ is bounded by 
\begin{equation}
    a_1=\summ{Z\cap A\neq \emptyset\\Z\cap X\neq \emptyset}\normm{L_{Z}}{cb}\leq \normm{\Li}{F}\abs{A}\abs{X} \, F(d(A,X)).
\end{equation}
This follows from summing Eq.~\eqref{eq:FNormLindbladian} for all the point $x\in A$ and $y\in X$.
For the general term
\begin{equation}
    \summ{Z_1\cap X\neq \emptyset}\dots \summ{Z_{m-1}\cap Z_m\neq \emptyset\\Z_m\cap A\neq \emptyset }\normm{L_{Z_1}}{cb}\dots \normm{L_{Z_m}}{cb},
\end{equation}
the idea is to use the convolution property of $F$-functions (Equation \eqref{eq:convolutionCF}). For that, we need to decompose the above sum into sums over points in the intersection between consecutive $Z_i$ sets as follows
\begin{align}
    \summ{i_1,\dots,i_{m-1}\in \La\\i_0\in X\\i_m\in A}&\summ{Z_1\ni i_0,i_1}\dots \summ{Z_m \ni i_{m-1},i_m}\normm{L_{Z_1}}{cb}\dots \normm{L_{Z_m}}{cb}\\
    &\nonumber=\summ{i_1,\dots,i_{m-1}\in \La\\i_0\in X\\i_m\in A}\underbrace{\summ{Z_1\ni i_0,i_1}\normm{L_{Z_1}}{cb}}_{\leq \normm{\Li}{F}F(d(i_0,i_1))}\dots \underbrace{\summ{Z_m \ni i_{m-1},i_m}\normm{L_{Z_m}}{cb}}_{\leq \normm{\Li}{F}F(d(i_{m-1},i_m))}\\
    &\nonumber\leq\normm{\Li}{F}^m \summ{i_1,\dots,i_{m-1}\in \La\\i_0\in X\\i_m\in A}F(d(i_{0},i_1))\dots F(d(i_{m-1},i_m))\\
    &\nonumber\leq  \abs{X}\abs{A} \normm{\Li}{F}^m C_F^{m-1} \, F (d(X,A)) ,
\end{align}
where $C_F$ is defined in Equation \eqref{eq:convolutionCF}. Hence, for every $m\in\mathbb{N}$,
\begin{equation}
     a_m\leq \abs{X}\abs{A} \normm{\Li}{F}^m C_F^{m-1} \, F (d(X,A)).
\end{equation}
Therefore, we conclude that 
\begin{equation}
    \normm{K(\sigma_{-A}(t))}{1}\leq C(X,t)\leq \dfrac{\abs{A}\abs{X}}{C_F} (e^{\normm{\Li}{F}C_Ft}-1) F(d(X,A)) ,
\end{equation}
where in the last step we used Eq.~\eqref{eq:CXTtaylor} together with the estimate on $a_m$.
\end{proof}  

The Lieb-Robinson bound above allows to localize the time evolution to $AB$ appplied to the reduced state $\sigma_{-A}$.

    \medskip
    \begin{lem}\label{lem:differenceTimeEvolutionsGeneral}
       Let $\Li=\{L_Z\}_{Z\subset \mathcal{P}(\G)}$ be a  frustration free $k$-local Lindbladian doubly anchored to a $F$-function $F$. Let $A\sqcup B\sqcup C=\La\subset \subset\G$ be a tripartition of $\La$ as above, and $R=d(A,C)$. Then,
       \begin{equation}
            \normm{\left (\left (T^\Lambda_t\right )^*-\left (T^{AB}_t\right )^*\right )(\sigma_{-A})}{1}\leq p(\abs{A},\abs{C}) (e^{\normm{\Li}{F}C_Ft}-1+t)R^DF(R/2)  \, ,     
        \end{equation}
       where $p:\mathbb{N}\times\mathbb{N}\to \mathbb{N}$ is a polynomial function.
    \end{lem}
   
    \begin{proof}[Proof of Lemma \ref{lem:differenceTimeEvolutionsGeneral}]
    Applying Duhamel's formula, we have 
    \begin{equation}
        ((T^\Lambda_t)^*-(T^{AB}_{t})^*)(\sigma_{-A})=\int_0^t (T^{AB}_{t-s})^*\tilde \Li^*_C\rho(s)ds,
    \end{equation}
    and thus, taking norms, 
    \begin{equation}
        \normm{((T^\Lambda_t)^*-(T^{AB}_{t})^*)(\sigma_{-A})}{1}\leq \int_0^t\summ{Z\cap C\neq \emptyset}\normm{L_Z^*\rho(s)}{1}ds.
    \end{equation}
    Before moving on, let us divide $B$ into two disjoint subsets $B_1\sqcup B_2=B$ as follows
    \begin{equation}
        B_1:=\{x\in B:d(A,x)\leq d(A,C)/2\},\quad B_2:=B\setminus B_1.
    \end{equation}

    Now, we will divide the sum $\summ{Z\cap C\neq \emptyset}\normm{L_Z^*\rho(s)}{1}$ into two sums in
    \begin{itemize}
        \item $\Omega_1=\{Z\subset \La :Z\cap C\neq \emptyset$ and $Z\cap AB_1\neq \emptyset\}$, and
        \item $\Omega_2=\{Z\subset \La :Z\cap C\neq \emptyset$ and $Z\cap AB_1= \emptyset\}$.
    \end{itemize}
    The idea behind this division is that the elements in $\Omega_1$ are at least $R/2$ long (recall $R=d(A,C)$), so their norm should decay with $R$. Meanwhile, the elements in $\Omega_2$ are at least $R/2$ far from $A$, so that we can apply Lemma \ref{lem:LiebRobinsonTypeBoundGeneral} as a Lieb-Robinson-type bound that also decays with $R$.
 On one hand, the terms in $\Omega_1$ are bounded by:
    \begin{multline}
        \summ{Z\in\Omega_1}\normm{L_Z}{cb}\leq \abs{AB_1}\abs{C}\normm{\Li}{F}F(R/2)\\
        \leq \abs{\Omega_1}\abs{A}\kappa (R/2)^{D}\abs{C}F(R/2)\leq  p'(\abs{A},\abs{C})R^DF(R/2),
    \end{multline}
    with $p'$ a two-entry polynomial.
    
    On the other hand, those terms that are in $\Omega_2$ are at distance at least $R/2$:
    \begin{multline}
        \summ{Z\in \Omega_2}\normm{L_Z ^*\rho(s)}{1}\leq \summ{Z\in\Omega_2}\dfrac{\abs{A}\abs{Z}}{C_F} (e^{\normm{\Li}{F}C_Fs}-1) \, F(d(A,Z))\\
        =p''(\abs{A},\abs{C}) \, (e^{\normm{\Li}{F}C_F s}-1)F(R/2),
    \end{multline}
    for a polynomial $p''$.
    We end the proof integrating over $s$:
    \begin{equation}
        \normm{((T^\Lambda_t)^*-(T^{AB}_{t})^*)\sigma_{-A}}{1}\leq p(\abs{A},\abs{C}) (e^{\normm{\Li}{F}C_Ft}-1+t)R^DF(R/2)
    \end{equation}
    with $p$ also polynomial.
    \end{proof}

\subsubsection{Recovery map}

Using the Lemma \ref{lem:differenceTimeEvolutionsGeneral} together with the condition that $\Li$ satisfies rapid-mixing, we can prove the following decay of the conditional mutual information. 

\medskip
\begin{thm}\label{thm:decayCMIgeneralLindbladian}
    Let $\Li=\sum_{Z\subset \La}L_Z$ be a $k$-local, primitive, frustration-free Lindbladian doubly anchored to the $F$-function $F$, satisfying rapid mixing RM($n,\lambda$) for $n\in\mathbb{N}^+$ and $\lambda >0$ (see Def. \ref{def:globalRM}), and let $\sigma$ be the fixed-point of $\Li$. Then, for any tripartition $A,B,C\subset\La$ where $A\sqcup B\sqcup C=\La$ with the geometry in Figure \ref{fig:geometryCMI}, with $R:=d(A,C)$. Then,
    \begin{equation}
            I(A : C \,|\, B)_\sigma \le \, p(\abs{A},\abs{C})\, R^{nD/2} (F(R))^{\lambda/2(\lambda+v)} ,
    \end{equation}
    where $p:\mathbb{N}\times\mathbb{N}\to \mathbb{N}$ is a polynomial function and $v=\normm{\Li}{F}C_F$ the Lieb-Robinson velocity. 

\end{thm}
\begin{proof}
    In order to prove decay of CMI, we will define a recovery map and then use Proposition \ref{prop:FR}. Consider
    \begin{equation}
        \recovery:=e^{t\Li_{AB}^*} \, ,
    \end{equation}
    where $\Li_{AB}=\sum_{Z\subset AB}L_Z$. Note that it is a sum of $k$-local terms that act only on $AB$. Let $\sigma$ be the unique fixed point of $\Li$, namely $\Li^*(\sigma)=0$. Due to frustration-freeness, $\forall Z\subset \La$, $L_Z^*(\sigma)=0$; in particular, $\Li_{AB}^*(\sigma)=0$. Let us denote by $\sigma_{-A}:=\dfrac{\Id_A}{d_A}\otimes\tr_A(\sigma)$ the normalised partial trace over $A$ of the fixed-point (see Equation \eqref{eq:partialTrace}), and let $ \RR^t:=(T^\Lambda_t)^*=e^{t\Li^*} $
    be the global evolution (unlike the recovery map, this map acts on the entire $\La$). We aim to bound, using the triangular inequality,
    \begin{equation}
        \normm{\recovery (\sigma_{-A})-\sigma}{1}\leq \underbrace{\normm{\RR^t(\sigma_{-A})-\sigma}{1}}_{(1)}+\underbrace{\normm{(\recovery-\RR^t)(\sigma_{-A})}{1}}_{(2)}.
        \label{eq:CMItriangularInequality}
    \end{equation}
    The first term $(1)$ decays due to Hölder's inequality and rapid mixing,
    \begin{equation}
        (1):\normm{\RR^t(\sigma_{-A})-\sigma}{1}\leq\normm{(T^\Lambda_t)^*-(T^\Lambda_\infty)^*}{}\leq p(\abs{\La})e^{-\lambda t}=c(\abs{A}+\abs{B}+\abs{C})^ne^{-\lambda t},
    \end{equation}
    where recall that $n$ is the degree of the polynomial $p$ in the definition of global rapid mixing (see Definition \ref{def:globalRM}). Since $\abs{B}$ grows as $R^D$, the factor $\abs{B}^n$ adds a factor of $R^{nD}$ to the equation above
    \begin{equation}
        \normm{\RR^t(\sigma_{-A})-\sigma}{1}\leq p(\abs{A},\abs{C}) R^{nD}e^{-\lambda t}.
    \end{equation}

     The second term $(2)$ is bounded by the difference between the time evolutions generated by $\Li$ and $\Li_{AB}$ over $\sigma_{-A}$.
    \begin{equation}
         (2):\normm{(\recovery-\RR^t)(\sigma_{-A})}{1}=\normm{(e^{t\Li_{AB}}-e^{t\Li})(\sigma_{-A})}{1}
         \label{eq:diferenceRecoveries1}
      \end{equation}
The time evolution on a fixed-point with a perturbation on a set $A$, for short times, is localized ``around the bulk''; This holds due to frustration freeness and rapid mixing. For that, we use Lemma \ref{lem:differenceTimeEvolutionsGeneral}. 
       \begin{equation}
            \normm{\left ((T^\Lambda_t)^*-\left (T^{AB}_t\right )^*\right )(\sigma_{-A})}{1}\leq p(\abs{A},\abs{C}) (e^{\normm{\Li}{F}C_Ft}-1+t)R^{D}F(R/2)  \, ,     
        \end{equation}
Putting both bounds together in Equation~\eqref{eq:CMItriangularInequality},
\begin{equation}
    \normm{\recovery (\sigma_{-A})-\sigma}{1}\leq p(\abs{A},\abs{C})R^{nD}\left (e^{-\lambda t}+\bigl[e^{v t}-1+t\bigr] F(R)\right ),
\end{equation}
for $v=\normm{\Li}{F}C_F$. By choosing $t_0(R)=\dfrac{1}{\lambda +v }\log\left(\dfrac{F^{-1}(R)}{C_0}\right)$
we obtain that the decay is given by
\begin{equation}
    \normm{\mathcal{R}_{B\to AB}^{t_0} (\sigma_{-A})-\sigma}{1}\leq  \,p(\abs{A},\abs{C}) \cdot R^{nD} (F(R))^{\lambda/(\lambda+v)},
\end{equation}

Define the final recovery map as\[\mathcal{R}_{B\to AB}:=\mathcal{R}_{B\to AB}^{t_0}. \]

The existence of a local recovery map as $\RR_{B\to AB}$ implies that the conditional mutual information decays by Fawzi-Renner (see \cite{FawRen15}). We can now apply Proposition~\ref{prop:FR} combined with the trace distance via Fid$ \geq 1- \frac12 \normm{\cdot }{1}$, so that the CMI decays as
\begin{equation}
    I(A : C \,|\, B)_\sigma \le \, p(\abs{A},\abs{C})\, R^{nD/2} (F(R))^{\lambda/2(\lambda+v)} ,
\end{equation}
so that $\sigma$ satisfies a global Markov property. The additional factor of $1/2$ in the rate comes from the square root when relating fidelity and $1$-norm.
\end{proof}

In particular, if the Lindbladian is long-range, then, the CMI decays polynomially with $R$:
\medskip
\begin{cor}\label{cor:decayCMILongRange}
For local long-range Lindbladians with a decay rate\[\alpha > \dfrac{n(\lambda+v)}{\lambda}D,\]satisfying frustration freeness and RM($n,\lambda$), the CMI decays as
\begin{equation}
    I(A : C \,|\, B)_\rho \le \, p(\abs{A},\abs{C})\, F_{\alpha'}(R) ,
\end{equation}
for $\alpha'=\dfrac{\lambda \alpha}{2(\lambda +v)}-\dfrac{nD}{2}$.
\end{cor}

\subsection{Decay of Mutual Information}

In order to prove decay of mutual information, we first study how perturbations on the Lindbladian affect the fixed point. In other words, we bound how close the fixed point of the perturbed Lindbladian is to the original one. Crucially, this argument only needs locality and rapid mixing, and does not rely on frustration freeness as Theorem \ref{thm:decayCMIgeneralLindbladian} does. 

Firstly, in Section \ref{subsubsec:perturbedDynamics} we study how far-enough perturbations in the Lindbladian are locally indetectable, using Lieb-Robinson bounds for long-range systems. Using this, in Section \ref{subsubsec:perturbingBoundary} we prove that removing the terms of the Lindbladian that cross the boundary between two far-enough regions does not affect the fixed-point locally. Finally, in Section \ref{sec:MI}, we prove that the MI decays with the distance using the results above.

\subsubsection{Perturbed dynamics}\label{subsubsec:perturbedDynamics}

We first show that perturbations far away from a region cannot be distinguished locally because quantum systems have a finite (effective) speed of propagation given by the so-called Lieb-Robinson velocity, as defined in Theorem \ref{thm:LR_open}.

\medskip
\begin{thm}\label{thm:localPerturbation}
    Let $(\Gamma,d)$ be a countable metric space equipped with an F-function $F$. Let $\Li=\{L_Z\}_{Z\in \mathcal{P}_0(\Gamma)}$ be a dissipative interaction with $\normm{\Li}{F}<\infty$ and primitive, satisfying local rapid mixing. Fix $A,C\in \mathcal{P}_0(\Gamma)$ with $R:=d(A,C)>0$. For any $\Lambda\in  \mathcal{P}_0(\Gamma)$ with $A\cup C \subset \Lambda$, $O_A\in \mathcal{A}_A$, a perturbation $Q_C\in \mathcal{CB}_0(C)$ and $\LL:=\Li+Q_C$ the perturbed \textit{Lindbladian}.
Let $\sigma$ be the unique fixed point of $\Li$ with convergence governed by $g(t)=e^{-\lambda t}$ for $\lambda >0$ and $\pi$ a fixed point of $\LL$.
Then,
    \begin{equation}\label{eq:diff_fixed_points_Lindbladians}
        \abs{(\sigma-\pi)(O_A)}\leq c\cdot \norm{O_A} \cdot k(\abs{A})\cdot F(R)^{\lambda/(\lambda +v)} \, ,
    \end{equation}
    where $v=\normm{\mathcal{L}}{F}\cdot C_F$, $k$ a polynomial function and a constant $c=1+1/(v\cdot C_F)$.
\end{thm}
\begin{proof}

The left-hand side of Equation ~\eqref{eq:diff_fixed_points_Lindbladians} can be rewritten as
\begin{align}
\abs{(\sigma-\pi)(O_A)} &= \abs{\sigma(O_A)-\pi(O_A)} \\
&= \abs{\sigma(\T_t(O_A))-\pi(\TT_t(O_A))} \nonumber\\
&= \abs{(\sigma-\pi)(\T_t(O_A))-\pi((\TT_t-\T_t)(O_A))} \nonumber\\
&\le \abs{(\sigma-\pi)(\T_t(O_A))}+\abs{\pi((\TT_t-\T_t)(O_A))} \nonumber\\
&\le \underbrace{\abs{(\sigma-\pi)(\T_t(O_A))}}_{(1)}+\underbrace{\normm{(\TT_t-\T_t)(O_A)}{1}}_{(2)} \, , \nonumber
\end{align}
where we are using that $\sigma$ (resp. $\pi$), is a fixed point of $\Li$ (resp. $\LL$), on the second equality, triangle inequality in the first inequality and Hölder's in the last one. Let us analyse the last two terms separately. For the first one, we have
 \begin{equation*}
    (1):\quad \abs{(\sigma-\pi)(\T_t(O_A))} \leq p(\abs{A})\cdot   \norm{O_A} \cdot e^{-\lambda t}
 \end{equation*}
from local rapid mixing. For the second one, 
\begin{align}
(2):\;\lvert\lvert (\TT_t&-\T_t)(O_A)\rvert\rvert_1
\le \int_0^t \norm{Q_C(\T_s(O_A))}\,ds \\
&\nonumber\le \int_0^t\!\left(\dfrac{\normm{Q_C}{cb}\,\norm{O_A}}{C_F}\Big(e^{\normm{\Li}{F} C_F s}-1\Big)\sum_{x\in A}\sum_{y\in C}F(d(x,y))\right)\!ds
\\
&= \dfrac{\normm{Q_C}{cb}\,\norm{O_A}}{C_F}\left(\dfrac{e^{\normm{\Li}{F} C_F t}-\normm{\Li}{F} C_F t-1}{\normm{\Li}{F} C_F}\right)\sum_{x\in A}\sum_{y\in C}F(d(x,y)),
\nonumber
\end{align}
where we used the following expression for the difference of dynamics
\begin{equation}
    (\TT_t-\T_t)(O_A)=-\int_0^t \dfrac{d}{ds}\TT_{t-s}(\T_s(O_A)) ds=\int_0^t \TT_{t-s}((\LL-\Li)(\T_s(O_A)))ds , 
\end{equation}
and by the Duhamell's formula:
\begin{equation}
    \norm{(\TT_t-\T_t)(O_A)}\leq \int_0^t \norm{(\LL-\Li)(\T_s(O_A))}ds=\int_0^t \norm{Q_C(\T_s(O_A))}ds.
\end{equation}
With this, we obtain
\begin{multline}
    \abs{(\sigma-\pi)(O_A)}\leq\norm{O_A}\left (  p(\abs{A})\cdot e^{-\lambda t}\right.\\
    + \left.\dfrac{\normm{Q_C}{cb}}{C_F}\left ( \dfrac{e^{\normm{\Li}{F} C_F \cdot t}-\normm{\Li}{F} C_F\cdot t-1}{\normm{\Li}{F} C_F}\right ) \sum_{x\in A}\sum_{y\in C}F(d(x,y)) \right).
\end{multline}
Note that $F$ is a decreasing function, so that 
\begin{equation}
    \sum_{x\in A}\sum_{y\in C}F(d(x,y)) \leq \abs{A}\cdot\abs{C} F(R)
\end{equation}
where $R=d(A,C)$.
Then,

\begin{align}
\lvert (\sigma&-\pi)(O_A)\rvert\\
\nonumber&\le \norm{O_A}\left( p(\abs{A})\,e^{-\lambda t} + \dfrac{\normm{Q_C}{cb}}{C_F}\left( \dfrac{e^{\normm{\Li}{F} C_F t}-\normm{\Li}{F} C_F t-1}{\normm{\Li}{F} C_F}\right)\abs{A}\,\abs{C}\,F(R) \right) \\
&\le \norm{O_A}\left( p(\abs{A})\,e^{-\lambda t} + \dfrac{\normm{Q_C}{cb}}{v\,C_F} e^{\normm{\Li}{F} C_F t}\,\abs{A}\,\abs{C}\,F(R) \right). \nonumber
\end{align}

Defining $v:=\normm{\Li}{F} C_F$. Setting $\tilde{p}=\tilde{p}(A,C)=\dfrac{p(\abs{A})}{v\cdot \normm{Q_C}{cb}\cdot \abs{A}\cdot \abs{C}}$, and choosing a $t(R)$ such that
\begin{equation}
    e^{v t(R)}F(R)=\tilde{p} \cdot e^{-\lambda \cdot t(R)},
\end{equation}
we have that $t(R)=\log\left [ (p^{-1} F(R) )^{-1/(\lambda + v)}  \right ]$. Under such choice, it holds that
\begin{equation}
    e^{-\lambda \cdot t(R)}=F^{\lambda/(\lambda +v)}(R) \cdot  p^{-\lambda/(\lambda +v)}.
\end{equation}
We conclude that
\begin{equation}
    \abs{(\sigma-\pi)(O_A)}\leq c\cdot \norm{O_A} \cdot p(\abs{A})^{v/(v+\lambda)}\cdot F(R)^{\lambda/(\lambda +v)}
\end{equation}
with the constant $c:=1+1/(v\cdot C_F)$.
\end{proof}

By specializing the previous result to algebraically decaying $F$-functions, we obtain:
\medskip
\begin{cor}
    If $F_\alpha(R)=(1+R)^{-\alpha}$, for $\alpha>D$,
    \begin{equation}
        \abs{(\sigma-\pi)(O_A)}\leq c\cdot \norm{O_A} \cdot k(\abs{A})\cdot (1+R)^{-\alpha \lambda/(\lambda + v)}.
    \end{equation}
\end{cor}


\subsubsection{Perturbing terms in the boundary }\label{subsubsec:perturbingBoundary}

The following step now consists of perturbing the Lindbladian by removing the jump operators of the boundary $\partial_{AC}$ that disconnects two sets $A,C\subset \La$. The technical difficulty of doing this in the long-range regime comes from the fact that there are terms crossing the boundary that intersect $A\cup C$. This means that we need to consider those terms that intersect with $A\cup C$, for which Theorem \ref{thm:localPerturbation} provides the trivial bound.
\medskip
\begin{figure}[!h]
    \centering
    \includegraphics[width=0.5\linewidth]{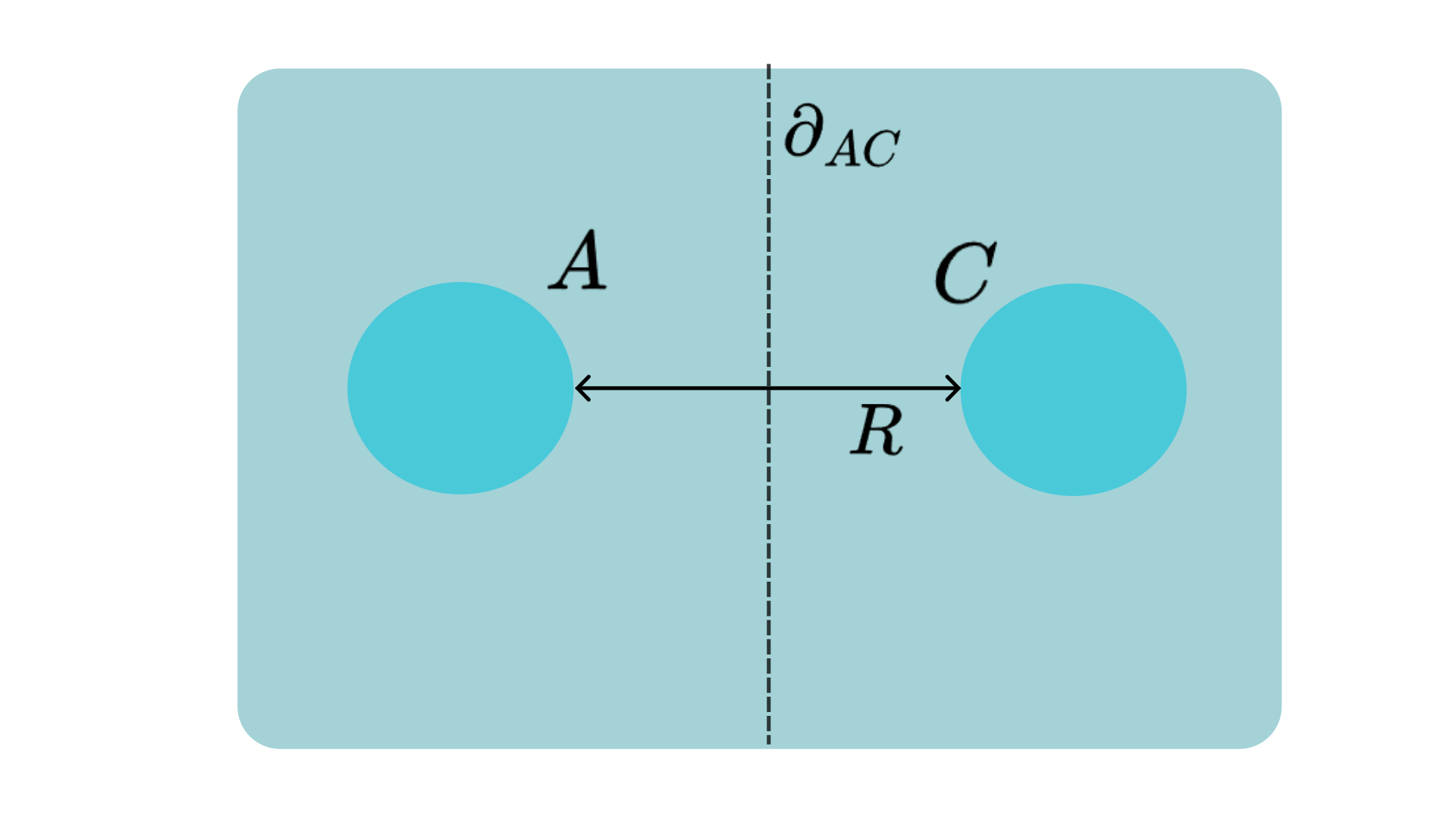}
    \caption{Geometry considered for the MI: $A$ and $C$ and their frontier $\partial_{AC}$}
    \label{fig:MIgeometry}
\end{figure}
\begin{thm}\label{thm:clusteringBoundary}
    Let $(\Gamma,d)$ be a countable metric space equipped with the F-function\[F_\alpha(r):=\dfrac{1}{(1+r)^\alpha},\]for $r\geq0$ and $\alpha>3D-1$. Let $\Li=\{L_Z\}_{Z\in \mathcal{P}_0(\Gamma)}$ be a local and primitive dissipative interaction with $\normm{\Li}{F}<\infty$, satisfying local rapid mixing. Fix $A,C\in \mathcal{P}_0(\Gamma)$ with $A\cap C =\emptyset$ (see Figure \ref{fig:MIgeometry}). For any $\Lambda\in  \mathcal{P}_0(\Gamma)$ with $A\cup C \subset \Lambda$, $O\in \mathcal{A}_{A\cup C}$, define the perturbation\footnote{We denote by $\partial_{AC}$ to the boundary separating $A$ and $C$ which is chosen equidistant to them (see Figure \ref{fig:MIgeometry}). With some abuse of notation, we denote $Z\cap \partial_{AC}\neq \emptyset$ to $Z$ that cross or touch the boundary $\partial_{AC}$.} $Q:=-\sum\limits_{\substack{Z\subset \Lambda \\ Z\cap \partial_{AC}\neq \emptyset}}L_Z\in CB_0(\partial_{AC})$ and $\LL:=\Li+ Q $ the perturbed \textit{Lindbladian}.
Let $\sigma\in\A_\La^*$ be the unique fixed point of $\Li$, which satisfies local rapid mixing ($\Li$ with a convergence governed by $g(t)=e^{-\lambda t}$ for $\lambda >0$), and let $\pi\in\A_\La^*$ a fixed point of $\LL$. Then,
    \begin{equation}
        \abs{(\sigma-\pi)(O)}\leq \norm{O}\cdot  k(\abs{A},\abs{C})\cdot F_{\alpha '}(R-1),
    \end{equation}
    where $\alpha'=\dfrac{\lambda(\alpha-2D)}{\lambda +v}$, $v=\normm{\Li}{F}\cdot  C_F$ the Lieb-Robinson speed, $k$ is a polynomial function, and $R=d(A,C)$. 

\end{thm}
\begin{proof}
    Notice that we can split the term $\abs{(\sigma-\pi)(O)}$ into two terms by using triangular inequality: the first term decays with time due to rapid mixing, the second term decays with $R$:
  \begin{align}
\abs{(\sigma-\pi)(O)} &= \abs{\sigma(O)-\pi(O)} \label{eq:diferenceStates} \\
&= \abs{\sigma(\T_t(O))-\pi(\TT_t(O))} \nonumber\\
&= \abs{(\sigma-\pi)(\T_t(O))-\pi((\TT_t-\T_t)(O))} \nonumber\\
&\le \abs{(\sigma-\pi)(\T_t(O))}+\abs{\pi((\TT_t-\T_t)(O))} \nonumber\\
&\le \underbrace{\abs{(\sigma-\pi)(\T_t(O))}}_{(1)}+\underbrace{\normm{(\TT_t-\T_t)(O)}{1}}_{(2)}. \nonumber
\end{align}

    For the first term, we use that $\Li$ satisfies local rapid mixing:
     \begin{equation*}
        (1):\quad \abs{(\sigma-\pi)(\T_t(O))} \leq p(\abs{A})\cdot   \norm{O} \cdot e^{-\lambda t}
     \end{equation*}
     For the second term, we will distinguish between two types of perturbation: those that intersect $A$ or $C$, and those that do not. This distinction motivates the following definitions:
     \begin{equation}
         \Omega:=\{ Z\subset \Lambda : Z\cap \partial_{AC}\neq \emptyset,Z\cap(A\cup C)= \emptyset\},
         \label{eq:Omega}
     \end{equation}
     \begin{equation}
         \Theta := \{ Z\subset \Lambda : Z\cap \partial_{AC}\neq \emptyset,Z\cap(A\cup C)\neq \emptyset\}.
         \label{eq:Theta}
     \end{equation}
     This decomposition separates perturbation terms according to whether
    they are spatially disjoint from the observable support $A\cup C$.
    The terms in $\Omega$ are at positive distance from $A\cup C$ and can
    therefore be controlled directly using the long-range Lieb–Robinson
    estimate. 
    
    In contrast, the terms in $\Theta$ intersect $A\cup C$ and cannot be
    treated purely by quasi-locality. This distinction is essential in the
    long-range setting, since interaction terms crossing $\partial_{AC}$
    may have arbitrarily large diameter.
    
    Using this decomposition, the perturbed Lindbladian can be written as
     \begin{equation}
         \LL=\Li -\sum_{Z'\in \Omega}L_{Z'}-\sum_{Z''\in \Theta}L_{Z''}
     \end{equation}
    we can rewrite the second term in the inequality above as:
     \begin{multline}
        (2):\normm{(\TT_t-\T_t)(O)}{1}\leq \int_0^t\underbrace{\sum_{Z'\in \Omega} \norm{L_{Z'}(\T_s(O))}}_{(4)}ds+\int_0^t\underbrace{\sum_{Z''\in \Theta} \norm{L_{Z''}(\T_s(O))}}_{(3)}ds
        \label{eq:(2)}
     \end{multline}
 Subsequently, we bound the term $(3)$ by applying Lemma \ref{lem:technicalLemmaFfunction2} and using the contractivity of $\T_t$.  
\begin{align}
(3):\;\sum_{Z''\in \Theta} \norm{L_{Z''}(\T_s(O))} &\le \sum_{Z''\in \Theta} \normm{L_{Z''}}{cb}\norm{\T_s(O)} \\
&\le \norm{O}\sum_{Z''\in \Theta} \normm{L_{Z''}}{cb} \nonumber\\
&\le \norm{O}\sum_{x\in A\cup C}\sum_{l\ge R/2}\underbrace{\sum_{\substack{Z''\subset \La \\ Z''\ni x \\ d_x(Z'')=l}}\norm{L_{Z''}}}_{\le \kappa \normm{\Li}{F} F_{\alpha-D+1}(l)} \nonumber\\
&\le \kappa \normm{\Li}{F}\norm{O}\abs{A\cup C}\underbrace{\sum_{l\ge R/2} F_{\alpha-D+1}(l)}_{\le \frac{1}{\alpha-D}F_{\alpha-D}(R/2-1)} \nonumber\\
&\le \norm{O}\underbrace{\kappa \normm{\Li}{F}\dfrac{1}{\alpha-D}\abs{A\cup C}}_{=:\tilde{c}'} F_{\alpha-D}(R/2-1). \nonumber
\end{align}
     
     For the terms of the Lindbladian that cross the boundary $\partial_{AC}$ but do not intersect $A\cup C$, we can use Theorem \ref{thm:localPerturbation}:
    \begin{align}
    (4):\;\sum_{Z'\in \Omega} & \norm{L_{Z'}(\T_s(O))}\nonumber \\
    &\nonumber \le \sum_{Z'\in \Omega}\left( \dfrac{\normm{L_{Z'}}{cb}\,\norm{O}}{C_F}\left( e^{\normm{\Li}{F} C_F s}-1\right)\sum_{x\in A\cup C}\sum_{y\in Z'}F_\alpha(d(x,y))\right) \\
    &= \left( \dfrac{\norm{O}}{C_F}\left( e^{\normm{\Li}{F} C_F s}-1\right)\underbrace{\sum_{Z'\in \Omega}\normm{L_{Z'}}{cb}\sum_{x\in A\cup C}\sum_{y\in Z'}F_\alpha(d(x,y))}_{(*)}\right).
    \label{eq:termsCrossingBoundary}
    \end{align}
    
    We will rewrite the term $(*)$ and split it into two terms (a) and (b) that are bounded independently:
    \begin{align}
(*):\;\sum_{Z'\in \Omega} \normm{L_{Z'}}{cb}&\sum_{x\in A\cup C}\sum_{y\in Z'}F_\alpha(d(x,y)) \nonumber\\ 
&\le \sum_{r\ge 1}\sum_{\substack{Z'\subset \La \\ d(A\cup C,Z')= r}}\normm{L_{Z'}}{cb}\sum_{x\in A\cup C}\sum_{y\in Z'}F_\alpha(d(x,y)) \label{eq:(*)split} \\
&= \underbrace{\sum_{r=1}^{R/2}\sum_{\substack{Z'\subset \La \\ d(A\cup C,Z')= r}}\normm{L_{Z'}}{cb}\sum_{x\in A\cup C}\sum_{y\in Z'}F_\alpha(d(x,y))}_{(a)} \nonumber\\
&\quad + \underbrace{\sum_{r\ge R/2+1}\sum_{\substack{Z'\subset \La \\ d(A\cup C,Z')= r}}\normm{L_{Z'}}{cb}\sum_{x\in A\cup C}\sum_{y\in Z'}F_\alpha(d(x,y))}_{(b)}. \nonumber
\end{align}

     Before continuing, note that
     \begin{equation}
         \sum_{x\in A\cup C}\sum_{y\in Z'}F_\alpha(d(x,y))\leq \text{min}\{\abs{A\cup C},\abs{Z'}\}\,\sup_{x\in \Lambda}\sum\limits_{\substack{y\in \Lambda \\ 
         d(x,y)\geq d(A\cup C,Z')
         }}   
         F_\alpha\left (d(x,y)\right).
     \end{equation}
    For the decay function $F_\alpha (r):=(1+r)^{-\alpha}$, a simple integration shows that for all $R\in \mathbb{N},\alpha > D$, there exists a constant $c > 0$ such that
    \begin{equation}
        \sup_{x\in \Lambda}\sum\limits_{\substack{y\in \Lambda \\ 
         d(x,y)\geq R
         }}   
         F_\alpha\left (d(x,y)\right)\leq \dfrac{c}{4}F_{\alpha-D}(R).
    \end{equation}
    In particular,
    \begin{align}
        \sum_{x\in A\cup C}\sum_{y\in Z'}F_\alpha(d(x,y))& \leq \dfrac{c}{4}\text{min}\{\abs{A\cup C},\abs{Z'}\}F_{\alpha-D}(d(A\cup C,Z')) \\ & \nonumber \leq\dfrac{c}{4}(\abs{A}+\abs{C})F_{\alpha-D}(d(A\cup C,Z')).
    \end{align}
    Going back to the Equation \eqref{eq:(*)split}. The first term $(a)$ corresponds to a sum over sets $Z\subset \La$ with diameter greater than $R/2-r$.
\begin{align}
(a):\;\sum_{r=1}^{R/2}&
\sum_{\substack{Z'\subset \La \\ d(A\cup C,Z')= r}}
\normm{L_{Z'}}{cb}
\underbrace{\sum_{x\in A\cup C}\sum_{y\in Z'}F_\alpha(d(x,y))}_{\le \dfrac{c}{4}\abs{A\cup C}F_{\alpha-D}(r)}
\\
&\le
\dfrac{c}{4}\abs{A\cup C}
\sum_{r=1}^{R/2}
F_{\alpha-D}(r)
\sum_{l\ge R/2-r}
\sum_{\substack{Z'\subset \La \\ d(A\cup C,Z')= r \\ d_x(Z')=l}}
\normm{L_{Z'}}{cb}
\nonumber\\
&\le
\dfrac{c}{4}\abs{A\cup C}
\sum_{r=1}^{R/2}
F_{\alpha-D}(r)
\sum_{l\ge R/2-r}
\sum_{\substack{x\in \La \\ d(x,\partial_{AC})= r}}
\sum_{\substack{x\in Z'\subset \La \\ d(A\cup C,Z')= r \\ d_x(Z')=l}}
\normm{L_{Z'}}{cb}
\nonumber\\
&\le
\dfrac{c}{4}\kappa \abs{A\cup C}
\sum_{r=1}^{R/2}
(1+r)^{D-1}
F_{\alpha-D}(r)
\sum_{l\ge R/2-r}
(\abs{\partial_{A}}+\abs{\partial_C})\underbrace{\sum_{\substack{x\in Z'\subset \La \\ d(A\cup C,Z')= r \\ d_x(Z')=l}}
\normm{L_{Z'}}{cb}}_{\le \kappa \normm{\Li}{F}F_{\alpha-D+1}(l)}
\nonumber\\
&\le
\dfrac{c}{4}\kappa^2 \normm{\Li}{F}
(\abs{\partial_A}+\abs{\partial_C})
\abs{A\cup C}
\sum_{r=1}^{R/2}
F_{\alpha-2D+1}(r) \underbrace{\sum_{l\ge R/2-r}F_{\alpha-D+1}(l)}_{\le \frac{1}{\alpha-D}F_{\alpha-D}(R/2-r-1)}
\nonumber\\
&\le
\dfrac{c}{4}\kappa^2 \normm{\Li}{F}
\frac{1}{\alpha-D}
(\abs{\partial_A}+\abs{\partial_C})
\abs{A\cup C}
\underbrace{\sum_{r=1}^{R/2}
F_{\alpha-2D+1}(r-1)
F_{\alpha-2D+1}(R/2-r-1)}_{\alpha-2D+1 >D \text{ needed for convolution}}
\nonumber\\
&\le
\dfrac{c}{4}\kappa^2 \normm{\Li}{F}
\frac{C_F}{\alpha-D}
(\abs{\partial_A}+\abs{\partial_C})
\abs{A\cup C}
F_{\alpha-2D+1}(R/2-1),
\nonumber
\end{align}

where we used the convolution property of the $F$-functions, see Equation \ref{eq:convolutionCF}. Remember that the convolution property is a property of $F$-functions, and $F_{alpha}$ is a $F$-function iff $\alpha >D$. That is the reason why the condition $\alpha -2D+1>D$ arises, because we need $F_{\alpha-2D+1}$ to be a $F$-function.

For the second term $(b)$, we use that $\norm{L_Z}\leq \normm{\Li}{F};\forall Z\subset \La $:
    \begin{align}
(b):\;\sum_{r\ge R/2+1}&\sum_{\substack{Z\subset \La \\ d(A\cup C,Z)= r}}\normm{L_{Z}}{cb}\sum_{x\in A\cup C}\sum_{y\in Z}F_\alpha(d(x,y))\\
&\le \dfrac{c}{4}\abs{A\cup C}\sum_{r\ge R/2+1}\sum_{\substack{Z\subset \La \\ d(A\cup C,Z)= r}}\normm{L_{Z}}{cb}F_{\alpha-D}(d(A\cup C,Z))
\nonumber\\
&\le \dfrac{c}{4}\abs{A\cup C}\sum_{r\ge R/2+1}F_{\alpha-D}(r)\sum_{l\ge 1}\sum_{\substack{x\in \La \\ d(x,\partial_{AC})= r}}\sum_{\substack{x\in Z\subset \La \\ d(A\cup C,Z)= r \\ d_x(Z)=l}}\normm{L_{Z}}{cb}
\nonumber\\
&\le \dfrac{c}{4}\abs{A\cup C}\sum_{r\ge R/2+1}F_{\alpha-D}(r)\,(\abs{\partial_{A}}+\abs{\partial_{C}})\,\kappa(1+r)^{D-1}\sum_{l\ge 1}\sum_{\substack{x\in Z\subset \La \\ d(A\cup C,Z)= r \\ d_x(Z)=l}}\normm{L_{Z}}{cb}
\nonumber\\
&\le \dfrac{c}{4}\kappa(\abs{\partial_{A}}+\abs{\partial_{C}})\abs{A\cup C}\sum_{r\ge R/2+1}F_{\alpha-2D+1}(r)\sum_{l\ge 1}\kappa \normm{\Li}{F}F_{\alpha-D+1}(l)
\nonumber\\
&\le \dfrac{c}{4}\kappa^2 \normm{\Li}{F}(\abs{\partial_{A}}+\abs{\partial_{C}})\abs{A\cup C}\sum_{r\ge R/2+1}F_{\alpha-2D+1}(r)\,\dfrac{1}{\alpha-D}\underbrace{F_{\alpha-D}(0)}_{=1}
\nonumber\\
&\le \dfrac{c}{4}\dfrac{1}{\alpha-D}\kappa^2 \normm{\Li}{F}(\abs{\partial_{A}}+\abs{\partial_{C}})\abs{A\cup C}\sum_{r\ge R/2+1}F_{\alpha-2D+1}(r)
\nonumber\\
&\le \dfrac{c}{4}\dfrac{1}{\alpha-D}\kappa^2 \normm{\Li}{F}(\abs{\partial_{A}}+\abs{\partial_{C}})\abs{A\cup C}\,\dfrac{1}{\alpha-2D}F_{\alpha-2D}(R/2)
\nonumber\\
&\le \dfrac{c}{4}\dfrac{1}{(\alpha-D)(\alpha-2D)}\kappa^2 \normm{\Li}{F}(\abs{\partial_{A}}+\abs{\partial_{C}})\abs{A\cup C}F_{\alpha-2D}(R/2).
\nonumber
\end{align}

    Going back to $(*)$ from Equation \eqref{eq:termsCrossingBoundary}:
    \begin{equation}
        (*):\sum_{Z'\in \Omega} \normm{L_{Z'}}{cb}\sum_{x\in A\cup C}\sum_{y\in Z'}F_\alpha(d(x,y))\leq \hat{c} F_{\alpha-2D}(R/2-1),
    \end{equation}
    where
    \begin{equation}
        \tilde{c}=\dfrac{c\cdot C_F}{4}\dfrac{(\alpha-2D +1/C_F)}{(\alpha-D )(\alpha-2D)}\kappa^2 \normm{\Li}{F} (\abs{\partial_{A}}+\abs{\partial_{C}})\abs{A\cup C}.
    \end{equation}

    We are now ready to bound the term $(4)$ of the Equation \eqref{eq:(2)}:
     
     \begin{equation}
         (4):\sum_{Z'\in \Omega} \norm{L_{Z'}(\T_s(O))}\leq\left ( \dfrac{\norm{O}}{C_F}\left ( e^{\normm{\Li}{F} C_F\cdot s}-1\right )\tilde{\tilde{c}}F_{\alpha-2D}(R/2-1)\right )
     \end{equation}
     with $
         \tilde{\tilde{c}}=\tilde{c}'+\tilde{c}
    $. Going back to the Equation \eqref{eq:(2)},
    \begin{multline}
        (2):\normm{(\TT_t-\T_t)(O)}{1}\\\leq  \dfrac{\tilde{\tilde{c}}\norm{O}}{C_F}\left ( \dfrac{e^{v\cdot t}-vt-1}{v} \right )F_{\alpha-2D}(R/2-1)ds        +\norm{O}\tilde{c}' F_{\alpha-D-1}(R/2-1)\cdot t
        \label{eq:clusteringTime}
     \end{multline}

     We are ready to bound Equation \eqref{eq:diferenceStates}:
    \begin{align}
|(&\sigma-\pi)(O)|\\
&\nonumber \le \norm{O}\Bigg[ p(\abs{A})e^{-\lambda t} + \dfrac{\tilde{\tilde{c}}}{C_F}\left(\dfrac{e^{vt}-vt-1}{v}\right)F_{\alpha-2D}(R/2-1) + \tilde{c}'\,F_{\alpha-D-1}(R/2-1)\,t \Bigg]
\\
&= \norm{O}\Big[ p(\abs{A})e^{-\lambda t} + c_1\,F_{\alpha-2D}(R/2-1)\,e^{vt} + c_2\,F_{\alpha-D-1}(R/2-1)\,t \Big] \nonumber\\
&\le \norm{O}\Big[ p(\abs{A})e^{-\lambda t} + \underbrace{(c_1+c_2\,c_3)}_{=:c}\,F_{\alpha-2D}(R/2-1)\,e^{vt} \Big]. \nonumber
\end{align}
     where $c_3:=\dfrac{1}{v\cdot e}=\max_{t\geq 0} t e^{-vt}$ is a constant so that $t\leq c_3 e^{vt}$. The other constants are
     \begin{equation}
         c_1:= \dfrac{\tilde{\tilde{c}}}{C_F \cdot v}=\kappa \dfrac{\normm{\Li}{F}}{C_Fv(\alpha-D)}\abs{A\cup C}+\dfrac{c\normm{\Li}{F}}{4vC_F}\dfrac{(\alpha-2D +1)}{(\alpha-D )(\alpha-2D)}\kappa^2 (\abs{\partial_{A}}+\abs{\partial_{C}})\abs{A\cup C}
     \end{equation}
     \begin{equation}
         c_2:= \tilde{c}'-\dfrac{\tilde{\tilde{c}}}{C_F}
     \end{equation}

    Now, using our freedom to choose any $t>0$, we look for a $t=t(R)$ so that\[p(\abs{A}) e^{-\lambda t}=\mathcal{C}v F_{\alpha-2D}(R/2-1) \cdot e^{vt}.\] Explicitly, we get that
    \begin{equation}
        t(R)=\log \left(  (c^{-1}p(\abs{A})F_{\alpha-2D}(R/2-1)^{-1})^{1/(\lambda+v)}  \right),
    \end{equation}
    and
    \begin{equation}
        c\cdot F_{\alpha-2D}(R/2-1) \cdot e^{vt}=c^{\lambda/(\lambda+v)}\cdot p(\abs{A})^{v/(\lambda +v)} F_{\alpha '}(R/2-1)
    \end{equation}
    with a long-range decay given by $\alpha'=\dfrac{\lambda(\alpha-2D)}{\lambda +v}$.
    We  conclude the proof by observing that:
    \begin{equation}
         \abs{(\sigma-\pi)(O)}\leq \norm{O} \underbrace{ 2^{1+\alpha'}c^{\lambda/(\lambda+v)} p^{v/(\lambda +v)} }_{=: k}F_{\alpha '}(R-1)\,.
    \end{equation}
\end{proof}


\subsubsection{Mutual Information} \label{sec:MI}

Finally, we are ready to prove the main result of this section: local rapid mixing implies decay of mutual information on the fixed point for long-range open quantum systems. The proof is a generalization of the proof in \cite{kastoryano_rapid_2013} for finite-range systems to long-range systems without assuming neither detailed balance nor a Log-Sobolev constant \cite{KasTem13}. Long-range systems have the technical difficulty that the interactions extend to all the space, so the boundary terms are not localized in a compact support; this is solved thanks to Theorem \ref{thm:clusteringBoundary}.

\medskip

\begin{thm}[Mutual information decay]\label{thm:MutualInfoDecay}
    Let $(\Gamma,d)$ be a countable metric space equipped with the F-function\[F_\alpha(r):=\dfrac{1}{(1+r)^\alpha},\]for $r\geq 0$ and $\alpha>3D-1$. Let $\Li=\{L_Z\}_{Z\in \mathcal{P}_0(\Gamma)}$ be a local and primitive (Definition \ref{def:primitivity}), regular (Definition \ref{def:regular}) dissipative interaction with $\normm{\Li}{F}<\infty$, satisfying local rapid mixing. Fix $A,C\in \mathcal{P}_0(\Gamma)$ with $A\cap C=\emptyset$. 

    Let $\sigma$ be the unique fixed point of $\Li$ with convergence governed by $e^{-\lambda t}$ for $\lambda >0$. Then
    \begin{equation}
        I_\sigma(A:C)\leq k(\abs{A},\abs{C}) \cdot F_{\alpha '}(R-1),
    \end{equation}
     where $\alpha'=\dfrac{\lambda(\alpha-2D)}{\lambda +v}$, $v=\normm{\Li}{F}\cdot  C_F$ the Lieb-Robinson speed, $k$ is a polynomial function, and $R=d(A,C)$.
\end{thm}
\begin{proof}

    As in Theorem \ref{thm:clusteringBoundary}, let $Q:=-\sum\limits_{\substack{Z\subset \Lambda \\ Z\cap \partial_{AC}\neq \emptyset}}L_Z\in CB_0(\partial_{AC})$ be a perturbation constructed by removing the elements of $\Li$ that cross the boundary $\partial_{AC}$; and $\LL:=\Li+ Q $ the perturbed Lindbladian. Since $\Li$ is regular (see Definition \ref{def:regular}), $\LL$ is primitive; let $\pi$ be the fixed point of $\LL$. We denote by $\sigma_{AC}$ to $\tr_{(AC)^c}(\sigma)$ to the reduced state in $AC$, the mutual information is quantifying the difference (by the relative entropy) between the state $\sigma_{AC}$ and the tensor product $\sigma_A\otimes \sigma_C$:
    \begin{multline}
        S(\sigma_{AC} ||\sigma_A\otimes\sigma_C) =- S(\sigma_{AC})+S(\sigma_A)+S(\sigma_C)\\
        \leq -S(\sigma_{AC})-\text{tr}(\sigma_A \log\pi_A)-\text{tr}(\sigma_C \log \pi_C) =S(\sigma_{AC}||\pi_A\otimes\pi_C),
    \end{multline}
    where in the second line we used that $S(\sigma_{A}\otimes \sigma_C||\pi_{A}\otimes \pi_C)\geq 0$. Note also that \[S(\sigma ||\sigma_A\otimes\sigma_C)={I_\sigma(A:C)},\]is the mutual information between subsystems $A$ and $C$. Using the operator inequality\footnote{Since $\pi$ is a fixed point of $\tilde{\Li}$ that is composed of the local terms of a primitive Lindbladian $\Li$, then, $\pi_A\otimes\pi_C >0$.} $S(\rho ||\sigma)\leq \log \norm{\sigma^{-1}}\cdot \normm{\rho-\sigma}{1}$, we obtain the following bound for the mutual information:
    \begin{multline}
        I_\sigma (A:C)=S(\sigma_{AC} ||\sigma_A\otimes\sigma_C)       \leq S(\sigma_{AC} ||\pi_A\otimes\pi_C)
        \leq \log(\norm{\sigma_{AC}^{-1}}) \normm{\sigma_{AC}-\pi_A\otimes \pi_C}{1}\\
        = \log(\norm{\sigma_{AC}^{-1}}) \cdot \sup_{O\in \A_{AC}} \abs{(\sigma-\pi)(O)} \leq \underbrace{\log(\norm{\sigma_{AC}^{-1}})}_{\poly(\abs{A},\abs{C})} \cdot k \cdot F_{\alpha '}(R-1),
        \label{eq:proofDecayMI}
    \end{multline}
     where $\alpha'=\dfrac{\lambda(\alpha-2D)}{\lambda +v}$, $v=\normm{\Li}{F}\cdot  C_F$ the Lieb-Robinson speed, $k=k(\abs{A},\abs{C},\abs{\partial_{A}},\abs{\partial_{C}})$ is a polynomial function, and $R=d(A,C)$. We have used the Theorem \ref{thm:clusteringBoundary} for the last inequality. This proves the claim.
\end{proof}

A natural question is whether Theorem \ref{thm:MutualInfoDecay} can be recovered directly from the decay of correlations of fixed points of rapidly-mixing Lindbladians, such as the ones obtained in \cite{cubitt_stability_2015,roon_quasi-locality_2024}. One can, in fact do this, but obtaining exponential overheads that do not appear in Theorem \ref{thm:MutualInfoDecay}.
In \cite{roon_quasi-locality_2024}, they establish bounds of the form\[\abs{\pi(AC)-\pi(A)\pi(C)}\leq f(R)\]with a polynomially decaying function $f$, and for observables $O_A\in\A_A$ and $O_C\in\A_C$. This is done through a combination of the global rapid mixing condition and the dissipative Lieb-Robinson bound, and does not rely on the analysis of perturbations considered here. 

One can prove the decay of MI from Lemma by bounding the term
\begin{equation}
    \normm{\sigma_{AC}-\sigma_A\otimes\sigma_C}{1}=\sup_{O\in\A_{AC}}\dfrac{\abs{(\sigma_{AC}-\sigma_A\otimes\sigma_C)(O)}}{\norm{O}},
\end{equation}
decomposing a general $O\in\A_{AC}$ in a tensor product basis $\{E_i\otimes F_j\}$ and applying decay of correlations to each of the basis term.
However, the number of non-vanishing terms in the decomposition,\[O=\summ{i=1}^{\text{dim}(\A_A)}\summ{j=1}^{\text{dim}(\A_C)}c_{ij}E_i\otimes F_j,\]grows exponentially with $\abs{A},\abs{C}$. More precisely, it grows as $d^{2(\abs{A}+\abs{C})}$ where $d$ is the local dimension. Through a tighter bound \cite{Correa_2022}, one can also reduce to just a prefactor $d^{\abs{A}}$.
This pre-factor carries through in a bound of the mutual information, and so the dependence on the size of the subsystem is exponential.
Instead, Theorem \ref{thm:MutualInfoDecay} has a polynomial dependence with the sizes of $A$ and $C$. In the language of the introduction, Theorem \ref{thm:MutualInfoDecay} thus proves a \emph{global} decay of the MI, instead of a local or pairwise. 


\section{Gibbs states of Long-Range Hamiltonians}\label{sec:thermal}

We now specifically target thermal states of (non-)commuting Hamiltonians. There are Lindbladians that have them as their exact fixed point (see \cite{chen_efficient_2023,Ding_2025}) but these are not strictly $k$-local. Instead, their jump operators are quasi-local, with tails controlled by Lieb-Robinson light-cones \cite{rouze_efficient_2025}, so that even under rapid mixing (which can be proven in some regimes \cite{rouze_optimal_2024,bakshi_dobrushin_2025,smíd2026rapidmixingquantumgibbs}) we cannot directly apply the results from Section \ref{sec:long-range}.

However, it is still possible to prove a Markov property for Gibbs states directly, using the specific structure of their generating KMS Lindbladians. This is done by considering the recovery map constructed in \cite{chen_quantum_2025} through the CKG Lindbladian from \cite{chen_efficient_2023}. Here, we extend this proof to thermal states of long-range Hamiltonians,  which are the fixed point of certain long-range quasi-local Lindbladians. 


\subsection{CKG Lindbladian}\label{subsec:CKGLindbladian}
We now introduce the CKG generator from \cite{chen_efficient_2023}. We set a $\beta >0$ and a non-commuting long-range Hamiltonian $H$, as described above, over a $D$-dimensional lattice of spin systems $\Lambda=\mathbb{Z}^D$, and a set of jumps $\{\vA_a\}_a$ containing their adjoints each with norm $\norm{\A_a}\leq 1$. Then, the CKG Lindbladian is defined as
\begin{multline}
    \Li^{\beta,B}(\rho ):=-i [B,\rho]\\
    +\sum_{i\in\Lambda,\alpha\in[3]}\int_{-\infty}^\infty\gamma(\omega) \left ( \vA_{i,\alpha}(\w)\rho \vA_{i,\alpha}(\w)^\dagger-\dfrac{1}{2}\left \{\vA_{i,\alpha}(\w)^\dagger \vA_{i,\alpha}(\w),\rho\right \}   \right)    d\omega    \equiv \sum_{i\in\Lambda}\Li_i^\beta(\rho)
\label{eq:LindbladianChen}
\end{multline}
with the ``Metropolis`` weight\footnote{There are other possible weights that can be considered, such as the Gaussian transition weight which sometimes simplifies computations \cite{chen_efficient_2023,Ding_2025}.} $\gamma(\omega)$ defined as
\begin{equation}
    \gamma(\omega)=\exp \left ( -\beta \max \left ( \w +\dfrac{\beta \hat\sigma^2}{2},0\right ) \right ),
    \label{eq:metropolisWeight}
\end{equation}
with an energy width $\hat{\sigma}>0$. Most of the times, for simplicity, we omit the $\beta $ subscript of $\Li^\beta$ and write $\Li$ if the temperature is known by the context.
The Fourier transform of the jump operators is defined as
\begin{equation}
    \vA_{i,\alpha}(\w):=\dfrac{1}{\sqrt{2\pi}}\int_{-\infty}^\infty e^{i H t}\vA_{i,\alpha}e^{-iHt}e^{-i\w t}f(t) dt
    \label{eq:defA(w)}
\end{equation}
with a filter function $f(t):=\text{exp}(-t^2/\beta^2)\sqrt{\beta^{-1}\sqrt{2/\pi}}$, with a decay at large $t$ that ensures the quasi-locality of $\vA_{i,\alpha}(\w)$. 

In most of the cases, the jump operators are chosen to be the single-qubit Pauli jumps: $\vA_{i,x}=X^i,\vA_{i,y}=Y^i,\vA_{i,z}=Z^i$, where $\{X^i,Y^i, Z^i\}$ denotes the Pauli matrices $\{X,Y,Z\}$ acting on the $i$-th site\footnote{Meaning that $X^i=\Id_{\{i\}^C}\otimes X^i$.}. Moreover, we define the set of single-qubit Pauli jumps for the region $A\subset \La$ as
\begin{equation}
    P_A^1:=\{X^i,Y^i,Z^i\}_{i\in A},
    \label{eq:pauliJumpsInA}
\end{equation}
to be the set of all single-site Pauli operators in $A$.

The localized CKG Lindbladian in a region $A\subset \La$ is defined as
\begin{equation}
    \Li_A:= \summ{a \in P_A^1}\Li_a,
    \label{eq:localizationLindbladian}
\end{equation}
where each $\Li_a$ is the Lindbladian corresponding to each single-qubit Pauli jump $\vA_a$ with Metropolis weight. The above way of denoting the localized CKG should not be confused with the notation for strictly local Lindbladians from Equation \eqref{eq:localInteraction}; in Equation \eqref{eq:localizationLindbladian}, each $\Li_A$ is not strictly supported in $A$ due to the quasi-locality of the jump operators.

The Hamiltonian term is given by
\begin{equation}
    B=\sum_{i\in\Lambda,\alpha\in[3]}\int_{-\infty}^\infty b_1(t) e^{-i\beta Ht}\left ( \int_{-\infty}^\infty b_2(t')e^{i\beta Ht'}\vA_{i,\alpha}e^{-2i\beta Ht'}\vA_{i,\alpha}e^{i\beta Ht'}dt' \right ) e^{i\beta Ht}dt
\end{equation}
where $b_1,b_2$ are also fast-decaying functions with $\normm{b_1}{1},\normm{b_2}{1}\leq 1$ defined as
\begin{equation}
    b_1(t):= 2\sqrt{\pi}e^{1/8}\left (  \dfrac{1}{\text{cosh}(2\pi t)}\star_t \sin(-t)e^{-2t^2}  \right ),
\end{equation}
\begin{equation}
    b_2(t):=\dfrac{1}{2\pi} \sqrt{\dfrac{1}{\pi}} \exp (-4t^2-2it).
\end{equation}

This Lindbladian has the Gibbs state $\rho_\beta$ as a fixed point, and is quasi-local in the sense that it can be approximated by a strictly local Lindbladian with an error that decays polynomially as shown in Lemma \ref{lem:quasiLocality}.


\subsection{Recovery map}

For simplicity, in this section we refer to $\mathcal L^{\beta H}$, the CKG Lindbladian, simply by $\mathcal{L}$.
 We follow the proof of \cite{chen_quantum_2025}, where they crucially use a recovery map defined as the temporal average of a dissipative evolution given by a Lindbladian $\Li$:
\begin{equation}
   \mathcal{R}^{t}[\cdot]:=\dfrac{1}{t}\int_0^t e^{s\Li^*}[\cdot ]ds, 
\end{equation}
and the localized\footnote{Recall that $\Li_A$ is defined in Equation \eqref{eq:localizationLindbladian}.} recovery map in the region $A\subset \La$:
\begin{equation}
   \mathcal{R}^{t}_A[\cdot]:=\dfrac{1}{t}\int_0^t e^{s\Li_A^*}[\cdot ]ds.
\end{equation}
In Theorem III.$1$ \cite{chen_quantum_2025}, they prove that there are positive numbers $r, \mu > 0$ and $0 < \lambda < 1$ depending only on $\beta, k$ such that
\begin{equation}
    \bigl\| \mathcal{R}^{t}_A[\rho_{\beta,-A}] - \rho_\beta \bigr\|_1
    \le r \, e^{\mu |A|} \, t^{-\lambda}.
\end{equation}

The objective of this section is to generalize this theorem to long-range Hamiltonians.

There are two main technical differences when it comes to considering a long-range Hamiltonian instead of finite-range systems.
\begin{enumerate}
    \item \textit{Change of the temperature threshold of the imaginary-time Lieb-Robinson bound}: The temperature threshold $\beta_{max}$ from which the term\[\norm{e^{-\beta H}Oe^{\beta H}}<\infty,\quad \beta <\beta_{max};\]converges will change with respect to the finite-range case and will depend on the locality of the Hamiltonian. See Lemma \ref{lem:imaginary_time_lr}.
    \item \textit{Quasi-locality of the jump operators:} The CKG Lindbladian that has as the fixed-point the Gibbs state of a long-range Hamiltonian is polynomially quasi-local; in other words, the jump operators can be approximated by local operators with a polynomially decaying error. See Lemma \ref{lem:quasiLocality}. 
\end{enumerate}

The imaginary-time Lieb-Robinson bounds for long-range Hamiltonians are controlled, meaning that they are finite, for high enough temperatures. The following Lemma provides a temperature threshold by controlling the norm of nested commutators seeing them as elements of a polymer inspired in \cite{sanchez-segovia_high-temperature_2025}.

\medskip
\begin{lem}[Imaginary-time Lieb-Robinson bound]
\label{lem:imaginary_time_lr}
Let $H = \sum_{Z \subset \Lambda} h_Z$ be a $k$-local long-range Hamiltonian
on a $(\kappa, D)$-regular lattice. Define:
\begin{align}
    E &:= \sup_{i \in \Lambda}\sum_{\substack{Z \subset \Lambda \\ Z \ni i}}
    \|h_Z\|, \label{eq:local_energyMain}\\
    u &:= \sup_{i \in \Lambda} \sum_{j \in \Lambda} F_\alpha(d(i,j))
    \leq \|F_\alpha\| < \infty, \label{eq:u_defMain}\\
    g &:= \max_{Z} \frac{\|h_Z\|}{E} \leq 1. \label{eq:g_defMain}
\end{align}
Here $E$ is the local energy density (maximum single-site energy contribution),
$u = \|F_\alpha\|$ is the F-function norm from Definition~\ref{def:Ffunction},
and $g \in (0,1]$ is a normalized coupling constant.

For any $O_Z \in \A_Z$ and $x \in \mathbb{C}$ with
$\mathrm{Re}(x) > 0$ and $|x| < \frac{1}{8guk}$:
\begin{equation}
    \left\|e^{xH} O_Z\, e^{-xH}\right\| \leq \frac{\|O\|\cdot \abs{Z}}{1 - 8gEu|x|}.
    \label{eq:imag_lrMain}
\end{equation}
The temperature threshold is 
\begin{equation}
    \beta_{\max} := \frac{1}{8guk}.
    \label{eq:temperatureThreshold}
\end{equation}
\end{lem}

\begin{proof}
    Proof in Appendix, see Lemma \ref{lem:imaginary_time_lr_appendix}.
\end{proof}

The temperature threshold $\beta_{max}$ from Equation \eqref{eq:temperatureThreshold} is the temperature from which the imaginary-time Lieb-Robinson bound converges. 
This convergence is necessary to bound the product of two operators in the so-called weighted-norm defined as
\begin{equation}
    \normm{\cdot }{\rho}:=\normm{\rho^{1/4}\cdot \rho^{1/4}}{2},
\end{equation}
using the $2$-norm. More specifically, \cite{chen_quantum_2025} proved that for any operators normalized by $\|O\| \le 1$, $\|A\| \le 1$, and any pair of inverse temperatures 
$\beta_0, \beta$ such that $\beta > 4\beta_0 > 0$,
\begin{equation}
    \| AO \|_{\rho_\beta}, \, \| OA \|_{\rho_\beta}
    \lesssim 
    \left(
        \frac{ e^{\sigma^2 \beta'^2} }{ \beta' \sigma }
        + \frac{ e^{\sigma^2 \beta_0^2} }{ \beta_0 \sigma }
    \right)
    \| O \|_{\rho_\beta}^{\frac{4\beta_0}{\beta}}
    \Bigl(
        \| \rho_{\beta_0} A \rho_{\beta_0}^{-1} \|
        + \| \rho_{\beta_0}^{-1} A \rho_{\beta_0} \|
    \Bigr),
\end{equation}
where $\beta' := \beta/4 - \beta_0$. In this way, the product of two operators in the weighted-norm at low temperatures can be bounded with the imaginary-time Lieb-Robinson bound at a lower inverse temperature $\beta_0:=\beta_{max}/2$, in a way that it converges. This can be done after bounding
\begin{equation}
    \| \vA(\omega) \| 
    \le 
    \frac{ e^{ -\beta \omega + \sigma^2 \beta^2 } }{ \sqrt{ \sigma \sqrt{2\pi} } } 
    \bigl\| e^{\beta H} \vA e^{ -\beta H } \bigr\|.
\end{equation}

For our purpose, the choice of $\beta_0$ is determined by Lemma \ref{lem:imaginary_time_lr}
as $\beta_0=\beta_{max}/2=\dfrac{1}{16 guk},$ in a way that
\begin{equation}
\normm{e^{xH}O_Z e^{-xH}}{\rho_{\beta_0}}=\normm{O_Z}{}\abs{Z}.
\label{eq:imaginaryTimeLiebRobinsonInBeta0}
\end{equation}

We omit the proof of the following Theorem since the only thing that changes with respect to the Theorem III.1 \cite{chen_quantum_2025} is the temperature threshold $\beta_{max}$ from Equation \eqref{eq:imaginaryTimeLiebRobinsonInBeta0}, the rest of the proof does not depend on the range of the Hamiltonian.

\medskip
\begin{thm}[Generalization of Theorem III.1 \cite{chen_quantum_2025}: Quasi-local recovery maps via time-averaged Gibbs sampling]\label{thm:quasi-Local-RecoveryMap-longRange}
Consider the Gibbs state of a $k$-local long-range Hamiltonian $H$ and a region $A \subset \Lambda$. Then, the time-averaged Lindblad dynamics 
$\mathcal{R}_{A,t}$ with single-qubit Pauli jumps $\{\vA_\alpha\}_{\alpha \in P_A^1}$ (see Equation \eqref{eq:pauliJumpsInA})
Metropolis weight $\gamma(\omega)$ (see Equation \eqref{eq:metropolisWeight}) and $\sigma = 1/\beta$, $t > 0$ gives 
an approximate recovery map at all temperatures $\beta$:
\begin{equation}
    \bigl\| \mathcal{R}^{t}_A[\rho_{\beta,-A}] - \rho_\beta \bigr\|_1
    \le |A|^22^{2|A|}
    \begin{cases}
        r(\beta, k)\, t^{-\frac{128 \beta_0^4}{\beta^3(\beta + 5\beta_0)}}, & \text{if } \beta > 4\beta_0, \\[0.6em]
        r'(\beta, k)\, t^{-\frac{2\beta_0}{\beta + 5\beta_0}}, & \text{if } \beta \le 4\beta_0,
    \end{cases}
\end{equation}
for some explicit functions $r(\beta,k)$ and $r'(\beta,k)$, where $\beta_0 := 1/16guk$. 
Therefore, there are numbers $r, \mu > 0$ and $0 < \lambda < 1$ depending only on $\beta, d$ such that
\begin{equation}
    \bigl\| \mathcal{R}^{t}_A[\rho_{\beta,-A}] - \rho_\beta \bigr\|_1 
    \le r\, e^{\mu |A|}\, t^{-\lambda},
    \label{eq:recoveryMap}
\end{equation}
for every $\beta >0$.
\end{thm}

\subsection{Quasi-locality of the recovery-map}

Although Equation \eqref{eq:recoveryMap} provides a recovery map for every temperature, it is not strictly localized in the region $A$, but it is a quasi-local recovery map. In order to be able to bound the CMI, we need a recovery map that is localized in $AB$. This section computes the error made by approximating $\mathcal{R}_A^t$ by a strictly local map $\RR_{A,l}^t$ supported on $A(l)$; in particular, when considering $l=R=d(A,C)$, then note that $A(l)=A(R)=AB$, which is what is needed for bounding the CMI.

Firstly, we will relate the locality of the recovery map with the locality of the Lindbladian. For that purpose, we need to adapt the Lieb-Robinson bound for long-range Hamiltonians from a bound on the commutator to a bound on truncating the Hamiltonian.

\medskip
\begin{thm}[Truncation bound under the Hastings--Koma Lieb--Robinson estimate]\label{thm:LRboundLongRange}
Let $H = \summ{Z\subset \La} h_Z$ be a long-range Hamiltonian on $\La \subset \subset \mathbb{Z}^D$ such that Theorem \ref{th:HKLR} holds, and fix a region $X\subset \La$. 

Then, for every operator $O_X\in\A_X$ supported on $X$, if the interaction rate satisfies $\alpha>\eta>d$, 
then there exists a constant 
\begin{equation}
\big\|T_{H,t}(O_X) - T_{H_R,t}(O_X)\big\|
\ \leq\  c\,
\|O_X\|\, \abs{X}\,(\mathrm{e}^{v|t|}-1-v|t|)\,F_{\alpha-D}(d(X,Y)).
\label{eq:difTimeEvolutionsRMain}
\end{equation}

\end{thm}
\begin{proof}
    See Theorem \ref{thm:LRboundLongRangeAppendix} in the appendix.
\end{proof}

This directly yields the quasi-locality property of the CKG Lindbladian:
\medskip
\begin{lem}[Quasi-locality]\label{lem:quasiLocality}
For a $k$-local Hamiltonian $H$ with long-range interactions and a region $A\subset \La$, denote by 
$\mathcal{L}_A^{(R)}$ the generator of the Gibbs sampler with 
Metropolis weight $\gamma_M$ (see Equation \ref{eq:metropolisWeight}), single-qubit Pauli jumps as the jump operators $\vA_\alpha \in P_A^1$ (see Equation \eqref{eq:pauliJumpsInA}) and Hamiltonian 
\begin{equation}
    H_R:=\summ{Z\subset B_R(A)}h_Z.
\end{equation}

Then,
\begin{equation}
    \|\mathcal L_A - \mathcal L_A^{(R)}\|_{\infty\to\infty} = \hat{\mathcal{O}}(R^{-(\alpha-D)}).
\end{equation}
\end{lem}
\begin{proof}
The proof is divided into two parts, each bounding the decay of one specific part of the Lindbladian (namely, the coherent and the incoherent part). Consider the CKG Lindbladian
\begin{equation}
\label{eq:L-CKG}
\mathcal L_A[\rho]
=
\underbrace{-i[B_A,\rho]}_{\text{coherent part}}
+
\underbrace{\sum_{\alpha\in P^1_A}
\int_{-\infty}^{\infty}\!\gamma_\alpha(\omega)
\left(
\vA_\alpha(\omega)\,\rho\,\vA_\alpha(\omega)^\dagger
-\tfrac12\!\left\{\vA_\alpha(\omega)^\dagger\vA_\alpha(\omega),\rho\right\}
\right)\,d\omega}_{\text{incoherent part}},
\end{equation}
where each Fourier component is given by
\begin{equation}
\label{eq:Ahat}
\vA_\alpha(\omega)
=
\frac{1}{\sqrt{2\pi}}
\int_{-\infty}^{\infty}
e^{+iHt}\vA_\alpha e^{-iHt}\,e^{-i\omega t}\,f_\alpha(t)\,dt,
\qquad 
\| \vA_\alpha\|\le 1,
\end{equation}
for integrable kernels\footnote{Note that $e^{v\abs{t}}-1-v\abs{t}\leq (v\abs{t})^2e^{v\abs{t}}/2$ for every $t$ by Taylor.} $f_\alpha\in L^1(\mathbb R,(1+t^2)e^{v|t|}dt)$ 
and weights $\gamma_\alpha\in L^1(\mathbb R)$.

From now on, for simplicity, we will assume that every $\alpha$ is a jump operator in $P_A^1$.

We define the truncated operators
\[
\vA^{(R)}_\alpha(\omega)
=
\frac{1}{\sqrt{2\pi}}
\int_{-\infty}^{\infty}
e^{+iH_Rt}\vA_\alpha e^{-iH_Rt}\,e^{-i\omega t}\,f_\alpha(t)\,dt,
\]\[
\mathcal L_A^{(R)}[\rho]
:=
-i[B^{(R)},\rho]
+\sum_{\alpha\in P_A^1}\!\int\!\gamma_\alpha(\omega)
\mathcal \,D[\vA^{(R)}_\alpha(\omega)](\rho)\,d\omega,
\]
where $\mathcal D[L](X)=LXL^\dagger-\tfrac12\{L^\dagger L,X\}$.

For later convenience, for any $O_X\in\A_X$ supported in $X$, apply the Lieb-Robinson bound for long-range Hamiltonians given by Theorem \ref{thm:LRboundLongRange}
\[
\|T_{H,t}(O_X)-T_{H_R,t}(O_X)\|
\le
\frac{C}{v}\big(\mathrm{e}^{v|t|}-1-v|t|\big)
\,F_{\alpha-D}(R)\,\|O_X\|,
\]
and insert the previous estimate to the dissipative part of the Lindbladian to obtain
\[
\|\vA_\alpha(\omega)-\vA^{(R)}_\alpha(\omega)\|
\le
\frac{C}{\sqrt{2\pi}v}\,
F_{\alpha-D}(R)
\int_{-\infty}^{\infty}\!|f_\alpha(t)|\,
(\mathrm{e}^{v|t|}-1-v|t|)\,dt=:\Xi_\alpha(R)
\sim R^{-(\alpha-D)}<\infty.
\]
The quasi-locality of the jump operators is key to prove the quasi-locality of the incoherent part and the coherent part.

\underline{Incoherent part:} For the incoherent part, we have that for any operators $L,\widetilde L$ and $X$, the map
$\mathcal D[L]$
satisfies $(\mathcal{D}L - \mathcal{D}\tilde{L})(X) = (L-\tilde{L})X L^\dagger + \tilde{L} X (L-\tilde{L})^\dagger - \frac{1}{2}\{L^\dagger L - \tilde{L}^\dagger \tilde{L}, X\}$. Using $\|L^\dagger L - \tilde{L}^\dagger \tilde{L}\| \leq (\|L\| + \|\tilde{L}\|)\|L - \tilde{L}\|$. Then,
\[
\|\mathcal D[L]-\mathcal D[\widetilde L]\|
\le
3(\|L\|+\|\widetilde L\|)\|L-\widetilde L\|.
\]
Setting $L=\vA_\alpha(\omega)$ and $\widetilde L=\vA^{(R)}_\alpha(\omega)$, then,
\begin{equation}
\Big\|\int_{-\infty}^{\infty}\!\gamma_\alpha(\omega)
\big(\mathcal D[L]-\mathcal D[\widetilde L]\big)d\omega\Big\|\leq 3\,\Gamma_\alpha\,\left (\norm{\vA_\alpha(\w)}+\norm{\vA^{(R)}_\alpha(\w)}\right )\,\norm{\vA_\alpha(\w)-\vA^{(R)}_\alpha(\w)}.
\label{eq:boundIncoherentDifference}
\end{equation}

The norm of $\vA_\alpha(\w)$ and $\vA^{(R)}_\alpha(\w)$ can be bounded by a constant independent of $\w$.
\begin{equation}
    \norm{\vA_\alpha(\w)}\leq \norm{\vA_\alpha(t)}\underbrace{\frac{1}{\sqrt{2\pi}}\int_{-\infty}^{\infty}\!|f_\alpha(t)|\,dt}_{=:M_\alpha}\leq \norm{\vA_\alpha}M_\alpha\leq M_\alpha,
\end{equation}
(same with $\norm{\vA_\alpha^{(R)}(\w)}$) with the following constants:
\begin{equation}
\Gamma_\alpha :=\int_{-\infty}^{\infty}\!|\gamma_\alpha(\omega)|\,d\omega,
\qquad
M_\alpha:=\frac{1}{\sqrt{2\pi}}\int_{-\infty}^{\infty}\!|f_\alpha(t)|\,dt,
\end{equation}

Going back to Equation \eqref{eq:boundIncoherentDifference} and integrating
against $|\gamma_\alpha(\omega)|$ yields
\begin{equation}
\Big\|\int_{-\infty}^{\infty}\!\gamma_\alpha(\omega)
\big(\mathcal D[L]-\mathcal D[\widetilde L]\big)d\omega\Big\| \le
3\,\Gamma_\alpha\,(M_\alpha+\Xi_\alpha(R))\,\Xi_\alpha(R).
\end{equation}

\medskip
\underline{Coherent part:} We will follow the proof of the Corollary A.$2$ from \cite{chen_quantum_2025}. We aim at bounding the coherent part of the CKG Lindbladian with the Metropolis weight (Equation \eqref{eq:metropolisWeight}). The main differences with respect to finite or short-range comes from the Lieb-Robinson bounds used, while  the part of the proof corresponding to the manipulations of $B_a$ is unchanged with respect to~\cite{chen_quantum_2025}.

The coherent part is generated by the operator\[B_A=\sum_{a\in A}B_a,\]where the operator $B_a$ associated to the jump operator $\vA_a$ with Metropolis weight is
\begin{equation}
B_{a}
:=
\int_{-\infty}^{\infty}
b_1(t)\,
e^{-i\beta H t}
\lim_{\eta\to 0^+}
\left(
\int_{-\infty}^{\infty}
b_2^{\eta}(t')\,
\vA_{a}(\beta t')\,
\vA_{a}(-\beta t')\,dt'
+
\frac{1}{8\sqrt{2\pi}}\,
\vA_{a}^{\dagger}\vA_{a}
\right)
e^{i\beta H t}\,dt.
\end{equation}
with the functions
\begin{align}
b_1(t)
&:=
2\sqrt{\pi}\,e^{\frac18}
\left(
\frac{1}{\cosh(2\pi t)}
*
\sin(-t)\,e^{-2t^2}
\right),
\qquad
\|b_1\|_1 < 1,
\\[1ex]
b_{2}^{\eta}(t)
&:=
\mathbf{1}(|t|\ge \eta)\,
\underbrace{\frac{1}{2\sqrt{2\pi}}\,
\frac{\exp(-2t^2-it)}{t(2t+i)}}_{=b_2^M(t)}
\equiv
\mathbf{1}(|t|\ge \eta)\,
b_2^{M}(t).
\end{align}

    The goal is to show that $B_a$ is quasi-local, meaning that replacing $H$ by the Hamiltonian $H_R$ (containing only terms inside $A(R)$) only introduces a small error.

    The first step consists of truncating the double integral to $\abs{t},\abs{t'}\leq T$. Because $b_1\in L^1 $ ($\normm{b_1}{1}\leq 1$) with exponential decay and $b_2^M(t)\lesssim e^{-2t^2}/\abs{t}$ for large $\abs{t}$, both integrals are absolutely convergent outside any finite window. In particular, truncating the integral to $\abs{t}\leq T$ introduces an error
    \begin{equation}
        \normm{b_1(t) \Id(\abs{t}\geq T)}{1}\lesssim e^{-cT},
    \end{equation}
    and truncating the inner integral to $\abs{t'}\leq T$ introduces an error controlled by
    \begin{equation}
        \normm{b_2^M(t)\Id(\abs{t'}\geq T)}{1}\lesssim e^{-2cT^2}.
    \end{equation}
    Both terms are exponentially small in $T$. Defining $B^T$ as the double truncated version of $B$, we have that it can be locally approximated with an error
    \begin{equation}
        \normm{B_A-B_A^{(R)}}{}\leq \normm{B_A-B_A^T}{}-\norm{B_A^T-(B_A^{(R)})^T}\lesssim e^{-cT}+e^{-2cT^2}+\norm{B_A^T-(B_A^{(R)})^T},
    \end{equation}
    so the proof reduces to bounding the truncated difference $\normm{B_A^T-(B_A^{(R)})^T}{}$. For simplicity we may denote $B\equiv B_A$.

The truncated $B^T$ splits as
\begin{align}
B^T:= &\int_{-T}^{T} b_1(t) \left (
\lim_{\eta \to 0} \big(
\underbrace{\int_{-T}^{T} 
 \Id(|t| \ge 1) 
b_2^{\eta}(t') \vA^{\dagger}( \beta(t' - t)) \vA(-\beta(t' + t)) \, dt'}_{(a)}
\right. \\ \nonumber
&+\underbrace{\int_{-T}^{T} 
\Id(1 \ge |t| \ge \eta) 
b_2^{\eta}(t') \vA^{\dagger}( \beta(t' - t)) \vA(-\beta(t' + t)) \, dt'}_{(b)}\big)
\left.  +\underbrace{\frac{(\vA^{\dagger} \vA)(-\beta t)}{16 \sqrt{2\pi}}}_{(c)} \right ) dt.
    \end{align}

Each of the terms requires a different treatment.

The term $(c)$ involves the operator $\vA^\dagger \vA$ evolved in time, and the difference\[\normm{(\vA^\dagger \vA)_H(t)-(\vA^\dagger \vA)_{H_R}(t)}{} ,\]
is directly bounded by the Lieb-Robinson bound:
\begin{equation}
    (c)\lesssim \norm{\vA_X}^2 F_{\alpha-D}(R)(e^{vt}-1-vt).
\end{equation}

The term $(a)$ contributes the same order as term $(c)$. For $\abs{t'}\geq 1$ the kernel $b_2^M(t')$ is bounded and decays as $e^{-2t^2}$, so there is no singularity. The difference between $H$ and $H_R$ enters only through the time evolution $\vA_H(\beta t')\mapsto \vA_{H_R}(\beta t') $. Applying Lieb-Robinson bound for long-range interactions and using $\normm{b_1}{1},\normm{b_2^M(\abs{t'}\geq 1)}{1}\lesssim  1$, one obtains
 \begin{equation}
     (a)\lesssim \norm{\vA_X}^2 F_{\alpha-D}(R)(e^{2vt}-1-2vt).
 \end{equation}

The heart of the argument is controlling the term $(b)$. Due to the divergence of $b_2^M(t')\sim 1/t'$ at $t'=0$, one cannot simply bound $\normm{b_2^M(t')}{}\cdot \normm{\vA_H(\beta t')-\vA_{H_R}(\beta t')}{}$ and integrate.

The solution is the Corollary A.1 from \cite{chen_quantum_2025} which bounds the coherent part by commutators of the jump operatos with the Hamiltonian:

\begin{multline}
    \big( \Id (1 \ge |t| \ge \eta) \, \text{term} \big)\lesssim 
\beta \sup_{|t| \le \beta (T + 1)} 
\| [\vA, H]_H(t) - [\vA, H]_{H_R}(t) \|\\
\nonumber+ 
\sup_{|t| \le \beta (T + 1)} 
\| \vA_H(t) - \vA_{H_R}(t) \|.
\end{multline}

Now the term $(b)$ is controlled by the Lieb-Robinson bounds of the commutator $[\vA,H]$ and the jump operator $\vA$:
\begin{multline}
    ( c )
\lesssim 
\beta \sup_{|t| \le \beta (T + 1)} 
\| [\vA, H]_H(t) - [\vA, H]_{H_R}(t) \|
+ 
\sup_{|t| \le \beta (T + 1)} 
\| \vA_H(t) - \vA_{H_R}(t) \|,
\end{multline}
Using the Lieb-Robinson bounds for long-range interactions from \cite{kuwahara_strictly_2020}, and applying the same steps as in the proof of Theorem \ref{thm:LRboundLongRange}, the following is satisfied for $\alpha>2D+1$:\[\norm{\vA_H(t)-\vA_{H_R}(t)}\lesssim \dfrac{t^{D+1}}{(l-v t)^{\alpha-D}}.\]

Finally, considering the exponentially small error made by truncating the integral of the coherent part to a time $T$ (dependent on $R$) and using the quasi-locality of $B^T$, one ends up with the quasi-locality of the coherent part $B$
\begin{align}
    \norm{B-B^{(R)}}&\leq \norm{B-B^T}+\norm{B^T-B^{(R)}}\\
    &\nonumber\lesssim \normm{b_1(t)\Id(\abs{t}\geq T)}{1}+\normm{b_2^M(t)\Id (\abs{t}\geq T)}{1}+\abs{A}\dfrac{(\beta T)^{D+1}}{(l-v (\beta T))^{\alpha-D}}\\
    &\nonumber\lesssim e^{-cT}+e^{-2T^2}+\abs{A}\dfrac{(\beta T)^{D+1}}{(l-v (\beta T))^{\alpha-D}}\lesssim \dfrac{\log(R)^{D+1)}}{R^{\alpha -D}}=\hat{\mathcal{O}}(R^{-(\alpha-D)}),
\end{align}
by choosing\[T(R)=\dfrac{\alpha-D}{c}\log(R)-\dfrac{D+1}{c}\log \log (R)+\Theta(1).\]Thus, we get a polynomial decay of the coherent part with a logarithmic correction. 
    \begin{equation}
        \norm{B-B^{(R)}}\leq \hat{\mathcal{O}}(R^{-(\alpha-D)})
    \end{equation}
    for $\alpha >2D+1$.

  \underline{Quasi-locality of the Lindbladian:}    The original proof concludes by combining both the coherent and incoherent parts to give:
\begin{multline}
    \normm{\Li_{A,R}^\dagger-\Li_{A}^\dagger}{\infty \to \infty} \leq \normm{\mathcal{D}^\dagger - \mathcal{D}_R^\dagger}{\infty \to \infty }+2 \normm{B_A-B_A^{(R)}}{\infty\to\infty}
    \\
    \lesssim\,\abs{A}\left (\Gamma_\alpha\,(M_\alpha+\Xi_\alpha(R))\,\Xi_\alpha(R) + \dfrac{\log(R)^{D+1)}}{R^{\alpha -D}}\right )
    \\
    \lesssim \abs{A}\left ( \dfrac{\log(R)^{D+1}}{R^{\alpha -D}}\right )=\hat{\mathcal{O}}(R^{-(\alpha-D)}).
\end{multline}
This concludes the proof.
\end{proof}


The quasi-locality of the recovery map follows from the quasi-locality of the CKG Lindbladian, which was proven in Lemma \ref{lem:quasiLocality}. Indeed:
\medskip
\begin{lem}[Lemma VII.$2$ \cite{chen_quantum_2025}: Truncation error]\label{lem:truncationError}
Consider the time-averaging map $\mathcal{R}^{\dagger}_{A,t,R}$ associated with the Hamiltonian 
$H_R$ containing all jumps $A^{\alpha}$ with $\alpha \in P_{A_R}^1$. Then,
\begin{equation}
    \bigl\| (\mathcal{R}^t_{A,t,R} )^\dagger-  (\mathcal{R}^t_{A,t})^\dagger \bigr\|_{\infty \to \infty}
    \lesssim
    t \, \bigl\| \mathcal{L}^{\dagger}_{A,R} - \mathcal{L}^{\dagger}_{A} \bigr\|_{\infty \to \infty},
\end{equation}
where $\mathcal{L}_A$ (resp. $\mathcal{L}_{A,R}$) is the Lindbladian of the Gibbs sampler 
with Hamiltonian $H$ and jumps $A^{\alpha} \in P_A^1$ (respectively the one associated to the 
Hamiltonian $H_R$ with jumps $A^{\alpha} \in P_{A_R}^1$).
\end{lem}

Now that we have proven that the recovery map recovers $\rho_\beta$ from $\p_{\beta,-A}$ (see Theorem \ref{thm:quasi-Local-RecoveryMap-longRange}) and that it is quasi-local (see Lemma \ref{lem:truncationError}), we have all the ingredients to prove the decay of the CMI:
\medskip
\begin{thm}\label{thm:cmiCKGLindbladian}
    Let $H$ be a $k$-local long-range Hamiltonian and a tripartition $A\sqcup B\sqcup C=\La$ as in Figure \ref{fig:geometryCMI}. Then,
    the CMI decays as
    \begin{equation}
        I(A : C \,|\, B)_\rho \le \log(\dim(C)) \,e^{\mu ''\abs{A}} \dfrac{(\log (R))^{\lambda'(D+1)/2}}{R^{\lambda ' (\alpha-D)/2}} =\tilde{\mathcal{O}}(R^{-\alpha'}),
    \end{equation}
    for $\alpha':=\lambda ' (\alpha-D)/2$, where $\lambda':=\lambda/(1+\lambda)$ and $\lambda$ is defined in Equation \eqref{eq:lambdaValues} (note that $\lambda$ depends on the inverse temperature $\beta$).
\end{thm}

\begin{proof}
    The proof is straightforward once there is a quantum channel $\RR_{A,t}$ that recovers $\rho_\beta$ from $\rho_{\beta,-A}$ (see Theorem \ref{thm:quasi-Local-RecoveryMap-longRange}). In order for $\RR_{A,t}$ to be a recovery map, it needs to be localized in a region $A(R)=A\cup B$. Lemma \ref{lem:truncationError} proves that the quasi-locality of the recovery map is linearly related to the quasi-locality of the Lindbladian used in the definition of the recovery map. Then, using both Lemma \ref{lem:truncationError} and \ref{lem:quasiLocality}, the error made by localizing $\RR_{A,t}$ is
\begin{equation}
    \bigl\| (\mathcal{R}^t_{A,R})^* -  (\mathcal{R}^t_{A})^* \bigr\|_{\infty \to \infty}
    \lesssim
    t \, \abs{A}\,\dfrac{(\log (R))^{D+1}}{R^{\alpha-D}}.
\end{equation}
By choosing a time
\begin{equation}
    t_0(R)=\dfrac{e^{\mu \abs{A}/(1+\lambda)}R^{\beta/(1+\lambda)}}{(\log (R))^{(D+1)/(1+\lambda)}},
    \label{eq:timeRecovery}
\end{equation}
and using the triangular inequality, we get that the recovery map recovers the original state with a polynomial precision with logarithmic corrections, i.e.
\begin{equation}
     \bigl\| \mathcal{R}^{t_0}_{A,R}[\rho_{\beta,-A}] - \rho_\beta \bigr\|_1 \lesssim e^{\mu '\abs{A}} \dfrac{(\log (R))^{\lambda'(D+1)}}{R^{\lambda ' (\alpha-D)}},
     \label{eq:localRecovery}
\end{equation}
where $\lambda'=\lambda/(1+\lambda)\in (0,1)$, and $\lambda$ defined as:
\begin{equation}\label{eq:lambdaValues}
    \lambda =
    \begin{cases}
        \frac{128\beta_0^4}{\beta^3(\beta + 5\beta_0)}, & \text{if } \beta > 4\beta_0,\\[0.6em]
        \frac{2\beta_0}{\beta + 5\beta_0}, & \text{if } \beta \le 4\beta_0,
    \end{cases}
\end{equation}

Finally, the CMI can be bounded by the approximate recovery map
\begin{equation}
    I(A : C \,|\, B)_\rho 
\le 
\Delta \log(\dim(C)) + h_2(\Delta)
\lesssim 
\log(\dim(C)) \sqrt{\Delta},
\end{equation}
where $h_2(x) = -x \log_2(x) - (1 - x) \log_2(1 - x)
\lesssim \sqrt{x} \text{ for } x \in [0,1]$,
and\[
\quad
\Delta := \tfrac{1}{2} \bigl\| \mathcal{R}_{A,R}^t[\rho_{\beta,-A}] - \rho_\beta \bigr\|_1 .
\]
This proves the claim.
\end{proof}


\subsection{Numerical results}
\label{sec:numerics}
In this section, we simulate the long-range classical $1$D Ising model and study the CMI decay of its Gibbs state at positive temperatures. In Section \ref{sec:LR-TFIM} we add a transverse field to the model and study in which regimes of $h$ and decay rate $\alpha$ the decay of the CMI holds with the bounds proven above.

\subsubsection{Classical $1$D spin chains}\label{sec:classical}

\begin{figure}[t]
     \centering
     \includegraphics[width=0.8\linewidth]{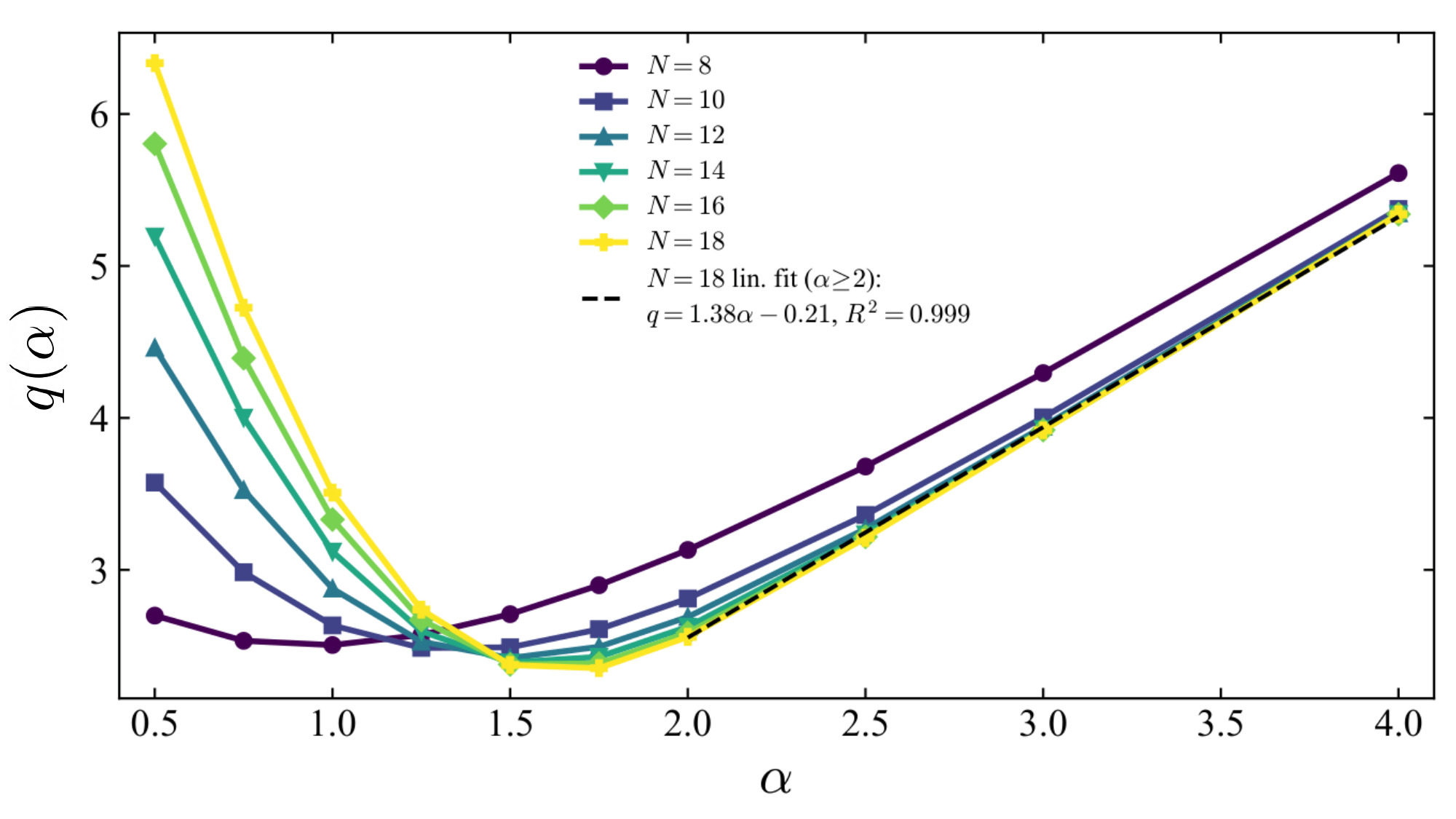}
     \caption{Decay exponent $q(\alpha)$ of the conditional mutual information at $\beta=5$ for different system sizes $N$. The dashed line indicates the linear fit for $N=18$ ($\alpha \geq 2$).}
     \label{fig:1}
\end{figure}

While our results above provide upper bounds, one might wonder whether there even are explicit examples of states where the CMI decays exactly as a polynomial of the size of the buffer region.
As a simple case study that shows such a decay we consider the $1$D Ising model with long-range interactions,
\begin{equation}
H = -\sum_{i<j}\frac{J_0}{|i-j|^\alpha}\, s_i s_j ,
\qquad s_i=\pm1.
\end{equation} We focus on the Gibbs distribution
$
p(\mathbf{s})=\frac{1}{Z}e^{-\beta H(\mathbf{s})},
$
which defines a classical thermal state. This model admits a quantum (commuting) counterpart,
\begin{equation}
H_Q=-\sum_{i<j}\frac{J_0}{|i-j|^\alpha}\,Z_iZ_j,
\end{equation}
whose Gibbs state is diagonal in the computational basis and therefore identical to the classical distribution above.

For each interaction exponent $\alpha$ and system size $N$, we compute numerically the conditional mutual information $I(A:C|B_\ell)$ as a function of the buffer size $\ell$. We find that the decay is well described by a power law,
\begin{equation}
I(A:C|B_\ell)\sim \ell^{-q(\alpha)},
\end{equation}
which allows us to define an effective decay exponent $q(\alpha)$ through a log--log regression.

In Fig.~\ref{fig:1} we show the decay exponent $q(\alpha)$ for system sizes up to $N=18$, together with a weighted linear regression. The data are well described by
\begin{equation}
q(\alpha)=a\,\alpha+b,
\end{equation}
with $
a= 1.38, \, b= -0.21$,
indicating an approximately linear growth of the CMI decay exponent with the interaction exponent.
Notably, the fact that $a>1$ implies that the CMI decays faster than the strength of the interactions.

In the Figure ~\ref{fig:1} we show $q(\alpha)$ for several system sizes $N=8,10,12,14,15,16,18$. While finite-size effects are visible for smaller $N$, the curves rapidly converge as $N$ increases, indicating that the observed scaling is not a finite-size artifact.

\subsubsection{Long-Range Transverse-Field Ising Model}
\label{sec:LR-TFIM}

A natural question is whether this decay survives in the
presence of quantum fluctuations.  We now address this
question numerically for the one-dimensional long-range
transverse-field Ising model (LR-TFIM). The Hamiltonian is
\begin{equation}
  H = -\sum_{i < j} \frac{J}{|i-j|^{\alpha}}\,Z_i Z_j
      - h \sum_{i} X_i ,
  \label{eq:H}
\end{equation}
where $\alpha > 1$ controls the interaction range and $h \geq 0$
is the transverse field. 

We compute the CMI by exact diagonalization, for system sizes up to $N=10$. The reason for the small sizes is that computing these quantities requires the calculation of the entire spectrum of $\rho_\beta$; we leave it for future work to develop methods to reach larger sizes.

Figure \ref{fig:cmiDecayCurves} shows the CMI decay profile evolves with both $\alpha$ and $h$. Despite the limited number of qubits and the possibility of finite-size effects, we observe the following:
\begin{itemize}
    \item For small $h$, one recovers the polynomial decay of the CMI found in Section \ref{sec:classical}.
    \item For high enough $\alpha$ the long-range interactions are so weak  that the system resembles a short- or even a finite-range one, suggesting an exponential decay of CMI. At $N=10$ the maximum buffer size is $R_\mathrm{max} = 8$, which does not allow a quantitative distinction between power-law and exponential decay; the fits only serve as qualitative guides.
    \item For $\alpha=3$, there is a change from a polynomial decay of the CMI between $h=0.3$ and $h=5$. This suggests that there might be a critical value of $h$ for which the CMI changes from a polynomial to an exponential decay, but the number of qubits are insufficient to conclude this.
\end{itemize}

\begin{figure}
    \centering
    \includegraphics[width=\linewidth]{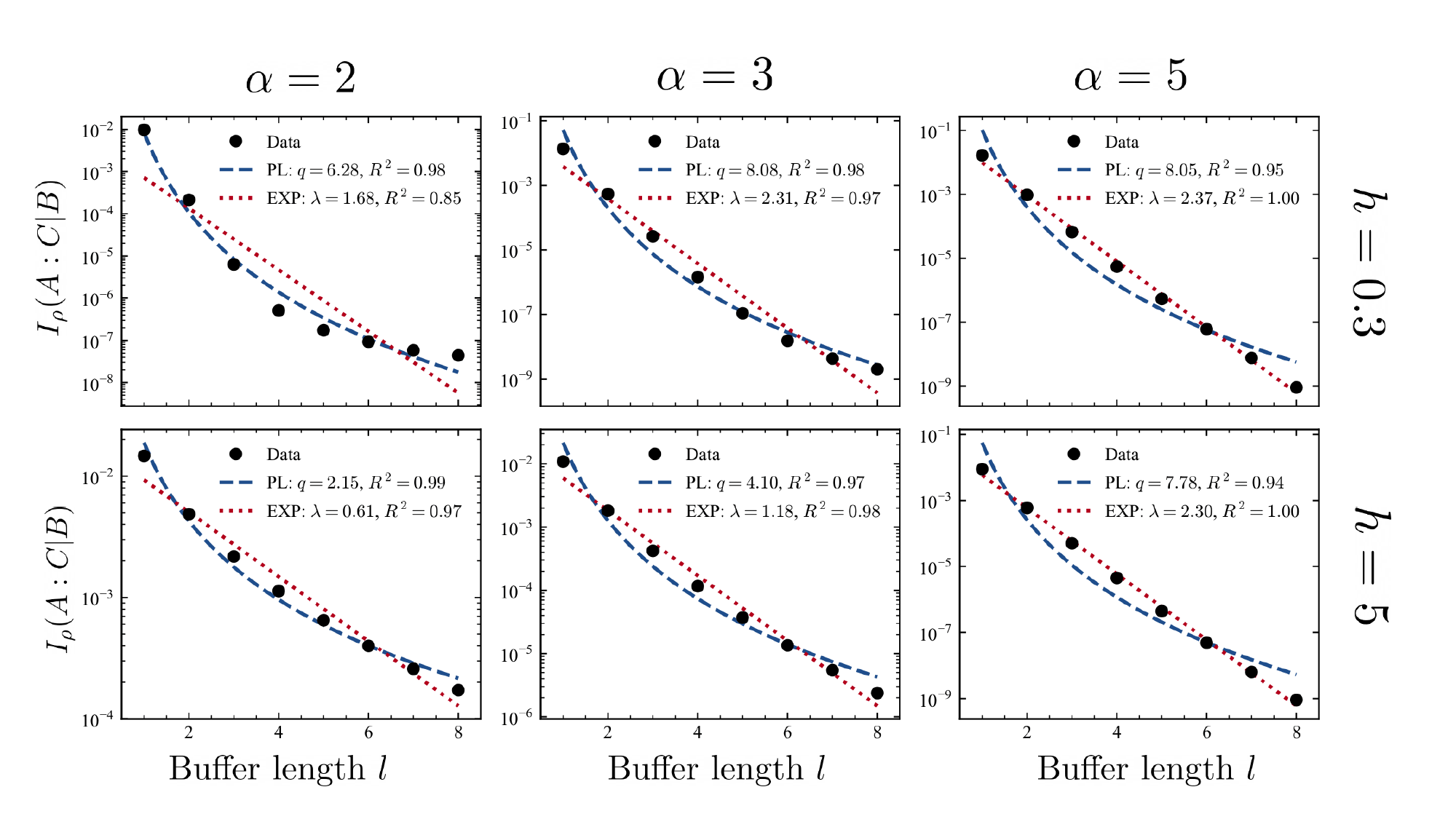}
    \caption{CMI decay $I(A:C|B_\ell)$ as a function of buffer size $\ell$ for selected values of $\alpha$ and transverse field $h$, at inverse temperature $\beta = 5$, $N=10$. Power-law (PL) and exponential (EXP) fits are shown; $R^2$ values quantify the relative quality of each ansatz. The running effective exponent $q_\mathrm{eff}(\ell)$ is shown in the legend.}
    \label{fig:cmiDecayCurves}
\end{figure}

For $T=0$, the ground state of the LR-TFIM has a phase transition for a critical $h_c=h_c(\alpha)$ dependent on $\alpha$, \cite{shiratani_stochastic_2024,koziol_quantum-critical_2021}. This further motivates the study of whether there is a critical value of $h$ for which there is a qualitative change in the behaviour of the CMI even for large $N$, both at $T\to 0$ and also at higher temperatures. We leave this question for future work.

\section{Conclusions}

We have established the decay of mutual information for fixed-point of long-range (local) Lindbladians under rapid mixing. If frustration-freeness is also satisfied by the Lindbladian, then the fixed point shows global Markov property, meaning that the conditional mutual information $I(A:C|B)$ decays polynomially with distance and scales polynomially with the sizes $\abs{A},\,\abs{C}$.

Additionally, we prove a local Markov property of the Gibbs state of a non-commuting long-range Hamiltonian. The proof is an extension of the proof from \cite{chen_quantum_2025} to long-range Hamiltonians. We use the CKG Lindbladian \cite{chen_efficient_2023} which is a quasi-local thermal Lindbladian whose fixed-point is a Gibbs state of a non-commuting Hamiltonian.

\paragraph{Tightness.}
The decay exponent $\alpha'=\dfrac{\lambda \alpha}{2(\lambda +v)}-\dfrac{nD}{2}$ in Theorem \ref{thm:decayCMIgeneralLindbladian} for the CMI
is not expected to be tight in general. A linear dependence on $\alpha$ is expected from the
classical case (Section~\ref{sec:classical}),
which for short-range systems
($\alpha\to\infty$) is consistent with
exponential decay.
The rate
 observed numerically in Section~\ref{sec:classical} (namely, $q(\alpha)\approx 1.38\alpha$) is larger than the theoretical prediction $q(\alpha)\simeq\lambda\alpha/2(\lambda+v)$.
We leave the question of the optimal decay exponent depending on parameters such as $D,v$ as an open problem.

\paragraph{Open problems.}
Several natural directions remain open:
\begin{enumerate}
\item \emph{Sharp threshold for rapid mixing.}
The condition $\alpha>3D+1$ for CKG rapid mixing is set by the method of
\cite{rouze_efficient_2025} and is likely not tight, due to overheads related to the Lieb-Robinson bound and other proof steps. In 1D specifically, the classical long-range Ising
model thermalizes rapidly for all $\alpha>1$. Closing this gap
would immediately improve the thermal CMI result, both in the decay with distance and in the dependence with the size of $\vert A \vert, \vert C \vert$. A similar situation occurs for the MI in Theorem \ref{thm:MutualInfoDecay}, where we believe that a decay rate of just $\alpha$ is possible \cite{kim_thermal_2025}.

\item \emph{Connection to classification of thermal phases.}
The work of \cite{MKS25} shows that CMI decay implies the existence of
low-depth quantum circuits that prepare Gibbs states within the same fixed thermal phase.
Our results establish the prerequisites for extending this classification to
long-range systems, although the polynomial tails appear too large for a straightforward extension of that result. An explicit circuit construction remains to be done.


\item \emph{Connection to Strong-to-Weak Spontaneous Symmetry Breaking (SW-SSB)}: SW-SSB is characterized by the coexistence of decaying ordinary correlations with a non-decaying information-theoretic diagnostic \cite{Lessa_2025}. Typically formulated as a non-decaying CMI \[I(A:C|B)\xrightarrow[d(A,C) \to \infty]{} c >0, \]and a decaying correlator function\[|\text{Cov}(O_A,O_C)|\xrightarrow[d(A,C)\to\infty]{} 0.\]. 

Our results constrain both sides of this diagnostic simultaneously. Theorem \ref{thm:MutualInfoDecay} shows that, under local rapid mixing and primitivity alone (no frustration-freeness required), the mutual information (and hence the ordinary correlations between $A$ and $C$) decays for any $\alpha>3D-1$. 

 On the other hand, Theorem \ref{thm:decayCMIgeneralLindbladian} shows that the CMI decays unconditionally in the short-range case (i.e. whenever $F$ decays exponentially), and Corollary \ref{cor:decayCMILongRange} extends this to long-range Lindbladians whenever $\alpha>\alpha^*:=n(\lambda+v)D/\lambda$

Hence, for $\alpha>\max \{\alpha^*,3D-1\}$ (and unconditionally in the short-range case) rapid mixing forbids SW-SSB, since neither the ordinary correlations nor the CMI survive at long distance. For $\alpha\leq \alpha^*$, it remains open whether the CMI could fail to decay while the ordinary correlations still do, which would allow SW-SSB to occur even under rapid mixing.

\item \emph{Removing frustration-freeness.}
Theorem~\ref{thm:decayCMIgeneralLindbladian} assumes frustration-freeness (FF).
Whether FF can be replaced by a weaker condition (e.g. detailed balance
or else) is an important open problem even in the short-range setting.

\item \textit{Larger numerical simulations.} Numerically, we confirm the polynomial decay of the CMI for the classical long-range Ising chain, consistent with the above bounds, and find a similar polynomial regime in the long-range transverse field Ising model at a small transverse field $h$ (Section~\ref{sec:LR-TFIM}). As $h$ or $\alpha$ increase, the decay instead appears to cross over to an exponential one. Clarifying whether this is a genuine effect or a finite-size effect, is left for future work.

\end{enumerate}

\section{Acknowledgements}
AMA acknowledges useful discussions with Shengqi Sang and Ruochen Ma. The authors would like to thank David P\'erez-Garc\'ia for pointing them to the proofs of \cite{brandao_area_2015}. The authors would like to thank to Tim Möbus for pointing out an error in the Fawzi-Renner bound in Proposition \ref{prop:FR}. PR, MS and AMA acknowledge support from the Spanish Agencia Estatal de Investigacion through the grants ``IFT Centro de Excelencia Severo Ochoa CEX2020-001007-S", ``PID2023-150847NA-I00'', ``EUR2025-164823'' and  ``PCI2024-153448'', funded by MICIU/AEI/10.13039/501100011033 and co-funded by the EU. AC acknowledges the support of the Deutsche Forschungsgemeinschaft (DFG, German Research Foundation) - Project-ID 470903074 - TRR 352.  This
project was funded within the QuantERA II Programme that has received funding from the EU’s H2020 research and innovation programme under the GA No 101017733.

\printbibliography

\appendix


\section{Auxiliary results}\label{app:auxiliary}

We now prove two technical lemmas that will allow us to prove some properties of open quantum systems with long-range interactions on $D$-regular graphs.

\medskip
\begin{lem}\label{lem:technicalLemmaFfunction}
    Let $\Gamma$ be a $D$-regular graph and let $\Li=\sum_{Z\in\mathcal{P}_0(\Gamma)}L_Z$ be a local interaction governed by the $F$-function $F_\alpha(r) :=1/(1+r)^\alpha;\,r\geq0$. Fix an $l>0$, then
    \begin{equation}
        \sup_{x\in\G}\sum\limits_{\substack{Z\in\mathcal{P}_0(\Gamma)\\ Z\ni x  \\ {\rm diam}(Z)= l}}\norm{L_Z}\leq \kappa \normm{\Li}{F}\,F_{\alpha-D}(l/2).
    \end{equation}
\end{lem}
\begin{proof}For any $x\in \G$,
\begin{equation}
\begin{aligned}
\sum_{\substack{
    Z\in\mathcal{P}_0(\Gamma),\, Z\ni x \\
    \mathrm{diam}(Z)= l
}}
\norm{L_Z}
&\leq
\sum_{n=0}^{l/2}
\underbrace{
    \sum_{\substack{
        y\in\G \\
        d(x,y)=l/2+n
    }}
}_{=\kappa (1+l/2+n)^{D-1}}
\underbrace{
    \sum_{\substack{
        Z\in\mathcal{P}_0(\Gamma) \\
        Z\ni x,y
    }}
    \norm{L_Z}
}_{\;\;\;\leq \normm{\Li}{F} F_\alpha (l/2+n)}
\\
&\leq
\kappa \normm{\Li}{F}
\sum_{n=0}^{l/2}
(1+l/2+n)^{D-1}
F_\alpha (l/2+n)
\\
&=
\kappa \normm{\Li}{F}
\sum_{n=0}^{l/2}
F_{\alpha -D+1}(l/2+n)
\\
&\leq
\kappa \normm{\Li}{F}
(1+l/2)
F_{\alpha-D+1}(l/2)
\\
&=
\kappa \normm{\Li}{F}
F_{\alpha -D}(l/2),
\end{aligned}
\end{equation}
where the last inequality follows from upper-bounding $F(l/2+n)\leq (l/2)$. Taking the superior gives the claim.
\end{proof}
\begin{lem}\label{lem:technicalLemmaFfunction2}
    Let $\Gamma$ be a $D$-regular graph and let $\Li=\sum_{Z\in\mathcal{P}_0(\Gamma)}L_Z$ be a local interaction governed by the $F$-function $F_\alpha(r) :=1/(1+r)^\alpha;\,r\geq0$. Fix an $l>0$, then
    \begin{equation}
    \sup_{x\in\G}\sum\limits_{\substack{Z\in\mathcal{P}_0(\Gamma)\\ Z\ni x  \\ d_x(Z)= l}}\norm{L_Z}\leq \kappa \normm{\Li}{F} F_{\alpha -D+1}(l)
    \end{equation}
\end{lem}
\begin{proof}For any $x\in \G$,
\begin{align}
\sum_{\substack{
    Z\in\mathcal{P}_0(\Gamma),\, Z\ni x \\
    d_x(Z)= l
}}
\norm{L_Z}
&\leq
\underbrace{
    \sum_{\substack{
        y\in\G \\
        d(x,y)=l
    }}
}_{=\kappa (1+l)^{D-1}}
\underbrace{
    \sum_{\substack{
        Z\in\mathcal{P}_0(\Gamma) \\
        Z\ni x,y
    }}
    \norm{L_Z}
}_{\leq \normm{\Li}{F} F_\alpha (l)}
\\
&\leq
\kappa \normm{\Li}{F}
(1+l)^{D-1}
F_\alpha (l)
\nonumber\\
&=
\kappa \normm{\Li}{F}
F_{\alpha -D+1}(l).
\nonumber
\end{align}
Taking the superior gives the claim.
\end{proof}

\subsection{Imaginary-time Lieb-Robinson bounds}
\begin{lem}[Lemma \ref{lem:imaginary_time_lr}]
\label{lem:imaginary_time_lr_appendix}
Let $H = \sum_{Z \subset \Lambda} h_Z$ be a $k$-local long-range Hamiltonian 
on a $(\kappa, D)$-regular lattice. Define:
\begin{align}
    E &:= \sup_{i \in \Lambda}\sum_{\substack{Z \subset \Lambda \\ Z \ni i}} 
    \|h_Z\|, \label{eq:local_energy}\\
    u &:= \sup_{i \in \Lambda} \sum_{j \in \Lambda} F_\alpha(d(i,j)) 
    \leq \|F_\alpha\| < \infty, \label{eq:u_def}\\
    g &:= \max_{Z} \frac{\|h_Z\|}{E} \leq 1. \label{eq:g_def}
\end{align}
Here $E$ is the local energy density (maximum single-site energy contribution), 
$u = \|F_\alpha\|$ is the F-function norm from Definition~\ref{def:Ffunction}, 
and $g \in (0,1]$ is a normalized coupling constant.

For any $O_Z \in \A_Z$ and $x \in \mathbb{C}$ with 
$\mathrm{Re}(x) > 0$ and $|x| < \frac{1}{8guk}$:
\begin{equation}
    \left\|e^{xH} O_Z\, e^{-xH}\right\| \leq \frac{\|O\|\cdot \abs{Z}}{1 - 8gEu|x|}.
    \label{eq:imag_lr}
\end{equation}
The temperature threshold is $\beta_{\max} := \frac{1}{8guk}$.
\end{lem}

\begin{proof}
    Expanding the exponential function in a power series ($H$ is bounded),
    \begin{equation}
        \norm{e^{xH}O_Z e^{-xH}}\leq \sum_{n=0}^\infty \dfrac{x^n}{n!}\norm{\text{ad}_H^n(O_Z)}
    \end{equation}
    One can bound each term $\text{ad}_H^n(v_Z)$ as the product of the norms of the elements of a polymer, as done in Appendix A.1 of \cite{sanchez-segovia_high-temperature_2025}. 
    \begin{multline}
        \norm{\text{ad}_H^n(O_Z)}=\| [\underbrace{H,\dots ,[H}_{n},O_Z]]\|=\|\summ{Z_1,\dots,Z_n}[h_{Z_1},\dots,[h_{Z_n},O_Z]]\| \\
        \leq \norm{v_Z}2^n \ n! \ \abs{Z} \ \summ{\gamma \nsim Z\\\norm{\gamma}=n}\prod_{ Z_i \in \gamma}\norm{h_{Z_i}}\leq  \norm{O_Z} \abs{Z}(8g u  k)^n,
    \end{multline}
    for which we have used the bound obtained in \cite{sanchez-segovia_high-temperature_2025}. 
    Then,
    \begin{equation}
         \norm{e^{xH}O_Z e^{-xH}}\leq \sum_{n=0}^\infty \dfrac{\abs{x}^n}{n!}\norm{\text{ad}_H^n(O_Z)}\leq   \norm{O_Z} \abs{Z} \sum_{n=0}^\infty (8guk \abs{x})^n=\dfrac{\norm{O_Z}\abs{Z}}{1-8guk \abs{x}}, 
    \end{equation}
    if $\abs{x}<\dfrac{1}{8guk}$.
\end{proof}

\subsection{Lieb-Robinson bounds}
\medskip
\begin{thm}[Truncation bound under the Hastings--Koma Lieb--Robinson estimate]\label{thm:LRboundLongRangeAppendix}
Let $H = \sum_Z h_Z$ be a Hamiltonian on a lattice of bounded degree in dimension $d$ such that Theorem \ref{th:HKLR} hold, 
and fix a finite region $X$. 

Then, for every operator $A_X$ supported on $X$, if the interaction satisfies a power--law decay
$\|h_Z\| \lesssim (1+\mathrm{diam}(Z))^{-\alpha}$ with $\alpha>\eta>d$, 
then there exists a constant $C_{d,\alpha,\eta}>0$ such that

\begin{equation}
\big\|T_{H,t}(A_X) - T_{H_R,t}(A_X)\big\|
\ \leq\ \tilde c
\|A_X\|\, \abs{X}\,(\mathrm{e}^{v|t|}-1-v|t|)\,F_{\alpha-D}(d(X,Y)).
\label{eq:difTimeEvolutionsR}
\end{equation}
where $\tilde c$ depends only on $c,v$ in Equation ~\eqref{eq:Hastings}. 

\end{thm}

\begin{proof}
Since both $H$ and $H_R$ are self-adjoint (finite-volume lattice, bounded operators), the maps $t\mapsto T_{H,t}(A_X)$ and $t\mapsto T_{H_R,t}(A_X)$ are differentiable in operator norm. Write
\[
T_{H,t}(A_X)-T_{H_R,t}(A_X)=\int_0^t \dfrac{d}{ds}\left [ T_{H,t}\circ T_{H_R,t-s}(A_X)\right ]ds.
\]
Computing the derivative: since \[\frac{d}{ds}T_{H,s}(B) = iT_{H,s}([H,B]),\] and
\[\frac{d}{ds}T_{H_R,t-s}(A_X) = -iT_{H_R,t-s}([H_R, A_X]),\]then,
\begin{align}
\frac{d}{ds}\left[T_{H,s}\!\left(T_{H_R,t-s}(A_X)\right)\right]
= iT_{H,s}\!\left([H,T_{H_R,t-s}(A_X)]\right)
- iT_{H,s}\!\left(T_{H_R,t-s}([H_R, A_X])\right).
\end{align}
Since $T_{H,s}$ is a $*$-automorphism, this simplifies as
\begin{align*}
[H,T_{H_R,t-s}(A_X)] &- T_{H_R,t-s}([H_R,A_X])\\
&= [H - H_R,\, T_{H_R,t-s}(A_X)]
+ [H_R,\, T_{H_R,t-s}(A_X)]
- T_{H_R,t-s}([H_R, A_X]).
\end{align*}
The last two terms cancel because $T_{H_R,t-s}$ is generated by $H_R$:
\begin{align}
[H_R,\, T_{H_R,t-s}(A_X)] = T_{H_R,t-s}([H_R, A_X]).
\end{align}
Hence:
\begin{align}
\frac{d}{ds}\left[T_{H,s} \circ T_{H_R,t-s}(A_X)\right]
= i\,T_{H,s}\!\left([H - H_R,\, T_{H_R,t-s}(A_X)]\right).
\end{align}

The Duhamel formula gives
\begin{equation}
T_{H,t}(A_X)-T_{H_R,t}(A_X)
= i\int_0^{t} T_{t,s}^{H}\!\big([H-H_R,\,T_{H_R,s}(A_X)]\big)\,ds.
\end{equation}
Using that $T_{t,s}^H$ contracts in operator norm,
\begin{equation}
\big\|T_{H,t}(A_X)-T_{H_R,t}(A_X)\big\|
\le \int_0^{|t|}\!\big\|[H-H_R,\,T_{H_R,s}(A_X)]\big\|\,ds.
\label{eq:duhamel}
\end{equation}
Expanding $H-H_R=\sum_{Z\not\subset B_R(X)} h_Z$, we get
\begin{multline}
\big\|[H-H_R,\,T_{H_R,s}(A_X)]\big\|
\le \underbrace{\summ{Z\subset (B_R(X))^C } \|h_Z\|\,
\big\|[T_{H_R,s}(A_X), O_Z]\big\|}_{(1)} \\
+\underbrace{\summ{Z\not\subset B_R(X)\\Z\cap X \neq \emptyset} \|h_Z\|\,
\big\|[T_{H_R,s}(A_X), O_Z]\big\|}_{(2)}
+\underbrace{\summ{Z\not\subset B_R(X)\\Z\cap X = \emptyset} \|h_Z\|\,
\big\|[T_{H_R,s}(A_X), O_Z]\big\|}_{(3)},
\end{multline}
where $\quad O_Z := \frac{h_Z}{\|h_Z\|}.$

For the first term, applying the bound in Theorem~\ref{th:HKLR} with 
$O_X=A_X$, $O_Z=O_Z$, $r=0$, and noting that $d(X,Z)\ge R$, we have
\[
\big\|[T_{H_R,s}(A_X), O_Z]\big\|
\le
\frac{c\,(e^{v s}-1)}{[1+d(X,Z)]^{\eta}}\,\|A_X\|.
\]
Then,

\begin{multline}
    (1)=\summ{Z\subset (B_R(X))^C } \|h_Z\|\,
\big\|[T_{H_R,s}(A_X), O_Z]\big\|=\summ{l=1}^\infty \sum_{r=R+1}^\infty \underbrace{\summ{y\in\La\\ d(X,y)=r}}_{\leq \kappa r^{D-1}}\underbrace{\summ{Z\ni y\\ d_y(Z)=l}\norm{h_Z}}_{\leq \kappa \normm{\Li}{F}F_{\alpha-D+1}(l)}F_\alpha(r)\\
\leq \kappa^2 \normm{\Li}{F}\underbrace{\summ{l=1}^\infty F_{\alpha-D+1}(l)}_{\leq \dfrac{1}{\alpha-D}} \underbrace{\sum_{r=R+1}^\infty F_{\alpha-D+1}(r)}_{\leq \dfrac{F_{\alpha-D}(R)}{\alpha-D}}\leq \dfrac{\kappa^2 \normm{\Li}{F}}{(\alpha-D)^2}F_{\alpha-D}(R).
\end{multline}

For the second term $(2)$, note that every element $Z$ in the sum has a diameter bigger or equal to $R$. However, Lieb-Robinson bounds cannot be used since $Z\cap X\neq \emptyset$, then, the trivial bound with the contractive property of $T^H$ is applied in $\norm{[T_{H_R,s}(A_X), O_Z]}\leq 2 \norm{A_X}$. Using that fact, one can derive the following bound
\begin{multline}
    (2)=\summ{Z\not\subset B_R(X)\\Z\cap X \neq \emptyset} \|h_Z\|\,
\big\|[T_{H_R,s}(A_X), O_Z]\big\|\leq 2 \summ{Z\not\subset B_R(X)\\Z\cap X \neq \emptyset} \|h_Z\|\,
\big\|A_X\big\|\\
\leq 2 \norm{A_X} \sum_{r=R}^\infty \summ{x,y\\x\in X\\d(X,y)=r}\summ{Z\ni x,y}\norm{h_Z}\leq 2 \kappa \abs{X}\norm{A_X}\sum_{r=R}^\infty r^{D-1} F_\alpha (r)\\
\leq \dfrac{ 2 \kappa \abs{X}\norm{A_X}}{\alpha-D}F_{\alpha-D}(R).
\end{multline}

The last term, $(3)$, we can use the Lieb-Robinson bounds, 
\begin{multline}
    (3)=\summ{Z\not\subset B_R(X)\\Z\cap X = \emptyset} \|h_Z\|\,
\big\|[T_{H_R,s}(A_X), O_Z]\big\|\leq \sum_{l=1}^\infty\sum_{r=1}^R \summ{x\in \mathsf X_r }\summ{Z\ni x,\\ d_x(Z)=R-r+l}\norm{h_Z}\underbrace{\big\|[T_{H_R,s}(A_X), O_Z]\big\|}_{\leq\frac{c\,\big(e^{v(t-r)} - 1\big)}{[1+r]^{\alpha}}\,\|A_X\|,}\\
\leq c\,\kappa\norm{A_X}\abs{X}\,\big(e^{v(t-r)} - 1\big) \sum_{l=1}^\infty\underbrace{\sum_{r=1}^R F_{\alpha-D+1}(R-r+l)F_\alpha(r)}_{\leq C_F F_{\alpha-D+1}(R+l)}\leq \\
c\,\kappa^2C_F\norm{A_X}\normm{\Li}{F}\abs{X}\,\big(e^{v(t-r)} - 1\big) \sum_{l=1}^\infty F_{\alpha-D+1}(R+l)\\
\leq \dfrac{c \,C_F\kappa^2\norm{A_X}\normm{\Li}{F}\abs{X}\,\big(e^{v(t-r)} - 1\big)}{\alpha-D}F_{\alpha-D}(R),
\end{multline}
where we have used Lemma \ref{lem:technicalLemmaFfunction2}.
Inserting into~\eqref{eq:duhamel},
\begin{align}
\big\|T_{H,t}(A_X)-T_{H_R,t}(A_X)\big\|
&\le c\, F_{\alpha-D}(R)\int_0^{|t|}\!(e^{v s}-1)\,ds.
\end{align}
Evaluating the integral gives
\[
\int_0^{|t|}(e^{v s}-1)\,ds
= \frac{1}{v}(e^{v|t|}-1) - |t|
\ =\ \frac{1}{v}\,(\mathrm{e}^{v|t|}-1-v|t|),
\]
which proves~\eqref{eq:difTimeEvolutionsR} with $\tilde c = c/v$.
\end{proof}
\end{document}